%% file: Bartoszek_HarmonicQuadraticSums.tex
\documentclass[12pt]{article}
\pdfoutput=1
\usepackage{lineno}
\usepackage{graphicx}
\usepackage{mathrsfs}
\usepackage{amssymb}
\usepackage{amsmath}
\usepackage{graphicx}
\usepackage{amsmath,amsfonts}
\usepackage{algorithm}
\usepackage[square,numbers]{natbib}
\usepackage{amscd}
\usepackage[mathscr]{eucal}
\usepackage{lineno}
\usepackage{extarrows}
\usepackage{etoolbox}
\usepackage{tikz}
\usepackage{url}
\usepackage{listings}
\usepackage{bm}
\usepackage{xr}
\usepackage{url}
\usepackage{algpseudocode}
\usepackage{lineno}

\MakeRobust{\ref}% avoid expanding it when in a textual label

\makeatletter
\newcommand{\labeltext}[2]{%
   \@bsphack
   \csname phantomsection\endcsname % in case hyperref is used
    \def\@currentlabel{#1}{\label{#2}}%
    \@esphack
}

\makeatletter
\patchcmd{\NAT@test}{\else \NAT@nm}{\else \NAT@nmfmt{\NAT@nm}}{}{}

\DeclareRobustCommand\citepos
  {\begingroup
   \let\NAT@nmfmt\NAT@posfmt% ...except with a different name format
   \NAT@swafalse\let\NAT@ctype\z@\NAT@partrue
   \@ifstar{\NAT@fulltrue\NAT@citetp}{\NAT@fullfalse\NAT@citetp}}

\let\NAT@orig@nmfmt\NAT@nmfmt
\def\NAT@posfmt#1{\NAT@orig@nmfmt{#1's}}

\makeatother

\newtheorem{lemma}{Lemma}%[section]

\newtheorem{theorem}{Theorem}[section]

\newtheorem{remark}{Remark}[section]

\newcommand{\ba}{\begin{align}} \newcommand{\ea}{\end{align}}
\newcommand{\baa}{\begin{align*}} \newcommand{\eaa}{\end{align*}}
\newcommand{\ben}{\begin{enumerate}} \newcommand{\een}{\end{enumerate}}
\newcommand{\bi}{\begin{itemize}} \newcommand{\ei}{\end{itemize}}

\newenvironment{proof}{\noindent\\ \noindent\relax{\sc
     Proof}}{{\samepage\par\nopagebreak\hbox
     to\hsize{\hfill$\Box$}}}
\newcommand{\be}{\begin{equation}} \newcommand{\ee}{\end{equation}}
\newcommand{\bd}{\begin{displaymath}} \newcommand{\ed}{\end{displaymath}}

\newcommand{\proglang}[1]{\texttt{#1}}

\title{Closed and asymptotic formul\ae\ for harmonic and quadratic harmonic sums}
\author{Krzysztof Bartoszek
 \\
\noindent {\small \it 
Department of Computer and Information Science,
} 
{\small \it
 Link\"oping University, Link\"oping, Sweden}
\\
{\small \it
e--mail: krzysztof.bartoszek@liu.se; krzbar@protonmail.ch
}
}
\date{}

\begin{document}
\maketitle
\begin{abstract}
We present here a large collection of harmonic and quadratic harmonic sums, 
that can be useful in applied questions, e.g., probabilistic ones.
We find closed--form formul\ae , that we were not able to locate in the literature.
\\
MSC2020: 
05A19; 
11M06; 
40A25 
\\
Keywords: Abel's summation formula by parts; Combinatorial series identities and summation formul\ae; Generalized harmonic numbers; Harmonic sums; Quadratic harmonic sums; Riemann Zeta function %; Telescoping method
\end{abstract}

%\section*{ToDo}
%\begin{enumerate}
%\item corretly align right lemmata and proofs
%\end{enumerate}

\section{Introduction}
We are interested in closed form formul\ae\ for quadratic harmonic sums. 
Study of such 
goes back to at least Euler, who showed \citep[we report after][]{ASof2010}, 

$$
2\sum\limits_{j=1}^{\infty}\frac{H_{j,1}}{n^{p}}=(p+2)\zeta(p+1)-\sum\limits_{r=1}^{q-2}\zeta(r+1)\zeta(q-r)
$$
where $\zeta(\cdot)$ is the Riemann zeta function, and 
$H_{j,m}$ the $j$--th, $m$--th generalized harmonic number, i.e.,

$$
H_{j,m} := \sum\limits_{i=1}^{j}i^{-m},~~\zeta(m) = \sum\limits_{i=1}^{\infty}i^{-m}.
$$ 
Finite sum formul\ae\ have been studied in the literature \citep[e.g.,][]{ASof2010,ASof2011,ASof2013,ASofMHas2012},
but sometimes only asymptotic results have been obtained \citep[e.g.,][]{DBorJBor1995art,XWanWChu2018}.
In applications, e.g., sums of moments of exponential random variables, which are harmonic numbers, one 
requires closed form finite some formul\ae\  and here we provide a collection
of closed form formul\ae\ for various quadratic harmonic sums. 
Our motivation was that in many cases we were neither able to find many of them in literature, nor 
were we able to obtain them using \proglang{Mathematica} \citep{Mathematica}. 

We first introduce notation and results from literature that will be building blocks for us. 
Then, we provide a general recursive approach 
(Thms. \ref{thmGpqrsm} and \ref{thmVpqrsm})
for calculating the types of sums that are of interest to us

$$
G_{n,p,q}^{r,s,m}:=\sum\limits_{j=1}^{n}\frac{H_{j,m}}{(j+r)^{p}(j+s)^{q}},~~~
V_{n,p,q}^{r,s,m}:=\sum\limits_{j=1}^{n}\frac{H_{j,m}^{2}}{(j+r)^{p}(j+s)^{q}}.
$$
While this approach results in forms that still need to be further refined in each particular case,
we propose a general strategy for obtaining closed form sums and also
it exposes components of the sums that will be tractable.  
Afterwords, we specialize
on particular sums that arose in our applied studies, i.e., for
$r,s,p,q,m \in \{1,2\}$
(Lemmata \ref{lemSumHi2i1sqi2}-\ref{lemSumHi2sqi1i2}), and the special case 
of $V_{n,2,0}^{0,0,2}$ (Lemma \ref{lemSumHi2sqisq}).
In the work here lemmata or equations that have an ``H'' in front of
them can be found in the Appendix. 
In the Appendix 
we provide a comprehensive collection
of harmonic sums formul\ae. Some of the considered there sums are trivial, while some can be found in the literature.
However, we report all of them, with proofs, for completeness. Where possible we provide, references for the presented
formul\ae. If a reference is missing it means that either the formula is simple enough that we could not find its
first mention, or we are not able to find the formula in the literature. In no way do we claim precedence to these 
formul\ae\ in our work here. We use a variety of techniques and it can happen that a different one could 
be more optimal, that further refinements of the formul\ae\ are possible or that sometimes in proofs
an intermediate, instead of the final, version of some lemma might be used. 
In a few cases we were able to locate the asymptotic sum, but not the finite sum formula. 
We reference those works, where the formul\ae\ of interest are directly available
or easily deducible.
We also checked the results, where possible, in \proglang{Mathematica 12.0.0}. 
We discuss this and indicate our code's availability at the end of this work.

\section{Notation and known results}\label{secNotation}
We start with introducing some notation and some general results from the literature.
We assume throughout that $n\ge 1$ is a natural number.
Some special values of 
the Riemann Zeta function, 
that we will need, 
are
\begin{enumerate}
\item $\zeta(2) = \frac{\pi^{2}}{6}$;
\item $\zeta(3) \approx 1.202$ is Ap\'ery's constant; 
\item $\zeta(4) = \frac{\pi^{4}}{90}$; 
\item $\zeta(6)= \pi^{6}/945$ (e.g., OEIS entry A013664).
\end{enumerate}
We also recall that $H_{n,1}$ behaves asymptotically as $\ln n$. 
Throughout this work, we will make use of the following relationships 
\begin{enumerate}
\item $G_{\infty,2,0}^{0,0,1}=\sum\limits_{j=1}^{\infty}\frac{H_{j,1}}{j^{2}}=2\zeta(3)$ \citep[e.g.,][after Euler, 1775]{DBorJBor1995art}; 
\item $G_{\infty,3,0}^{0,0,1}=\sum\limits_{j=1}^{\infty}\frac{H_{j,1}}{j^{3}}=\frac{\pi^{4}}{72}$ \citep[e.g., Eq. $4$ in ][after Euler, 1775]{DBorJBor1995art};
\item $G_{\infty,4,0}^{0,0,1}=\sum\limits_{j=1}^{\infty}\frac{H_{j,1}}{j^{4}}=3\zeta(5) - \frac{\pi^{2}}{6}\zeta(3)$ \citep[e.g., p $1196$ in ][after Euler, 1775]{DBorJBor1995art}; 
\item $V_{\infty,2,0}^{0,0,1}=\sum\limits_{j=1}^{\infty}\frac{H_{j,1}^{2}}{j^{2}} = \frac{17\pi^{4}}{360}$ \citep[Eq. $3$ in ][]{DBorJBor1995art};
\item $G_{\infty,4,0}^{0,0,2}=\sum\limits_{j=1}^{\infty}\frac{H_{j,2}}{j^{4}} = \zeta^{2}(3)-\frac{4\pi^{6}}{2835} + \zeta(6) =\zeta(3)^{2} - \frac{\pi^{6}}{2835} $  \citep[p. $280$][]{DBorJBorRGir1995}.
\end{enumerate}
Finally, for any sequence $a_{n}$ if $i_{1}>i_{2}$ we set, by definition, 
$\sum_{j=i_{1}}^{i_{2}}a_{j}\equiv 0$.

\section{General recursive formul\ae }
In this section we will provide a general recursive approach to calculate $G_{n,p,q}^{r,s,m}$ and 
$V_{n,p,q}^{r,s,m}$. The final results will be in terms of initial 
conditions

\be\label{eqGnVn0}
G_{n,p,0}^{0,0,m} = \sum\limits_{j=1}^{n}\frac{H_{j,m}}{j^{p}},~~~V_{n,p,0}^{0,0,m} = \sum\limits_{j=1}^{n}\frac{H_{j,m}^{2}}{j^{p}};
\ee
sums of a finite, fixed, i.e., independent of $n$ number of terms that for specific values of $p,q,w,r,s,u$ are obtainable in terms of harmonic numbers

\be
R_{n,p,q}^{r,s}:=\sum\limits_{j=1}^{n}\frac{1}{(j+r)^{p}(j+s)^{q}},~~~
R_{n,p,q,w}^{r,s,u}:=\sum\limits_{j=1}^{n}\frac{1}{(j+r)^{p}(j+s)^{q}(j+u)^{w}};
\ee
or are constant, non--dependent on $n$

\be
G_{r,p,0}^{0,0,m}=\sum\limits_{j=1}^{r}\frac{H_{j,m}}{j^{p}},~~V_{r,p,0}^{0,0,m}=\sum\limits_{j=1}^{r}\frac{H_{j,m}^{2}}{j^{p}}.
\ee
The initial conditions  
$G_{n,p,0}^{0,0,m}$, and  $V_{n,p,0}^{0,0,m}$ are in general not tractable but have in some cases had their
limits, $n\to\infty$ calculated in the literature as we presented in Section \ref{secNotation}.
Moreover, the terms $R_{n,p,q}^{r,s}$, $R_{n,p,q,w}^{r,s,u}$ are generally convergent with $n$, unless
particular values of $p$ and $q$ are chosen (which will not be the case in our calculations).
Notice that obviously $G_{n,p,0}^{0,0,m}=G_{n,0,p}^{0,0,m}$, and $V_{n,p,0}^{0,0,m}=V_{n,0,p}^{0,0,m}$.
We only consider sums that have $H_{j,m}$ and $H_{j,m}^{2}$. Higher powers can be handled in the same way.

\begin{theorem}\label{thmGpqrsm}
Define, for integers $m\ge 1$, $p,q\ge 0$ (but $p+q>1$), and $r,s\ge 0$

$$
G_{n,p,q}^{r,s,m}:=\sum\limits_{j=1}^{n}\frac{H_{j,m}}{(j+r)^{p}(j+s)^{q}}.
$$
We then we have the following recursive relationship

$$
G_{n,p,q}^{r,s,m} = \left\{
\begin{array}{cc}
G_{n+r,p,0}^{0,0,m}-G_{r,p,0}^{0,0,m} -\sum\limits_{i=1}^{r}R_{n,m,p}^{i,r},& q=0, \\
\frac{1}{(s-r)^{q}}\sum\limits_{i=0}^{q}(-1)^{i}\binom{q}{i}G_{n,p-i,i}^{r,s,m}, & p>q, \\
\frac{1}{(s-r)^{p}}\sum\limits_{i=0}^{p}(-1)^{i}\binom{p}{i}G_{n,i,q-i}^{r,s,m},& p<q.
\end{array}
\right.
$$
\end{theorem}
\begin{proof}
We first consider the case $q=0$,

$$
\begin{array}{l}
G_{n,p,0}^{r,0,m}
=\sum\limits_{j=1}^{n}\frac{H_{j,m}}{(j+r)^{p}}
=\sum\limits_{j=1}^{n}\frac{H_{j+r,m}-\left(\frac{1}{(j+1)^{m}}+\ldots+\frac{1}{(j+r)^{m}} \right)}{(j+r)^{p}}
=G_{n+r,p,0}^{0,0,m}-G_{r,p,0}^{0,0,m} -\sum\limits_{i=1}^{r}R_{n,m,p}^{i,r}.
\end{array}
$$
We turn to the case $q<p$ (the case $p>q$ can be handled symmetrically),

$$
\begin{array}{l}
G_{n,p,q}^{r,s,m}=\sum\limits_{j=1}^{n}\frac{H_{j,m}}{(j+r)^{p}(j+s)^{q}}
= \sum\limits_{j=1}^{n}\frac{H_{j,m}}{(j+r)^{p-q}}\left(\frac{1}{s-r}\left(\frac{1}{j+r}-\frac{1}{j+s}\right)\right)^{q}
\\ =
\frac{1}{(s-r)^{q}}\sum\limits_{i=0}^{q}
(-1)^{i}\binom{q}{i}\sum\limits_{j=1}^{n}
\frac{H_{j,m}}{(j+r)^{p-i}(j+s)^{i}}
=
\frac{1}{(s-r)^{q}}\sum\limits_{i=0}^{q}
(-1)^{i}\binom{q}{i}G_{n,p-i,i}^{r,s,m}.
\end{array}
$$
\end{proof}

\begin{theorem}\label{thmVpqrsm}
Define, for integers $m\ge 1$, $p,q\ge 0$ (but $p+q>1$), and $r,s\ge 0$

$$
V_{n,p,q}^{r,s,m}:=\sum\limits_{j=1}^{n}\frac{H_{j,m}^{2}}{(j+r)^{p}(j+s)^{q}}.
$$
We then we have the following recursive relationship 

$$
V_{n,p,q}^{r,s,m} = \left\{
\begin{array}{cc}
\begin{array}{l}
V_{n,p,0}^{0,0,m}-V_{r,p,0}^{0,0,m}
-2\sum\limits_{i=1}^{r}G_{n,p,m}^{r,i,m}
\\ +\sum\limits_{i=1}^{r}R_{n,p,2m}^{r,i}
-2\sum\limits_{i,k=1}^{r}R_{n,p,m,m}^{r,i,k}
+2\sum\limits_{1=i_{1}<i_{2}}^{r}R_{n,p,m,m}^{r,i_{1},i_{2}}
,
\end{array}
& q=0, \\
\frac{1}{(s-r)^{q}}\sum\limits_{i=0}^{q}(-1)^{i}\binom{q}{i}V_{n,p-i,i}^{r,s,m}& p>q, \\
\frac{1}{(s-r)^{p}}\sum\limits_{i=0}^{p}(-1)^{i}\binom{p}{i}V_{n,i,q-i}^{r,s,m}& p<q.
\end{array}
\right.
$$
\end{theorem}
\begin{proof}
We first consider the case $q=0$,

$$
\begin{array}{l}
V_{n,p,0}^{r,0,m}
=\sum\limits_{j=1}^{n}\frac{H_{j,m}^{2}}{(j+r)^{p}}
=\sum\limits_{j=1}^{n}\frac{\left(H_{j+r,m}-\left(\frac{1}{(j+1)^{m}}+\ldots+\frac{1}{(j+r)^{m}} \right)\right)^{2}}{(j+r)^{p}}
\\ =
V_{n,p,0}^{0,0,m}-V_{r,p,0}^{0,0,m}
-2\sum\limits_{i=1}^{r}G_{n,p,m}^{r,i,m}
-2\sum\limits_{i=1}^{r}\sum\limits_{k=1}^{r}R_{n,p,m,m}^{r,i,k}
+\sum\limits_{i=1}^{r}R_{n,p,2m}^{r,i}
\\ +2\sum\limits_{i_{1}=1}^{r-1}\sum\limits_{i_{2}=i_{1}+1}^{r}R_{n,p,m,m}^{r,i_{1},i_{2}}
\end{array}
$$
We turn to the case $q<p$ (the case $p>q$ can be handled symmetrically),

$$
\begin{array}{l}
V_{n,p,q}^{r,s,m}
=\sum\limits_{j=1}^{n}\frac{H_{j,m}^{2}}{(j+r)^{p}(j+s)^{q}}
= \sum\limits_{j=1}^{n}\frac{H_{j,m}^{2}}{(j+r)^{p-q}}\left(\frac{1}{s-r}\left(\frac{1}{j+r}-\frac{1}{j+s}\right)\right)^{q}
\\ =
\frac{1}{(s-r)^{q}}\sum\limits_{i=0}^{q}
(-1)^{i}\binom{q}{i}\sum\limits_{j=1}^{n}
\frac{H_{j,m}^{2}}{(j+r)^{p-i}(j+s)^{i}}
=
\frac{1}{(s-r)^{q}}\sum\limits_{i=0}^{q}(-1)^{i}\binom{q}{i}V_{n,p-i,i}^{r,s,m}.
\end{array}
$$
\end{proof}

We can also provide a number of more specialized recursive formul\ae\ , in terms of $m$, with $n$ held constant, 
which will be useful in particular cases. 

\begin{lemma}\label{lemSumj1sqjm}
For $m\ge 2$ it holds that 

\be\label{eqSumj1sqjm}
\begin{array}{rcl}
\sum\limits_{j=1}^{n}\frac{1}{(j+1)^{2}(j+m)} & = &
\left\{
\begin{array}{lc}
H_{n+1,3}-1,~~~~~~~~~~~~~~~~~~~~~~~~~~~~~~~~~~~~~~~~~~~~  m=1, \\
\frac{H_{n+1,2}}{m-1} - \frac{1}{m-1}
- \frac{H_{m,1}}{(m-1)^{2}}+ \frac{1}{(m-1)^{2}}+\frac{1}{(m-1)^{2}}\sum\limits_{j=n+2}^{n+m}\frac{1}{j}, m \ge 2.
\end{array}
\right.
\\ & \xlongrightarrow{n \to \infty} &
\left\{
\begin{array}{lc}
\zeta(3)-1, & m= 1,\\
\frac{\pi^{2}}{6(m-1)} - \frac{1}{m-1} 
- \frac{H_{m,1}}{(m-1)^{2}}+ \frac{1}{(m-1)^{2}}, & m \ge 2.
\end{array}
\right.
\end{array}
\ee
\end{lemma}

\begin{proof}
For $m=1$ the result is obvious and for $m\ge 2$ we have
$$
\begin{array}{l}
\sum\limits_{j=1}^{n}\frac{1}{(j+1)^{2}(j+m)} 
=\sum\limits_{j=1}^{n}\frac{1}{(j+1)}\frac{1}{m-1}\left(\frac{1}{j+1}-\frac{1}{j+m}\right)
\\=\frac{1}{m-1}\left(\sum\limits_{j=1}^{n}\frac{1}{(j+1)^{2}}-\sum\limits_{j=1}^{n}\frac{1}{(j+1)(j+m)}\right)
=\frac{1}{m-1}\left(H_{n+1,2}-1-\sum\limits_{j=1}^{n}\frac{1}{(j+1)(j+m)}\right) 
\\=\frac{1}{m-1}\left(H_{n+1,2}-1-\frac{1}{m-1}\left(\sum\limits_{j=1}^{n}\frac{1}{j+1}-\sum\limits_{j=1}^{n}\frac{1}{j+m}\right)\right) %OK
\\ =\frac{1}{m-1}\left(H_{n+1,2}-1-\frac{1}{m-1}\left(H_{n+1,1}-1-H_{n+m,1}+H_{m,1}\right)\right)
\\ =
\frac{H_{n+1,2}}{m-1} - \frac{1}{m-1}
- \frac{H_{m,1}}{(m-1)^{2}}+ \frac{1}{(m-1)^{2}}+\frac{1}{(m-1)^{2}}\sum\limits_{j=n+2}^{n+m}\frac{1}{j} .
\end{array}
$$
\end{proof}

\newpage
\begin{lemma}\label{lemSumHj1jm}
For $m>1$ we have the following recursive representation 

\be\label{eqSumHj1jm}
G_{n,1,0}^{m,0,1}=\sum\limits_{j=1}^{n}\frac{H_{j,1}}{j+m} = 
G_{n,1,0}^{m-1,0,1}
+\frac{1}{m-1}\sum\limits_{j=n+2}^{n+m}\frac{1}{j}
-\frac{H_{m,1}}{m-1}
+\frac{1}{m(m-1)}
+\frac{H_{n+1,1}}{n+m},
\ee
with initial condition $m=1$ 

\be\label{eqSumHj1jq}
G_{n,1,0}^{1,0,1}=\sum\limits_{j=1}^{n}\frac{H_{j,1}}{j+1} 
= \frac{1}{2}\left(H_{n+1,1}^{2}-H_{n+1,2} \right)\sim \frac{1}{2}\ln^{2} n.
\ee
\end{lemma}

\begin{proof}
We first consider the special case $m=1$

$$
\begin{array}{l}
G_{n,1,0}^{1,0,1}=\sum\limits_{j=1}^{n}\frac{H_{j,1}}{j+1} 
= \sum\limits_{j=1}^{n}\frac{H_{j+1,1}}{j+1} - \sum\limits_{j=1}^{n}\frac{1}{(j+1)^{2}}
= \sum\limits_{j=1}^{n+1}\frac{H_{j,1}}{j}-1 - H_{n+1,2}+1
\\ \stackrel{Lemma~\ref{lemSumHii1}}{=}\frac{1}{2}\left(H_{n+1,1}^{2}+H_{n+1,2} \right)-H_{n+1,2}
=\frac{1}{2}\left(H_{n+1,1}^{2}-H_{n+1,2} \right).
\end{array}
$$
Now turning to the general case

$$
\begin{array}{l}
G_{n,1,0}^{m,0,1}=\sum\limits_{j=1}^{n}\frac{H_{j,1}}{j+m} = 
\sum\limits_{j=1}^{n}\frac{H_{j+1,1}}{j+m} -\sum\limits_{j=1}^{n}\frac{1}{(j+1)(j+m)}
\\ =\sum\limits_{j=1}^{n+1}\frac{H_{j,1}}{j+m-1}-\frac{1}{m} 
-\frac{1}{m-1}\left(
\sum\limits_{j=1}^{n}\frac{1}{j+1}- \sum\limits_{j=1}^{n}\frac{1}{j+m}
\right)
\\ =\sum\limits_{j=1}^{n+1}\frac{H_{j,1}}{j+m-1}-\frac{1}{m} 
-\frac{1}{m-1}\left(H_{n+1}-1-H_{n+m,1}+H_{m,1}
\right)
\\=\sum\limits_{j=1}^{n}\frac{H_{j,1}}{j+m-1}
+\frac{1}{m-1}\sum\limits_{j=n+2}^{n+m}\frac{1}{j}
-\frac{H_{m,1}}{m-1}
+\frac{1}{m(m-1)}
+\frac{H_{n+1,1}}{n+m}.
\end{array}
$$
\end{proof}

\newpage
\begin{lemma}\label{lemSumHj2jm}
For $m>1$ we have the following recursive representation 

\be\label{eqSumHj2jm}
\begin{array}{rcl}
G_{n,1,0}^{m,0,2}=\sum\limits_{j=1}^{n}\frac{H_{j,2}}{j+m}& = &
G_{n,1,0}^{m-1,0,2}
-\frac{1}{m(m-1)^{2}}-\frac{H_{n+1,2}}{m-1} + \frac{H_{m,1}}{(m-1)^{2}}
+\frac{H_{n+1,2}}{n+m} 
\\&&-\frac{1}{(m-1)^{2}}\sum\limits_{j=n+2}^{n+m}\frac{1}{j}, 
\end{array}
\ee
with initial condition $m=1$ 

\be\label{eqSumHj2j1}
G_{n,1,0}^{1,0,2}=\sum\limits_{j=1}^{n}\frac{H_{j,2}}{j+1} 
= H_{n+1,1}H_{n+1,2}-\sum\limits_{j=1}^{n+1}\frac{H_{j,1}}{j^{2}}.
\ee

\end{lemma}
\begin{proof}

The special case $m=1$  can be written as  

$$G_{n,1,0}^{1,0,2}=H_{n,1}H_{n,2}-G_{n,2,0}^{0,0,1}-\frac{H_{n,2}}{n+1}$$
and proven as follows

$$
\begin{array}{l} \sum\limits_{j=1}^{n}\frac{H_{j,2}}{j+1} 
=\sum\limits_{j=1}^{n}\frac{H_{j+1,2}}{j+1} - H_{n+1,3}+1
=\sum\limits_{j=1}^{n+1}\frac{H_{j,2}}{j}-1 - H_{n+1,3}+1
\\ \stackrel{Lemma~\ref{lemSumHi2i}}{=} H_{n+1,1}H_{n+1,2} + H_{n+1,3}-  \sum\limits_{j=1}^{n+1}\frac{H_{j,1}}{j^{2}} -1 - H_{n+1,3}+1
\\ = H_{n+1,1}H_{n+1,2} -  \sum\limits_{j=1}^{n+1}\frac{H_{j,1}}{j^{2}}.
\end{array}
$$
In the case $m>1$ we have.

$$
\begin{array}{l}
\sum\limits_{j=1}^{n}\frac{H_{j,2}}{j+m} 
= \sum\limits_{j=1}^{n}\frac{H_{j+1,2}}{j+m} -\sum\limits_{j=1}^{n}\frac{1}{(j+1)^{2}(j+m)}
= \sum\limits_{j=2}^{n+1}\frac{H_{j,2}}{j+m-1} -\sum\limits_{j=1}^{n}\frac{1}{(j+1)^{2}(j+m)}
\\ \stackrel{Lemma~ \ref{lemSumj1sqjmapp}}{=}
\sum\limits_{j=1}^{n}\frac{H_{j,2}}{j+m-1}-\frac{1}{m}+\frac{H_{n+1,2}}{n+m}
-\left(
\frac{H_{n+1,2}}{m-1} - \frac{1}{m-1}
- \frac{H_{m,1}}{(m-1)^{2}}+ \frac{1}{(m-1)^{2}}+\frac{1}{(m-1)^{2}}\sum\limits_{j=n+2}^{n+m}\frac{1}{j}
\right)
\\
= G_{n,1,0}^{m-1,0,2}
-\frac{1}{m}+ \frac{1}{m-1} - \frac{1}{(m-1)^{2}}
+\frac{H_{n+1,2}}{n+m}-\frac{H_{n+1,2}}{m-1} + \frac{H_{m,1}}{(m-1)^{2}}
-\frac{1}{(m-1)^{2}}\sum\limits_{j=n+2}^{n+m}\frac{1}{j}
\\ =
G_{n,1,0}^{m-1,0,2}
-\frac{1}{m(m-1)^{2}}-\frac{H_{n+1,2}}{m-1} + \frac{H_{m,1}}{(m-1)^{2}}
+\frac{H_{n+1,2}}{n+m} -\frac{1}{(m-1)^{2}}\sum\limits_{j=n+2}^{n+m}\frac{1}{j}.
\end{array}
$$
\end{proof}

\begin{lemma}\label{lemSumHj1sqjm}
For $m>1$ we have the following recursive representation 

\be\label{eqSumHj1sqjm}
\begin{array}{rcl}
V_{n,1,0}^{m,0,1}=\sum\limits_{j=1}^{n}\frac{H_{j,1}^{2}}{j+m} & = &
V_{n,1,0}^{m-1,0,1}+\frac{2}{m-1}G_{n,1,0}^{m-1,0,1}
-\frac{H_{n+1,1}^{2}}{m-1}
+ \frac{1}{m(m-1)^{2}}
- \frac{H_{m,1}}{(m-1)^{2}}
\\&& +\frac{1}{(m-1)^{2}}\sum\limits_{j=n+2}^{n+m}\frac{1}{j}
+\frac{H_{n+1,1}^{2}}{n+m}+\frac{2H_{n+1,1}}{(m-1)(n+m)}
\end{array}
\ee
with initial condition $m=1$ 

\be\label{lemSumHj1sqj1}
\begin{array}{rcl}
V_{n,1,0}^{1,0,1}&=&\sum\limits_{j=1}^{n}\frac{H_{j,1}^{2}}{j+1} 
= \frac{1}{3}H_{n,1}^{3}-\sum\limits_{j=1}^{n}\frac{H_{j,1}}{j^{2}}+\frac{2}{3}H_{n,3}
+\frac{H_{n,1}^{2}}{n+1}
\sim\frac{1}{3}\ln^{3}n.
\end{array}
\ee
\end{lemma}

\begin{proof}
We first consider the initial condition $m=1$ 

$$
\begin{array}{l}
V_{n,1,0}^{1,0,1}=\sum\limits_{j=1}^{n}\frac{H_{j,1}^{2}}{j+1} 
=\sum\limits_{j=1}^{n}\frac{\left(H_{j+1,1}-\frac{1}{j+1}\right)^{2}}{j+1} 
=\sum\limits_{j=1}^{n}\frac{H_{j+1,1}^{2}}{j+1} 
-2\sum\limits_{j=1}^{n}\frac{H_{j+1,1}}{(j+1)^{2}} 
+\sum\limits_{j=1}^{n}\frac{1}{(j+1)^{3}} 
\\ = \sum\limits_{j=1}^{n+1}\frac{H_{j,1}^{2}}{j} -1
-2\left(
\sum\limits_{j=1}^{n}\frac{H_{j,1}}{j^{2}}+\frac{H_{n+1,1}}{(n+1)^{2}}
\right)
+H_{n+1,3}-1
\\ =
\sum\limits_{j=1}^{n}\frac{H_{j,1}^{2}}{j}-2\sum\limits_{j=1}^{n}\frac{H_{j,1}}{j^{2}}+H_{n+1,3}
+\frac{H_{n+1,1}^{2}}{n+1}-\frac{2H_{n+1,1}}{(n+1)^{2}}
\\ \stackrel{Lemma~\ref{lemSumHj1sqj}}{=}
\frac{1}{3}H_{n,1}^{3}+\sum\limits_{i=1}^{n}\frac{H_{j,1}}{j^{2}}-\frac{1}{3}H_{n,3}
-2\sum\limits_{j=1}^{n}\frac{H_{j,1}}{j^{2}}+H_{n+1,3}
+\frac{H_{n+1,1}^{2}}{n+1}-\frac{2H_{n+1,1}}{(n+1)^{2}}
\\=
\frac{1}{3}H_{n,1}^{3}-\sum\limits_{j=1}^{n}\frac{H_{j,1}}{j^{2}}+\frac{2}{3}H_{n,3}
+\frac{H_{n+1,1}^{2}}{n+1}-\frac{2H_{n+1,1}}{(n+1)^{2}}+\frac{1}{(n+1)^{3}}.
\end{array}
$$
Then in the case $m>1$ we have

$$
\begin{array}{l}
V_{n,1,0}^{m,0,1}=\sum\limits_{j=1}^{n}\frac{H_{j,1}^{2}}{j+m} 
= \sum\limits_{j=1}^{n}\frac{\left(H_{j+1,1}-\frac{1}{j+1}\right)^{2}}{j+m}
= 
\sum\limits_{j=1}^{n}\frac{H_{j+1,1}^{2}}{j+m}
-2\sum\limits_{j=1}^{n}\frac{H_{j+1,1}}{(j+1)(j+m)}
\\+\sum\limits_{j=1}^{n}\frac{1}{(j+1)^{2}(j+m)}
\\ \stackrel{Lemma~\ref{lemSumj1sqjmapp}}{=}
\sum\limits_{j=2}^{n+1}\frac{H_{j,1}^{2}}{j+m-1}
-\frac{2}{m-1}\left(\sum\limits_{j=1}^{n}\frac{H_{j+1,1}}{j+1}-\sum\limits_{j=1}^{n}\frac{H_{j+1,1}}{j+m}\right)
\\+\frac{H_{n+1,2}}{m-1} - \frac{1}{m-1}
- \frac{H_{m,1}}{(m-1)^{2}}+ \frac{1}{(m-1)^{2}}
+\frac{1}{(m-1)^{2}}\sum\limits_{j=n+2}^{n+m}\frac{1}{j}
\end{array}
$$

$$
\begin{array}{l}
=
\sum\limits_{j=1}^{n}\frac{H_{j,1}^{2}}{j+m-1}-\frac{1}{m}+\frac{H_{n+1,1}^{2}}{n+m}
-\frac{2}{m-1}\left(\sum\limits_{j=2}^{n+1}\frac{H_{j,1}}{j}
-\sum\limits_{j=1}^{n}\frac{H_{j,1}}{j+m}
-\sum\limits_{j=1}^{n}\frac{1}{(j+1)(j+m)}
\right)
\\ +
\frac{H_{n+1,2}}{m-1} - \frac{1}{m-1}- \frac{H_{m,1}}{(m-1)^{2}}+ \frac{1}{(m-1)^{2}}+\frac{1}{(m-1)^{2}}\sum\limits_{j=n+2}^{n+m}\frac{1}{j}
\end{array}
$$

$$
\begin{array}{l}
%\\
\stackrel{Lemma ~\ref{lemSumHj1jmapp}}{=}
\sum\limits_{j=1}^{n}\frac{H_{j,1}^{2}}{j+m-1}
-\frac{2}{m-1}
\left(\sum\limits_{j=1}^{n+1}\frac{H_{j,1}}{j}-1
-\left(
\sum\limits_{j=1}^{n}\frac{H_{j,1}}{j+m-1}
+\frac{1}{m-1}\sum\limits_{j=n+2}^{n+m}\frac{1}{j}
-\frac{H_{m,1}}{m-1}
\right.\right.\\ \left. \left.
+\frac{1}{m(m-1)}
+\frac{H_{n+1,1}}{n+m}
\right)%policzone: \sum\limits_{j=1}^{n}\frac{H_{j,1}}{j+m}
-\frac{1}{m-1}\left(\sum\limits_{j=1}^{n}\frac{1}{j+1}
- \sum\limits_{j=1}^{n}\frac{1}{j+m}\right)
\right)
+
\frac{H_{n+1,2}}{m-1} - \frac{1}{m-1}-\frac{1}{m}- \frac{H_{m,1}}{(m-1)^{2}}
\\ + \frac{1}{(m-1)^{2}}+\frac{1}{(m-1)^{2}}\sum\limits_{j=n+2}^{n+m}\frac{1}{j}
+\frac{H_{n+1,1}^{2}}{n+m}
\end{array}
$$

$$
\begin{array}{l}
%\\
\stackrel{Lemma~\ref{lemSumHii1}}{=}
\\
\sum\limits_{j=1}^{n}\frac{H_{j,1}^{2}}{j+m-1}
-\frac{H_{n+1,1}^{2}}{m-1}-\frac{H_{n+1,2}}{m-1}+\frac{2}{m-1}
+\frac{2}{m-1}\sum\limits_{j=1}^{n}\frac{H_{j,1}}{j+m-1}
\\ +
\frac{H_{n+1,2}}{m-1}- \frac{H_{m,1}}{(m-1)^{2}}-\frac{2H_{m,1}}{(m-1)^{2}}+\frac{2H_{m,1}}{(m-1)^{2}}
 - \frac{2m^{2}-4m+1}{m(m-1)^{2}}+\frac{2}{m(m-1)^{2}}-\frac{2}{(m-1)^{2}}
\\ +\frac{1}{(m-1)^{2}}\sum\limits_{j=n+2}^{n+m}\frac{1}{j}
+\frac{2}{(m-1)^{2}}\sum\limits_{j=n+2}^{n+m}\frac{1}{j}
-\frac{2}{(m-1)^{2}}\sum\limits_{j=n+2}^{n+m}\frac{1}{j}
+\frac{H_{n+1,1}^{2}}{n+m}+\frac{2H_{n+1,1}}{(m-1)(n+m)}
\end{array}
$$

$$
\begin{array}{l}
%\\
=
V_{n,1,0}^{m-1,0,1}+\frac{2}{m-1}G_{n,1,0}^{m-1,0,1}
-\frac{H_{n+1,1}^{2}}{m-1}
+ \frac{1}{m(m-1)^{2}}
- \frac{H_{m,1}}{(m-1)^{2}}
+\frac{1}{(m-1)^{2}}\sum\limits_{j=n+2}^{n+m}\frac{1}{j}
\\+\frac{H_{n+1,1}^{2}}{n+m}+\frac{2H_{n+1,1}}{(m-1)(n+m)}
.
\end{array}
$$
\end{proof}

\begin{lemma}\label{lemSumHj1sqjmsq}
For $m>1$ we have the following recursive representation 

\be\label{eqSumHj1sqjmsq}
\begin{array}{rcl}
V_{n,2,0}^{m,0,1}&=&\sum\limits_{j=1}^{n}\frac{H_{j,1}^{2}}{(j+m)^{2}}  = 
V_{n,2,0}^{m-1,0,1}+\frac{2}{m-1}G_{n,2,0}^{m-1,0,1}+\frac{2}{(m-1)^{2}}G_{n,1,0}^{m-1,0,1}
\\&&
-\frac{H_{n,1}^{2}}{(m-1)^{2}}
+\frac{H_{n,2}}{(m-1)^{2}}
-\frac{H_{m-1,2}}{(m-1)^{2}}-\frac{2H_{m-1,1}}{(m-1)^{3}}
+\frac{H_{n,1}^{2}}{(n+m)^{2}}
\\&&
+\frac{2}{(m-1)^{3}}\sum\limits_{j=n+1}^{n+m-1}\frac{1}{j}
+\frac{1}{(m-1)^{2}}\sum\limits_{j=n+1}^{n+m-1}\frac{1}{j^{2}}.
\end{array}
\ee
with initial condition, $m=1$, 

\be\label{eqSumHisqip1sq}
\begin{array}{rcl}
V_{n,2,0}^{1,0,1} &=& \sum\limits_{j=1}^{n}\frac{H_{j,1}^{2}}{(j+1)^{2}}  
% =  \sum\limits_{j=1}^{n+1}\frac{H_{j,1}^{2}}{j^{2}}-2\sum\limits_{j=1}^{n+1}\frac{H_{j,1}}{j^{3}}+H_{n+1,4} 
 =  \sum\limits_{j=1}^{n}\frac{H_{j,1}^{2}}{j^{2}}-2\sum\limits_{j=1}^{n}\frac{H_{j,1}}{j^{3}}+H_{n,4} +\frac{H_{n,1}^{2}}{(n+1)^{2}}.
\end{array}
\ee
\end{lemma}

\begin{proof}
We first consider the initial condition $m=1$ 

$$
\begin{array}{l}
\sum\limits_{j=1}^{n}\frac{H_{j,1}^{2}}{(j+1)^{2}} = 
\sum\limits_{j=1}^{n}\frac{\left(H_{j+1,1}-\frac{1}{j+1}\right)^{2}}{(j+1)^{2}}
=
\sum\limits_{j=1}^{n}\frac{H_{j+1,1}^{2}}{(j+1)^{2}}
-2\sum\limits_{j=1}^{n}\frac{H_{j+1,1}}{(j+1)^{3}}
    +\sum\limits_{j=1}^{n}\frac{1}{(j+1)^{4}}
\\=
\sum\limits_{j=2}^{n+1}\frac{H_{j,1}^{2}}{j^{2}}
-2\sum\limits_{j=2}^{n+1}\frac{H_{j,1}}{j^{3}}
+\sum\limits_{j=2}^{n=1}\frac{1}{j^{4}}
=
\sum\limits_{j=1}^{n+1}\frac{H_{j,1}^{2}}{j^{2}}-1
-2\sum\limits_{j=1}^{n+1}\frac{H_{j,1}}{j^{3}}+2
+\sum\limits_{j=1}^{n+1}\frac{1}{j^{4}}-1
\\=
\sum\limits_{j=1}^{n+1}\frac{H_{j,1}^{2}}{j^{2}}
-2\sum\limits_{j=1}^{n+1}\frac{H_{j,1}}{j^{3}}
+H_{n+1,4}
.
\end{array}
$$
Then, in the case $m>1$, we have 

$$
\begin{array}{l}
%V_{n,2,0}^{m,0,1}=
\sum\limits_{j=1}^{n}\frac{H_{j,1}^{2}}{(j+m)^{2}} 
= \sum\limits_{j=1}^{n}\frac{\left(H_{j+1,1}-\frac{1}{j+1}\right)^{2}}{(j+m)^{2}}
\\= 
\sum\limits_{j=1}^{n}\frac{H_{j+1,1}^{2}}{(j+1+m-1)^{2}}
-2\sum\limits_{j=1}^{n}\frac{H_{j+1,1}}{(j+1)(j+1+m-1)^{2}}
+\sum\limits_{j=1}^{n}\frac{1}{(j+1)^{2}(j+1+m-1)^{2}}
\\=
\sum\limits_{j=1}^{n}\frac{H_{j,1}^{2}}{(j+m-1)^{2}}-\frac{1}{m^{2}}+\frac{H_{n+1,1}^{2}}{(n+m)^{2}}
-2\sum\limits_{j=1}^{n}\frac{H_{j,1}}{j(j+m-1)^{2}}+\frac{2}{m^{2}}-\frac{2H_{n+1,1}}{(n+1)(n+m)^{2}}
\\+\sum\limits_{j=1}^{n}\frac{1}{j^{2}(j+m-1)^{2}}-\frac{1}{m^{2}}+\frac{1}{(n+1)^{2}(n+m)^{2}}
\\=
V_{n,2,0}^{m-1,0,1}
-\frac{2}{m-1}\sum\limits_{j=1}^{n}\frac{H_{j,1}}{j(j+m-1)}
+\frac{2}{m-1}\sum\limits_{j=1}^{n}\frac{H_{j,1}}{(j+m-1)^{2}}
+\frac{1}{m-1}\sum\limits_{j=1}^{n}\frac{1}{j^{2}(j+m-1)}
-\frac{1}{m-1}\sum\limits_{j=1}^{n}\frac{1}{j(j+m-1)^{2}}
\\+\frac{H_{n,1}^{2}}{(n+m)^{2}}
+2\frac{H_{n,1}}{(n+1)(n+m)^{2}}
+\frac{1}{(n+1)^{2}(n+m)^{2}}
-\frac{2H_{n,1}}{(n+1)(n+m)^{2}}
-\frac{2}{(n+1)^{2}(n+m)^{2}}
+\frac{1}{(n+1)^{2}(n+m)^{2}}
\\=
V_{n,2,0}^{m-1,0,1}
-\frac{2}{(m-1)^{2}}\sum\limits_{j=1}^{n}\frac{H_{j,1}}{j}
+\frac{2}{(m-1)^{2}}\sum\limits_{j=1}^{n}\frac{H_{j,1}}{(j+m-1)}
+\frac{2}{m-1}G_{n,2,0}^{m-1,0,1}
+\frac{1}{(m-1)^{2}}\sum\limits_{j=1}^{n}\frac{1}{j^{2}}
\\-\frac{1}{(m-1)^{2}}\sum\limits_{j=1}^{n}\frac{1}{j(j+m-1)}
-\frac{1}{(m-1)^{2}}\sum\limits_{j=1}^{n}\frac{1}{j(j+m-1)}
+\frac{1}{(m-1)^{2}}\sum\limits_{j=1}^{n}\frac{1}{(j+m-1)^{2}}
+\frac{H_{n,1}^{2}}{(n+m)^{2}}
\end{array}
$$

$$
\begin{array}{l}
%\\
\stackrel{Lemma~\ref{lemSumHii1}}{=}
V_{n,2,0}^{m-1,0,1}
-\frac{2}{(m-1)^{2}}\frac{1}{2}\left(H_{n,1}^{2}+H_{n,2} \right)
+\frac{2}{(m-1)^{2}}G_{n,1,0}^{m-1,0,1}
+\frac{2}{m-1}G_{n,2,0}^{m-1,0,1}
+\frac{1}{(m-1)^{2}}\sum\limits_{j=1}^{n}\frac{1}{j^{2}}
\\-\frac{1}{(m-1)^{2}}\sum\limits_{j=1}^{n}\frac{1}{j(j+m-1)}
-\frac{1}{(m-1)^{2}}\sum\limits_{j=1}^{n}\frac{1}{j(j+m-1)}
+\frac{1}{(m-1)^{2}}\sum\limits_{j=1}^{n}\frac{1}{(j+m-1)^{2}}
+\frac{H_{n,1}^{2}}{(n+m)^{2}}
%\\=V_{n,2,0}^{m-1,0,1}+\frac{2}{m-1}G_{n,2,0}^{m-1,0,1}+\frac{2}{(m-1)^{2}}G_{n,1,0}^{m-1,0,1}-\frac{H_{n,1}^{2}}{(m-1)^{2}}-\frac{H_{n,2}}{(m-1)^{2}}+\frac{H_{n,2}}{(m-1)^{2}}-\frac{2}{(m-1)^{3}}\sum\limits_{j=1}^{n}\frac{1}{j}\\ +\frac{2}{(m-1)^{3}}\sum\limits_{j=1}^{n}\frac{1}{j+m-1}+\frac{1}{(m-1)^{2}}(H_{n+m-1,2}-H_{m-1,2})+\frac{H_{n,1}^{2}}{(n+m)^{2}}
\\=V_{n,2,0}^{m-1,0,1}+\frac{2}{m-1}G_{n,2,0}^{m-1,0,1}+\frac{2}{(m-1)^{2}}G_{n,1,0}^{m-1,0,1}
-\frac{H_{n,1}^{2}}{(m-1)^{2}}+\frac{H_{n,2}}{(m-1)^{2}}
-\frac{H_{m-1,2}}{(m-1)^{2}}-\frac{2H_{m-1,1}}{(m-1)^{3}}
\\+\frac{H_{n,1}^{2}}{(n+m)^{2}}
+\frac{2}{(m-1)^{3}}\sum\limits_{j=n+1}^{n+m-1}\frac{1}{j}
+\frac{1}{(m-1)^{2}}\sum\limits_{j=n+1}^{n+m-1}\frac{1}{j^{2}}.
\end{array}
$$
\end{proof}

\begin{lemma}\label{lemSumHi2sqim}
For $m>1$ we have the following recursive representation 

\be\label{eqSumHi2sqim}
\begin{array}{rcl}
V_{n,1,0}^{m,0,2}&=&\sum\limits_{j=1}^{n}\frac{H_{j,2}^{2}}{(j+m)}  = 
V_{n,1,0}^{m-1,0,2}-\frac{2}{(m-1)^{2}}G_{n,1,0}^{m-1,0,2}
-\frac{H_{n,2}^{2}}{m-1}+\frac{2H_{n,1}H_{n,2}}{(m-1)^{2}}
\\&&+\frac{H_{n,2}}{(m-1)^{3}}+\frac{H_{n,3}}{(m-1)^{2}}
-\frac{2}{(m-1)^{2}}\sum\limits_{j=1}^{n}\frac{H_{j,1}}{j^{2}}
 -\frac{H_{m-1,1}}{(m-1)^{4}}
+\frac{H_{n,2}^{2}}{(n+m)}
\\&&
+\frac{1}{(m-1)^{4}}\sum\limits_{j=n+1}^{n+m-1}\frac{1}{j}.
\end{array}
\ee
with initial condition, $m=1$,

\be\label{eqSumHi2sqi1}
V_{n,1,0}^{1,0,2}=
\sum\limits_{j=1}^{n}\frac{H_{j,2}^{2}}{j+1} = 
\sum\limits_{j=1}^{n}\frac{H_{j,2}^{2}}{j} - 2\sum\limits_{j=1}^{n}\frac{H_{j,2}}{j^{3}} +H_{n,5} + \frac{H_{n,2}^{2}}{n+1}.
\ee 
\end{lemma}

\begin{proof}
We first consider the initial condition $m=1$ 

$$
\begin{array}{l}
\sum\limits_{j=1}^{n}\frac{H_{j,2}^{2}}{j+1} =
\sum\limits_{j=1}^{n}\frac{\left(H_{j+1,2}-\frac{1}{(j+1)^{2}}\right)^{2}}{j+1} 
=
\sum\limits_{j=1}^{n}\frac{H_{j+1,2}^{2}}{j+1} -
2\sum\limits_{j=1}^{n}\frac{H_{j+1,2}}{(j+1)^{3}} 
+\sum\limits_{j=1}^{n}\frac{1}{(j+1)^{5}} 
\\ =
\sum\limits_{j=1}^{n+1}\frac{H_{j,2}^{2}}{j} - 1
- 2\sum\limits_{j=1}^{n+1}\frac{H_{j,2}}{j^{3}} + 2
+H_{n+1,5}-1
=
\sum\limits_{j=1}^{n+1}\frac{H_{j,2}^{2}}{j} 
- 2\sum\limits_{j=1}^{n+1}\frac{H_{j,2}}{j^{3}} 
+H_{n+1,5}.
\end{array}
$$
Then, in the case $m>1$, we have 

$$
\begin{array}{l}
V_{n,1,0}^{m,0,2}=\sum\limits_{j=1}^{n}\frac{H_{j,2}^{2}}{(j+m)} 
= \sum\limits_{j=1}^{n}\frac{\left(H_{j+1,2}-\frac{1}{j+2}\right)^{2}}{(j+m)}
\\= 
\sum\limits_{j=1}^{n}\frac{H_{j+1,2}^{2}}{(j+1+m-1)}
-2\sum\limits_{j=1}^{n}\frac{H_{j+1,2}}{(j+1)^{2}(j+1+m-1)}
+\sum\limits_{j=1}^{n}\frac{1}{(j+1)^{4}(j+1+m-1)}
\\=
\sum\limits_{j=1}^{n}\frac{H_{j,2}^{2}}{(j+m-1)}-\frac{1}{m}+\frac{H_{n+1,2}^{2}}{(n+m)}
-2\sum\limits_{j=1}^{n}\frac{H_{j,2}}{j^{2}(j+m-1)}+\frac{2}{m}-\frac{2H_{n+1,2}}{(n+1)^{2}(n+m)}
\\+\sum\limits_{j=1}^{n}\frac{1}{j^{4}(j+m-1)}-\frac{1}{m}+\frac{1}{(n+1)^{4}(n+m)}
\\=
V_{n,1,0}^{m-1,0,2}
-\frac{2}{m-1}\sum\limits_{j=1}^{n}\frac{H_{j,2}}{j^{2}}
+\frac{2}{m-1}\sum\limits_{j=1}^{n}\frac{H_{j,2}}{j(j+m-1)}
+\frac{1}{m-1}\sum\limits_{j=1}^{n}\frac{1}{j^{4}}
-\frac{1}{m-1}\sum\limits_{j=1}^{n}\frac{1}{j^{3}(j+m-1)}
\\ +\frac{H_{n+1,2}^{2}}{(n+m)}
-\frac{2H_{n+1,2}}{(n+1)^{2}(n+m)}
+\frac{1}{(n+1)^{4}(n+m)}
\end{array}
$$

$$
\begin{array}{l}
 \stackrel{Lemma~\ref{lemSumHi2isq}}{=}
V_{n,1,0}^{m-1,0,2}
-\frac{2}{m-1}\frac{1}{2}(H_{n,2}^{2}+H_{n,4})
+\frac{2}{(m-1)^{2}}\sum\limits_{j=1}^{n}\frac{H_{j,2}}{j}
-\frac{2}{(m-1)^{2}}\sum\limits_{j=1}^{n}\frac{H_{j,2}}{(j+m-1)}
\\ +\frac{H_{n,4}}{m-1}
-\frac{1}{(m-1)^{2}}\sum\limits_{j=1}^{n}\frac{1}{j^{3}}
+\frac{1}{(m-1)^{2}}\sum\limits_{j=1}^{n}\frac{1}{j^{2}(j+m-1)}
+\frac{H_{n,2}^{2}}{(n+m)}
\\ \stackrel{Lemma~\ref{lemSumHi2i}}{=}
V_{n,1,0}^{m-1,0,2}-\frac{2}{(m-1)^{2}}G_{n,1,0}^{m-1,0,2}
-\frac{H_{n,2}^{2}}{m-1}
+\frac{2}{(m-1)^{2}}(H_{n,1}H_{n,2}+H_{n,3}-\sum\limits_{j=1}^{n}\frac{H_{j,1}}{j^{2}})
\\ -\frac{H_{n,3}}{(m-1)^{2}}
+\frac{1}{(m-1)^{3}}\sum\limits_{j=1}^{n}\frac{1}{j^{2}}
-\frac{1}{(m-1)^{3}}\sum\limits_{j=1}^{n}\frac{1}{j(j+m-1)}
+\frac{H_{n,2}^{2}}{(n+m)}
\\=
V_{n,1,0}^{m-1,0,2}-\frac{2}{(m-1)^{2}}G_{n,1,0}^{m-1,0,2}
-\frac{H_{n,2}^{2}}{m-1}
+\frac{2H_{n,1}H_{n,2}}{(m-1)^{2}}
+\frac{H_{n,3}}{(m-1)^{2}}
-\frac{2}{(m-1)^{2}}\sum\limits_{j=1}^{n}\frac{H_{j,1}}{j^{2}}
+\frac{H_{n,2}}{(m-1)^{3}}
\\ -\frac{1}{(m-1)^{4}}\sum\limits_{j=1}^{n}\frac{1}{j}
+\frac{1}{(m-1)^{4}}\sum\limits_{j=1}^{n}\frac{1}{(j+m-1)}
+\frac{H_{n,2}^{2}}{(n+m)}
\\=
V_{n,1,0}^{m-1,0,2}-\frac{2}{(m-1)^{2}}G_{n,1,0}^{m-1,0,2}
-\frac{H_{n,2}^{2}}{m-1}+\frac{2H_{n,1}H_{n,2}}{(m-1)^{2}}
+\frac{H_{n,2}}{(m-1)^{3}}+\frac{H_{n,3}}{(m-1)^{2}}
-\frac{2}{(m-1)^{2}}\sum\limits_{j=1}^{n}\frac{H_{j,1}}{j^{2}}
\\-\frac{H_{n,1}}{(m-1)^{4}}
+\frac{H_{n,1}}{(m-1)^{4}}
-\frac{H_{m-1,1}}{(m-1)^{4}}
+\frac{1}{(m-1)^{4}}\sum\limits_{j=n+1}^{n+m-1}\frac{1}{j}
+\frac{H_{n,2}^{2}}{(n+m)}
\end{array}
$$

$$
\begin{array}{l}
=
V_{n,1,0}^{m-1,0,2}-\frac{2}{(m-1)^{2}}G_{n,1,0}^{m-1,0,2}
-\frac{H_{n,2}^{2}}{m-1}+\frac{2H_{n,1}H_{n,2}}{(m-1)^{2}}
+\frac{H_{n,2}}{(m-1)^{3}}+\frac{H_{n,3}}{(m-1)^{2}}
-\frac{2}{(m-1)^{2}}\sum\limits_{j=1}^{n}\frac{H_{j,1}}{j^{2}}
\\ -\frac{H_{m-1,1}}{(m-1)^{4}}
+\frac{H_{n,2}^{2}}{(n+m)}
+\frac{1}{(m-1)^{4}}\sum\limits_{j=n+1}^{n+m-1}\frac{1}{j}.
\end{array}
$$
\end{proof}

\section{Example special harmonic and quadratic harmonic sums}
Here we prove our main results, certain harmonic and quadratic harmonic sums in closed form.
We were unable to locate these in the literature, nor could we obtain them with 
\proglang{Mathematica}. In the Appendix 
we further present
a number of sums (that we could not get \proglang{Mathematica 12.0.0} to handle, and in some
cases we could not the limit using the software), which
are presented in terms of harmonic numbers.
These are sums of the form

$$
G_{n,p,q}^{1,2,m}, V_{n,p,1}^{1,2,m}
~\mathrm{and}~\sum\limits_{j=1}^{n}\frac{H_{j+p,1}H_{j+p,2}}{j}
$$
for $p,q\in\{0,1,2\}$, and $m=1,2$.
Here we derive Lemma \ref{lemSumHi2sqisq}, whose sum
is terms of $G_{n,4,0}^{0,0,2}$ 
but whose proof is more involved,
and it is also a building block for other sums
in the Appendix 
(Lemmata \ref{lemSumHi2sqimapp}--\ref{lemSumHi2sqip1sqip2sq}). 
Furthermore, in the Appendix, 
we consider 
$G_{n,2,0}^{0,0,1}$, $G_{n,3,0}^{0,0,1}$, $G_{n,4,0}^{0,0,1}$, $V_{n,2,0}^{0,0,1}$, and $V_{n,4,0}^{0,0,2}$
(Lemma \ref{lemSumHiip1sq} and most following it).

\begin{lemma}\label{lemSumHi2i1sqi2}

\be\label{eqSumHi2i1sqi2}
\begin{array}{rcl}
G_{n,2,1}^{1,2,2}=
\sum\limits_{j=1}^{n}\frac{H_{j,2}}{(j+1)^{2}(j+2)} & = &
\frac{1}{2}H_{n,2}^{2}- H_{n,2}-\frac{1}{2}H_{n,4}  + 1 +\frac{H_{n,2}}{n+2} - \frac{1}{n+1} +\frac{H_{n,2}}{(n+1)^{2}}
\\  & \xrightarrow{n\to \infty} &
\frac{\pi^{4}}{120}-\frac{\pi^{2}}{6}+1  .
\end{array}
\ee
\end{lemma}
\begin{proof}

$$
\begin{array}{l}
\sum\limits_{j=1}^{n}\frac{H_{j,2}}{(j+1)^{2}(j+2)} = 
\sum\limits_{j=1}^{n}\frac{H_{j,2}}{(j+1)}\left(\frac{1}{j+1}-\frac{1}{j+2} \right)
=
\sum\limits_{j=1}^{n}\frac{H_{j,2}}{(j+1)^{2}}-\sum\limits_{j=1}^{n}\frac{H_{j,2}}{(j+1)(j+2)}
\\ \stackrel{Lemmata~\ref{lemSumHi2i1i2}, \ref{lemSumHi2i1sq}}{=}
\frac{1}{2}\left(H_{n+1,2}^{2}-H_{n+1,4}\right)- \frac{n+1}{n+2}\left(H_{n+1,2}-1\right)
\\ =
\frac{1}{2}H_{n+1,2}^{2}-\frac{1}{2}H_{n+1,4}
- H_{n+1,2} + 1 +\frac{H_{n+1,2}}{n+2} - \frac{1}{n+2}
 \xrightarrow{n\to \infty}
\frac{\pi^{4}}{120}-\frac{\pi^{2}}{6}+1 \approx 0.167. 
\end{array}
$$
\end{proof}

\begin{lemma}\label{lemSumHi2i1i2sq}

\be\label{eqSumHi2i1i2sq}
\begin{array}{rcl}
G_{n,1,2}^{1,2,2}&=&
\sum\limits_{j=1}^{n}\frac{H_{j,2}}{(j+1)(j+2)^{2}}  = 
3H_{n,2}-\frac{1}{2}H_{n,2}^{2}+\frac{1}{2}H_{n,4}-4
-\frac{H_{n,2}}{n+2}+\frac{3}{n+1}
\\&&-\frac{H_{n,2}}{(n+1)^{2}}-\frac{H_{n,2}}{(n+2)^{2}}+\frac{1}{(n+1)^{2}}
\xrightarrow{n\to \infty} 
 \frac{\pi^{2}}{2}-\frac{\pi^{4}}{120}-4.
\end{array}
\ee
\end{lemma}
\begin{proof}

$$
\begin{array}{l}
\sum\limits_{j=1}^{n}\frac{H_{j,2}}{(j+1)(j+2)^{2}}
=
\sum\limits_{j=1}^{n}\frac{H_{j,2}}{j+2}\left(\frac{1}{j+1}-\frac{1}{j+2}\right)
=
\sum\limits_{j=1}^{n}\frac{H_{j,2}}{(j+1)(j+2)}
-\sum\limits_{j=1}^{n}\frac{H_{j,2}}{(j+2)^{2}}
\\ \stackrel{Lemmata~\ref{lemSumHi2i1i2},\ref{lemSumHi2i2sq}}{=}
\frac{n+1}{n+2}\left(H_{n+1,2}-1\right)
-
\frac{1}{2}\left(H_{n+2,2}^{2}-H_{n+2,4}\right)+2H_{n+2,2} - 3 + \frac{2n+3}{(n+2)^{2}}
\\
=
H_{n+1,2}-1-\frac{H_{n+1,2}}{n+2}+\frac{1}{n+2}
-\frac{1}{2}H_{n+2,2}^{2}
+\frac{1}{2}H_{n+2,4}
+2H_{n+2,2} - 3 + \frac{2n+3}{(n+2)^{2}}
\\=
3H_{n+1,2}+\frac{2}{(n+2)^{2}}
-\frac{1}{2}H_{n+2,2}^{2}
+\frac{1}{2}H_{n+2,4}
-4
-\frac{H_{n+1,2}}{n+2}
+\frac{1}{n+2}
+ \frac{2n+3}{(n+2)^{2}}
\\
=
3H_{n+1,2}
-\frac{1}{2}H_{n+2,2}^{2}
+\frac{1}{2}H_{n+2,4}
-4
-\frac{H_{n+1,2}}{n+2}
+\frac{1}{n+2}
+ \frac{2n+5}{(n+2)^{2}}
\\ \xrightarrow{n\to \infty}
 \frac{\pi^{2}}{2}-\frac{\pi^{4}}{120}-4 \approx 0.123.
\end{array}
$$
\end{proof}

\begin{lemma}\label{lemSumHi2i1sqi2sq}

\be\label{eqSumHi2i1sqi2sq}
\begin{array}{rcl}
G_{n,2,2}^{1,2,2}&=&
\sum\limits_{j=1}^{n}\frac{H_{j,2}}{(j+1)^{2}(j+2)^{2}}  = 
H_{n,2}^{2} -4H_{n,2}-H_{n,4}  +5 + \frac{2H_{n,2}}{n+2}
\\ &&+ \frac{2H_{n,2}}{(n+1)^{2}}+ \frac{H_{n,2}}{(n+2)^{2}}-\frac{4}{n+1}-\frac{1}{(n+1)^{2}}
\xrightarrow{n\to \infty} 
\frac{\pi^{4}}{60}-\frac{2\pi^{2}}{3}+5.
\end{array}
\ee
\end{lemma}
\begin{proof}

$$
\begin{array}{l}
\sum\limits_{j=1}^{n}\frac{H_{j,2}}{(j+1)^{2}(j+2)^{2}} =
\sum\limits_{j=1}^{n}H_{j,2}\left(\frac{1}{j+1}-\frac{1}{j+2}\right)^{2} =
\sum\limits_{j=1}^{n}\frac{H_{j,2}}{(j+1)^{2}}
-2\sum\limits_{j=1}^{n}\frac{H_{j,2}}{(j+1)(j+2)}
+\sum\limits_{j=1}^{n}\frac{H_{j,2}}{(j+2)^{2}}
\\ \stackrel{Lemmata~\ref{lemSumHi2i1i2},\ref{lemSumHi2i1sq},\ref{lemSumHi2i2sq}}{=}
 \frac{1}{2}\left(H_{n+1,2}^{2}-H_{n+1,4}\right)
-2\frac{n+1}{n+2}\left(H_{n+1,2}-1\right)
\\ +\frac{1}{2}\left(H_{n+2,2}^{2}-H_{n+2,4}\right)-2H_{n+2,2} 
 + 3 - \frac{2n+3}{(n+2)^{2}}
\\ =
\frac{1}{2}\left(H_{n+1,2}^{2}+H_{n+2,2}^{2}-H_{n+1,4}-H_{n+2,4}\right)
-2H_{n+1,2}-2H_{n+2,2}+ 5
\\ +\frac{2H_{n+1,2}}{n+2}-\frac{2}{n+2}- \frac{2n+3}{(n+2)^{2}}
\\ =
\frac{1}{2}\left(2H_{n+1,2}^{2}+\frac{2H_{n+1,2}}{(n+2)^{2}}+\frac{1}{(n+2)^{4}}-2H_{n+1,4}-\frac{1}{(n+2)^{4}}\right)
-4H_{n+1,2}-\frac{2}{(n+2)^{2}}
\\ +5
+\frac{2H_{n+1,2}}{n+2}-\frac{2}{n+2}- \frac{2n+3}{(n+2)^{2}}
\\ =H_{n+1,2}^{2}-H_{n+1,4} -4H_{n+1,2}  +5
+ \frac{H_{n+1,2}}{(n+2)^{2}} +\frac{2H_{n+1,2}}{n+2}-\frac{2}{n+2}- \frac{2n+5}{(n+2)^{2}}
\\ \xrightarrow{n\to \infty} 
\frac{\pi^{4}}{36}-\frac{\pi^{4}}{90}-\frac{2\pi^{2}}{3}+5 
=\frac{\pi^{4}}{60}-\frac{2\pi^{2}}{3}+5 \approx 0.044.
\\
\end{array}
$$
\end{proof}

\begin{lemma}[\citepos{KBarSSag2015bart} p. $74$]\label{lemSumHisqi1i2}

\be\label{eqSumHisqi1i2}
V_{n,1,1}^{1,2,1}=
\sum\limits_{j=1}^{n}\frac{H_{j,1}^{2}}{(j+1)(j+2)} = 
H_{n,2} +1 - \frac{2H_{n,1}}{n+1}-\frac{1}{n+1} -\frac{H_{n,1}^{2}}{n+2} 
\xrightarrow{n\to \infty} \frac{\pi^{2}}{6} +1
\ee
\end{lemma}
\begin{proof}

$$
\begin{array}{l}
\sum\limits_{j=1}^{n}\frac{H_{j,1}^{2}}{(j+1)(j+2)} 
= \sum\limits_{j=1}^{n}\frac{H_{j,2}+2\sum\limits_{i_{2}=1}^{j-1}i_{2}^{-1}\sum\limits_{i_{1}=i_{2}+1}^{j}i_{1}^{-1}}{(j+1)(j+2)} 
=\sum\limits_{j=1}^{n}\frac{H_{j,2}+2\left(H_{j,1}(H_{j,1}-j^{-1})-\sum\limits_{i_{2}=1}^{j}i_{2}^{-1}H_{i_{2},1}+j^{-1}H_{j,1}\right)}{(j+1)(j+2)} 
\\=\sum\limits_{j=1}^{n}\frac{H_{j,2}}{(j+1)(j+2)}
+2\sum\limits_{j=1}^{n}\frac{H_{j,1}^{2}}{(j+1)(j+2)}
-2\sum\limits_{j=1}^{n}\frac{\sum\limits_{i_{2}=1}^{j}i_{2}^{-1}H_{i_{2},1}}{(j+1)(j+2)}.
\end{array}
$$
Taking over to the left hand-side we obtain that

$$
\begin{array}{l}
\sum\limits_{j=1}^{n}\frac{H_{j,1}^{2}}{(j+1)(j+2)} 
=
2\sum\limits_{j=1}^{n}\frac{\sum\limits_{i_{2}=1}^{j}i_{2}^{-1}H_{i_{2},1}}{(j+1)(j+2)}
-\sum\limits_{j=1}^{n}\frac{H_{j,2}}{(j+1)(j+2)}.
\end{array}
$$
We will use Lemma \ref{lemSum1} to simplify

$$
\begin{array}{l}
\sum\limits_{j=1}^{n}\frac{\sum\limits_{i=1}^{j}i^{-1}H_{i,1}}{(j+1)(j+2)} 
= \sum\limits_{i=1}^{n}i^{-1}H_{i,1}\sum\limits_{j=i}^{n}\frac{1}{(j+1)(j+2)} 
=\sum\limits_{i=1}^{n}\frac{H_{i,1}}{i(i+1)}
- \frac{1}{n+2}\sum\limits_{i=1}^{n}i^{-1}H_{i,1}
\end{array}
$$
obtaining

$$
\begin{array}{l}
\sum\limits_{j=1}^{n}\frac{H_{j,1}^{2}}{(j+1)(j+2)} 
=
2\sum\limits_{j=1}^{n}\frac{H_{j,1}}{j(j+1)}
- 2\frac{1}{n+2}\sum\limits_{j=1}^{n}\frac{H_{j,1}}{j}
-\sum\limits_{j=1}^{n}\frac{H_{j,2}}{(j+1)(j+2)}.
\end{array}
$$
We input formul\ae\  obtained in Lemmata 
\ref{lemSumHii1}, \ref{lemSumHiii1},  and \ref{lemSumHi2i1i2}
to obtain

$$
\begin{array}{l}
\sum\limits_{j=1}^{n}\frac{H_{j,1}^{2}}{(j+1)(j+2)} 
=
2\sum\limits_{j=1}^{n}\frac{H_{j,1}}{j(j+1)}
- 2\frac{1}{n+2}\sum\limits_{j=1}^{n}\frac{H_{j,1}}{j}
-\sum\limits_{j=1}^{n}\frac{H_{j,2}}{(j+1)(j+2)}
\\ =
2\left( H_{n+1,2} - \frac{H_{n+1,1}}{n+1} \right)
- 2\frac{1}{n+2}\frac{1}{2}\left(H_{n,1}^{2} +H_{n,2}\right)
-\frac{n+1}{n+2}\left(H_{n+1,2}-1\right) 
\\=
H_{n+1,2}
+\frac{n^{2}+n-1}{(n+1)^{2}}
- 2\frac{H_{n,1}}{n+1} -\frac{H_{n,1}^{2}}{n+2} 
\xrightarrow{n\to \infty} \frac{\pi^{2}}{6} +1 \approx 2.645.
\end{array}
$$
\end{proof}

\begin{remark}
\citet{ASofMHas2012} considered (their Corollary $2$) a similar infinite sum

$$
\sum\limits_{j=1}^{\infty}\frac{H_{j,1}^{2}}{j(j+1)} = 3\zeta(3),
$$
\citet{CXuMZhaWZhu2016}, in their Eq. $(2.39)$, considered 
$\sum_{j=1}^{\infty}H_{j,1}^{2}/(j(j+k))$, 
\citet{KCheYChe2020} in their Thm. $3.1$ consider the limit for a more general series,
and
\citet{XWanWChu2018}, in their Corollary $3$, present the 
limit of Eq. \eqref{eqSumHisqi1i2}. 
The closed form of the sum is presented, without proof, on \citepos{KBarSSag2015bart} p. $74$.
\end{remark}

\begin{lemma}[see also \citepos{KCheYChe2020} Thm. $3.1$ for the limit]\label{lemSumHi2sqi1i2}

\be\label{eqSumHi2sqi1i2}
\begin{array}{rcl}
V_{n,1,1}^{1,2,2}=
\sum\limits_{j=1}^{n}\frac{H_{j,2}^{2}}{(j+1)(j+2)} & = &
H_{n+1,2}^{2}  -H_{n+1,2}   -H_{n+1,3} 
+1 
\\ && -\frac{H_{n+2,2}^{2}}{n+2}+\frac{2H_{n+1,2}}{n+2} 
- \frac{1}{n+2} 
 +2\frac{H_{n+2,2}}{(n+2)^{3}}-\frac{1}{(n+2)^{5}}
\\ &  \xrightarrow{n\to \infty} &
\frac{\pi^{4}}{36}-\frac{\pi^{2}}{6}-\zeta(3)+1.
\end{array}	
\ee
\end{lemma}
\begin{proof}

$$
\begin{array}{l}
\sum\limits_{j=1}^{n}\frac{H_{j,2}^{2}}{(j+1)(j+2)} =
\sum\limits_{j=1}^{n}\frac{H_{j,2}^{2}}{j+1} -\sum\limits_{j=1}^{n}\frac{H_{j,2}^{2}}{j+2} 
\end{array}
$$

$$
\begin{array}{l}
 \stackrel{Lemmata~\ref{lemSumHi2sqi1app}, \ref{lemSumHi2sqi2}}{=}
\sum\limits_{j=1}^{n+1}\frac{H_{j,2}^{2}}{j} 
- 2\sum\limits_{j=1}^{n+1}\frac{H_{j,2}}{j^{3}} 
+H_{n+1,5}
\\ -
\left(
\sum\limits_{j=1}^{n+2}\frac{H_{j,2}^{2}}{j} 
-2\sum\limits_{j=1}^{n+2}\frac{H_{j,2}}{j^{3}}
-H_{n+1,2}^{2}  +H_{n+1,2}   +H_{n+1,3} +H_{n+2,5} 
-1 
\right. \\ \left. 
-\frac{2H_{n+1,2}}{n+2} + \frac{1}{n+2}
\right)
\\ =
\sum\limits_{j=1}^{n+1}\frac{H_{j,2}^{2}}{j} 
- 2\sum\limits_{j=1}^{n+1}\frac{H_{j,2}}{j^{3}} 
+H_{n+1,5}
-\sum\limits_{j=1}^{n+2}\frac{H_{j,2}^{2}}{j} 
+2\sum\limits_{j=1}^{n+2}\frac{H_{j,2}}{j^{3}}
+H_{n+1,2}^{2}  -H_{n+1,2} 
\\   -H_{n+1,3}  -H_{n+2,5} +1 
+\frac{2H_{n+1,2}}{n+2} - \frac{1}{n+2}
\\ =
H_{n+1,2}^{2}  -H_{n+1,2}   -H_{n+1,3} 
+1 
-\frac{H_{n+2,2}^{2}}{n+2}+\frac{2H_{n+1,2}}{n+2} 
- \frac{1}{n+2} 
 +2\frac{H_{n+2,2}}{(n+2)^{3}}-\frac{1}{(n+2)^{5}}
\\ \xrightarrow{n\to \infty}
\frac{\pi^{4}}{36}-\frac{\pi^{2}}{6}-\zeta(3)+1 \approx 0.859.
\end{array}
$$
\end{proof}

\begin{lemma}\label{lemSumHi2sqisq} 

\be\label{eqSumHi2sqisq}
\begin{array}{rcl}
V_{n,2,0}^{0,0,1}=
\sum\limits_{j=1}^{n}\frac{H_{j,2}^{2}}{j^{2}}  
& = &  
\frac{1}{3}H_{n,2}^{3}  
+\sum\limits_{j=1}^{n}\frac{H_{j,2}}{j^{4}} 
- \frac{1}{3}H_{n,6}
\xrightarrow{n\to \infty}  
\zeta^{2}(3)+\frac{19\pi^{6}}{22680}
.
\end{array}
\ee
\end{lemma}
\begin{proof}
We will use Abel's summation by parts, \citep[see, e.g,][]{WChu2007} for two sequences $A_{j}$ 
and $B_{j}$

\be\label{eqAbelSumParts}
\sum\limits_{j=1}^{\infty}B_{j}\bigtriangledown A_{j+1} = B_{n}A_{n+1} -B_{1}A_{1}+\sum\limits_{j=2}^{n}A_{j}\bigtriangledown B_{j},
\ee
where $\bigtriangledown \tau_{j}:=\tau_{j}-\tau_{j-1}$. 
We take $A_{j}:=H_{j,2}$ and
$B_{j}:=H_{j,2}^{2}$. Then, $\bigtriangledown A_{j} = 1/j^{2}$ and 

\be\label{eqtridownBj}
\begin{array}{l}
\bigtriangledown B_{j} = H_{j,2}^{2}-H_{j-1,2}^{2} 
= \left(H_{j-1,2}+\frac{1}{j^{2}}\right)^{2}- H_{j-1,2}^{2}
= H_{j-1,2}^{2}-H_{j-1,2}^{2}-\frac{2H_{j-1,2}}{j^{2}}-\frac{1}{j^{4}} 
\\ = \frac{2H_{j-1,2}}{j^{2}}+\frac{1}{j^{4}}.
\end{array}
\ee
With the above we apply Eq. \eqref{eqAbelSumParts}

$$
\begin{array}{l}
\sum\limits_{j=1}^{n}\frac{H_{j,2}^{2}}{(j+1)^{2}}  
\stackrel{Eq. \eqref{eqtridownBj}}{=} 
H_{n,2}^{2}H_{n+1,2} - 1
-\sum\limits_{j=2}^{n}H_{j,2}\left(\frac{2H_{j-1,2}}{j^{2}}+\frac{1}{j^{4}}\right)
\\ 
\sum\limits_{j=1}^{n}\frac{\left(H_{j+1,2}-\frac{1}{(j+1)^{2}}\right)^{2}}{(j+1)^{2}}  
=
H_{n,2}^{2}H_{n+1,2} - 1
-\sum\limits_{j=2}^{n}\left(\frac{2H^{2}_{j,2}}{j^{2}}-\frac{2}{j^{4}}+\frac{1}{j^{4}}\right)
\\ 
\sum\limits_{j=1}^{n}\frac{H_{j+1,2}^{2}}{(j+1)^{2}}  
-2\sum\limits_{j=1}^{n}\frac{H_{j+1,2}}{(j+1)^{4}}  
+\sum\limits_{j=1}^{n}\frac{1}{(j+1)^{6}}  
=
H_{n,2}^{2}H_{n+1,2} - 1
-2\sum\limits_{j=2}^{n}\frac{H^{2}_{j,2}}{j^{2}}
+\sum\limits_{j=2}^{n}\frac{H_{j,2}}{j^{4}}
\\ 
\sum\limits_{j=1}^{n+1}\frac{H_{j,2}^{2}}{j^{2}}  -1
-2\sum\limits_{j=1}^{n+1}\frac{H_{j,2}}{j^{4}}  
 +2
+\sum\limits_{j=1}^{n+1}\frac{1}{j^{6}}  -1
 =
H_{n,2}^{3}+ \frac{H_{n,2}^{2}}{(n+1)^{2}} - 1
-2\sum\limits_{j=1}^{n}\frac{H^{2}_{j,2}}{j^{2}}
\\ +2
+\sum\limits_{j=1}^{n}\frac{H_{j,2}}{j^{4}} -1
\\
\sum\limits_{j=1}^{n}\frac{H_{j,2}^{2}}{j^{2}}  +\frac{H_{n+1,2}^{2}}{(n+1)^{2}}
=
H_{n,2}^{3}+ \frac{H_{n,2}^{2}}{(n+1)^{2}} 
-2\sum\limits_{j=1}^{n}\frac{H^{2}_{j,2}}{j^{2}}
+\sum\limits_{j=1}^{n}\frac{H_{j,2}}{j^{4}} 
+2\sum\limits_{j=1}^{n+1}\frac{H_{j,2}}{j^{4}}  
-\sum\limits_{j=1}^{n+1}\frac{1}{j^{6}}  
\\
\sum\limits_{j=1}^{n}\frac{H_{j,2}^{2}}{j^{2}}  
=
H_{n,2}^{3}+ \frac{H_{n,2}^{2}}{(n+1)^{2}} -\frac{H_{n+1,2}^{2}}{(n+1)^{2}}
-2\sum\limits_{j=1}^{n}\frac{H^{2}_{j,2}}{j^{2}}
+3\sum\limits_{j=1}^{n}\frac{H_{j,2}}{j^{4}} 
+\frac{2H_{n+1,2}}{(n+1)^{4}} - H_{n+1,6}
.
\end{array}
$$
Taking $2\sum\limits_{j=1}^{n}\frac{H^{2}_{j,2}}{j^{2}}$ to the other side we obtain

$$
\begin{array}{l}
3\sum\limits_{j=1}^{n}\frac{H_{j,2}^{2}}{j^{2}}  
=
H_{n,2}^{3}  -\frac{2H_{n,2}}{(n+1)^{4}}-\frac{1}{(n+1)^{6}}
+3\sum\limits_{j=1}^{n}\frac{H_{j,2}}{j^{4}} 
+\frac{2H_{n,2}}{(n+1)^{4}} +\frac{2}{(n+1)^{6}} 
- H_{n+1,6}
\\
\sum\limits_{j=1}^{n}\frac{H_{j,2}^{2}}{j^{2}}  
=
\frac{1}{3}H_{n,2}^{3}  
+\sum\limits_{j=1}^{n}\frac{H_{j,2}}{j^{4}} 
- \frac{1}{3}H_{n,6}  .
\end{array}
$$
Taking the limit $n\to \infty$ in the above we obtain

$$
\begin{array}{rcl}
\sum\limits_{j=1}^{\infty}\frac{H_{j,2}^{2}}{j^{2}}  
& = &
\frac{1}{3}\left(\frac{\pi^{2}}{6}\right)^{3}  
+ \zeta^{2}(3)-\frac{\pi^{6}}{2835}
- \frac{\pi^{6}}{3\cdot 945}
=
\frac{\pi^{6}}{3\cdot 216}
+ \zeta^{2}(3)-\frac{\pi^{6}}{2835}
- \frac{\pi^{6}}{2835}
\\ & = &
\zeta^{2}(3)+\pi^{6}\left(\frac{1}{648}-\frac{2}{2835} \right)
=
\zeta^{2}(3)+\frac{19\pi^{6}}{22680} \approx 2.250.
\end{array}
$$
\end{proof}

\section*{Software availability}
In \url{https://github.com/krzbar/HarmonicSums}  we provide a \proglang{Mathematica} script 
that contains all the code we used here.

\section*{Acknowledgments}
KB is supported by an ELLIIT Call C grant.

\input{Bartoszek_HarmonicQuadraticSums_Supp.tex}

\bibliographystyle{plainnat}
\bibliography{PCM_Mart}

\end{document}

%% file: Bartoszek_HarmonicQuadraticSums_Supp.tex
\renewcommand{\theequation}{H.\arabic{equation}}
\renewcommand{\thelemma}{H.\arabic{lemma}}
\renewcommand{\theremark}{H.\arabic{remark}}
\setcounter{equation}{0}
\setcounter{lemma}{0}
\setcounter{theorem}{0}

\newpage
\section*{Appendix: Harmonic sums lemmata}
\begin{lemma}\label{lemSumjmjr}
For $0\le m<r$ we have

\be\label{eqSumjmjr}
\begin{array}{rcl}
\sum\limits_{j=1}^{n}\frac{1}{(j+m)(j+r)} & = & \frac{1}{r-m}\left(H_{r,1}-H_{m,1}-\left(H_{n+r,1}-H_{n+m,1} \right) 
\right)\xlongrightarrow{n\to\infty}\frac{1}{r-m}\left(H_{r,1}-H_{m,1}\right)
.
\end{array}
\ee
\end{lemma}

\begin{proof}

$$
\begin{array}{l}
\sum\limits_{j=1}^{n}\frac{1}{(j+m)(j+r)} = \frac{1}{r-m}\left(\sum\limits_{j+1}^{n}\frac{1}{j+m}-\sum\limits_{j+1}^{n}\frac{1}{j+r} \right)
=\frac{1}{r-m}\left(H_{r,1}-H_{m,1}-\left(H_{n+r,1}-H_{n+m,1} \right) \right).
\end{array}
$$
\end{proof}

\begin{lemma}\label{lemSum0}
\be\label{eqSum0}
\sum\limits_{j=1}^{n}\frac{1}{j(j+1)} = 1-\frac{1}{n+1} \xrightarrow{n\to \infty} 1
\ee
\end{lemma}

\begin{proof}
We show this property 
by transforming it into a telescoping sum,

$$
\sum\limits_{j=1}^{n}\frac{1}{j(j+1)}
=\sum\limits_{j=1}^{n}\left(\frac{1}{j}-\frac{1}{j+1}\right)=1-\frac{1}{n+1}.
$$
\end{proof}

\begin{lemma}\label{lemSum02}
\be\label{eqSum02}
\sum\limits_{j=1}^{n}\frac{1}{j(j+2)} = \frac{3}{4} - \frac{1}{2(n+1)} - \frac{1}{2(n+2)} \xrightarrow{n\to \infty} \frac{3}{4}
\ee
\end{lemma}

\begin{proof}
$$
\begin{array}{l}
\sum\limits_{j=1}^{n}\frac{1}{j(j+2)}
=
\frac{1}{2}\left(
\sum\limits_{j=1}^{n}\frac{1}{j} - \sum\limits_{j=1}^{n}\frac{1}{j+2}
\right)
=
\frac{1}{2}\left(
\sum\limits_{j=1}^{n}\frac{1}{j} - \sum\limits_{j=1}^{n+2}\frac{1}{j} + 1 + \frac{1}{2}
\right)
\\ =
\frac{3}{4} - \frac{1}{2(n+1)} - \frac{1}{2(n+2)}.
\end{array}
$$
\end{proof}

\begin{lemma}\label{lemSum1}
\be\label{eqSum1}
\sum\limits_{j=1}^{n}\frac{1}{(j+1)(j+2)} = \frac{1}{2}-\frac{1}{n+2} \xrightarrow{n\to \infty} \frac{1}{2}.
\ee
\end{lemma}

\begin{proof}
We show this property 
by transforming it into a telescoping sum,

$$
\sum\limits_{j=1}^{n}\frac{1}{(j+1)(j+2)} = \sum\limits_{j=2}^{n+1}\frac{1}{j(j+1)}
=\sum\limits_{j=2}^{n+1}\left(\frac{1}{j}-\frac{1}{j+1}\right)=\frac{1}{2}-\frac{1}{n+2}.
$$
\end{proof}

\begin{lemma}\label{lemSumj2j3}

\be\label{eqSumj2j3}
\begin{array}{rcl}
\sum\limits_{j=1}^{n}\frac{1}{(j+2)(j+3)} & = & \frac{1}{3}-\frac{1}{n+3} 
\xlongrightarrow{n\to\infty} \frac{1}{3}.
\end{array}
\ee
\end{lemma}

\begin{proof}
Plugging in $m=2$ and $r=3$ in Lemma \ref{lemSumjmjr} we immediately obtain the result. 
\end{proof}

\begin{lemma}\label{lemSumj3j4}

\be\label{eqSumj3j4}
\begin{array}{rcl}
\sum\limits_{j=1}^{n}\frac{1}{(j+3)(j+4)} & = & \frac{1}{4}-\frac{1}{n+4} 
\xlongrightarrow{n\to\infty} \frac{1}{4}.
\end{array}
\ee
\end{lemma}

\begin{proof}
Plugging in $m=3$ and $r=4$ in Lemma \ref{lemSumjmjr} we immediately obtain the result. 
\end{proof}

\begin{lemma}\label{lemSumii1i2}

\be\label{eqSumii1i2}
\sum\limits_{j=1}^{n}\frac{1}{j(j+1)(j+2)} = \frac{1}{4} - \frac{1}{2(n+1)} + \frac{1}{2(n+2)}
\xrightarrow{n\to \infty} \frac{1}{4}
\ee
\end{lemma}

\begin{proof}

$$
\begin{array}{l}
\sum\limits_{j=1}^{n}\frac{1}{j(j+1)(j+2)} 
=
\sum\limits_{j=1}^{n}\frac{1}{j(j+1)}
-
\sum\limits_{j=1}^{n}\frac{1}{j(j+2)} 
\stackrel{Lemmata~\ref{lemSum0},\ref{lemSum02}}{=}
1-\frac{1}{n+1}
-\frac{3}{4}
\\ + \frac{1}{2(n+1)} + \frac{1}{2(n+2)}
 = \frac{1}{4} - \frac{1}{2(n+1)} + \frac{1}{2(n+2)} .
\end{array}
$$
\end{proof}

\begin{lemma}\label{lemSumj1j2j3}

\be\label{eqSumj1j2j3}
\begin{array}{rcl}
\sum\limits_{j=1}^{n}\frac{1}{(j+1)(j+2)(j+3)} & = &  \frac{1}{12}-\frac{1}{2(n+2)}+\frac{1}{2(n+3)}
\xlongrightarrow{n\to \infty} \frac{1}{12}.
\end{array}
\ee
\end{lemma}

\begin{proof}

$$
\begin{array}{l}
\sum\limits_{j=1}^{n}\frac{1}{(j+1)(j+2)(j+3)} =
\sum\limits_{j=2}^{n+1}\frac{1}{j(j+1)(j+2)} 
\stackrel{Lemma~\ref{lemSumii1i2}}{=}
\frac{1}{4}-\frac{1}{2(n+2)}+\frac{1}{2(n+3)}
-\frac{1}{6}
\\=\frac{1}{12}-\frac{1}{2(n+2)}+\frac{1}{2(n+3)}.
\end{array}
$$
\end{proof}

\begin{lemma}\label{lemSumj2j3j4}

\be\label{eqSumj2j3j4}
\begin{array}{rcl}
\sum\limits_{j=1}^{n}\frac{1}{(j+2)(j+3)(j+4)} & = &  \frac{1}{24}-\frac{1}{2(n+3)}+\frac{1}{2(n+4)}
\xlongrightarrow{n\to \infty} \frac{1}{24}
.
\end{array}
\ee
\end{lemma}

\begin{proof}

$$
\begin{array}{l}
\sum\limits_{j=1}^{n}\frac{1}{(j+2)(j+3)(j+4)} =
\sum\limits_{j=2}^{n+1}\frac{1}{(j+1)(j+2)(j+3)} 
\stackrel{Lemma~\ref{lemSumj1j2j3}}{=}
\frac{1}{12}-\frac{1}{2(n+3)}+\frac{1}{2(n+4)}
-\frac{1}{24}
\\=\frac{1}{24}-\frac{1}{2(n+3)}+\frac{1}{2(n+4)}.
\end{array}
$$
\end{proof}

\begin{lemma}\label{lemSumj1j2j3j4}

\be\label{eqSumj1j2j3j4}
\sum\limits_{j=1}^{n}\frac{1}{(j+1)(j+2)(j+3)(j+4)} = \frac{1}{72}-\frac{1}{6(n+2)}+\frac{1}{3(n+3)}-\frac{1}{6(n+4)}
\xlongrightarrow{n \to\infty} \frac{1}{72}.
\ee
\end{lemma}

\begin{proof}

$$
\begin{array}{l}
\sum\limits_{j=1}^{n}\frac{1}{(j+1)(j+2)(j+3)(j+4)} 
=
\frac{1}{6}\sum\limits_{j=1}^{n}\frac{1}{j+1}-\frac{1}{2}\sum\limits_{j=1}^{n}\frac{1}{j+2}+\frac{1}{2}\sum\limits_{j=1}^{n}\frac{1}{j+3}-\frac{1}{6}\sum\limits_{j=1}^{n}\frac{1}{j+4}
\\ =
\frac{1}{6}H_{n+1,1}-\frac{1}{6}-\frac{1}{2}H_{n+2,1}+\frac{1}{2}\frac{3}{2}+\frac{1}{2}H_{n+3,1}-\frac{1}{2}\frac{11}{6}
-\frac{1}{6}H_{n+4,1}+\frac{1}{6}\frac{25}{12}
\\ =
\frac{1}{72}-\frac{1}{6(n+2)}+\frac{1}{3(n+3)}-\frac{1}{6(n+4)}.
\end{array}
$$
\end{proof}

\begin{lemma}\label{lemSumii12}

\be\label{eqSumii12}
\sum\limits_{j=1}^{n}\frac{1}{j(j+1)^{2}} =  2-H_{n+1,2} - \frac{1}{n+1}
\xrightarrow{n\to \infty} 2-\frac{\pi^{2}}{6} .
\ee
\end{lemma}

\begin{proof}
We use Lemma \ref{lemSum0}

$$
\begin{array}{l}
\sum\limits_{j=1}^{n}\frac{1}{j(j+1)^{2}}
=
\sum\limits_{j=1}^{n}\frac{1}{j+1}\left(\frac{1}{j} -\frac{1}{j+1}\right)
=1-\frac{1}{n+1} -H_{n+1,2}+1
=2-H_{n+1,2} - \frac{1}{n+1}
\\ \xrightarrow{n\to \infty} 2-\frac{\pi^{2}}{6} \approx 0.355.
\end{array}
$$
\end{proof}

\begin{lemma}\label{lemSumii2sq}

\be\label{eqSumii2sq}
\sum\limits_{j=1}^{n}\frac{1}{j(j+2)^{2}} = 1 - \frac{1}{2}H_{n+2,2} - \frac{1}{4(n+1)} - \frac{1}{4(n+2)} 
\xrightarrow{n\to \infty} 1 -\frac{\pi^{2}}{12} .
\ee
\end{lemma}

\begin{proof}

$$
\begin{array}{l}
\sum\limits_{j=1}^{n}\frac{1}{j(j+2)^{2}}
=
\frac{1}{2}\sum\limits_{j=1}^{n}\frac{1}{j+2}\left(\frac{1}{j} -\frac{1}{j+2}\right)
\\ \stackrel{Lemma~ \ref{lemSum02}}{=}
\frac{1}{2}\left(\frac{3}{4} - \frac{1}{2(n+1)} - \frac{1}{2(n+2)} - H_{n+2,2} +1 + \frac{1}{4}\right)
\\ =
\frac{1}{2}\left(2 - H_{n+2,2} - \frac{1}{2(n+1)} - \frac{1}{2(n+2)} \right)
= 1 - \frac{1}{2}H_{n+2,2} - \frac{1}{4(n+1)} - \frac{1}{4(n+2)} 
\\ \xrightarrow{n\to \infty} 1 -\frac{\pi^{2}}{12} \approx 0.178.
\end{array}
$$
\end{proof}

\newpage
\begin{lemma}\label{lemSumj1sqjmapp}
For $m\ge 2$ it holds that 

\be\label{eqSumj1sqjmapp}
\begin{array}{rcl}
\sum\limits_{j=1}^{n}\frac{1}{(j+1)^{2}(j+m)} & = &
\left\{
\begin{array}{lc}
H_{n+1,3}-1,~~~~~~~~~~~~~~~~~~~~~~~~~~~~~~~~~~~~~~~~~~~~  m=1, \\
\frac{H_{n+1,2}}{m-1} - \frac{1}{m-1}
- \frac{H_{m,1}}{(m-1)^{2}}+ \frac{1}{(m-1)^{2}}+\frac{1}{(m-1)^{2}}\sum\limits_{j=n+2}^{n+m}\frac{1}{j}, m \ge 2.
\end{array}
\right.
\\ & \xlongrightarrow{n \to \infty} &
\left\{
\begin{array}{lc}
\zeta(3)-1, & m= 1,\\
\frac{\pi^{2}}{6(m-1)} - \frac{1}{m-1} 
- \frac{H_{m,1}}{(m-1)^{2}}+ \frac{1}{(m-1)^{2}}, & m \ge 2.
\end{array}
\right.
\end{array}
\ee
\end{lemma}

\begin{proof}
For $m=1$ the result is obvious and for $m\ge 2$ we have
$$
\begin{array}{l}
\sum\limits_{j=1}^{n}\frac{1}{(j+1)^{2}(j+m)} 
=\sum\limits_{j=1}^{n}\frac{1}{(j+1)}\frac{1}{m-1}\left(\frac{1}{j+1}-\frac{1}{j+m}\right)
\\=\frac{1}{m-1}\left(\sum\limits_{j=1}^{n}\frac{1}{(j+1)^{2}}-\sum\limits_{j=1}^{n}\frac{1}{(j+1)(j+m)}\right)
=\frac{1}{m-1}\left(H_{n+1,2}-1-\sum\limits_{j=1}^{n}\frac{1}{(j+1)(j+m)}\right) 
\\=\frac{1}{m-1}\left(H_{n+1,2}-1-\frac{1}{m-1}\left(\sum\limits_{j=1}^{n}\frac{1}{j+1}-\sum\limits_{j=1}^{n}\frac{1}{j+m}\right)\right) %OK
\\ =\frac{1}{m-1}\left(H_{n+1,2}-1-\frac{1}{m-1}\left(H_{n+1,1}-1-H_{n+m,1}+H_{m,1}\right)\right)
\\ =
\frac{H_{n+1,2}}{m-1} - \frac{1}{m-1}
- \frac{H_{m,1}}{(m-1)^{2}}+ \frac{1}{(m-1)^{2}}+\frac{1}{(m-1)^{2}}\sum\limits_{j=n+2}^{n+m}\frac{1}{j} .
\end{array}
$$
\end{proof}

\newpage
\begin{lemma}\label{lemSumi12i2}
\be\label{eqSumi12i2}
\sum\limits_{j=1}^{n}\frac{1}{(j+1)^{2}(j+2)} = H_{n+1,2} - \frac{3}{2} + \frac{1}{n+2}
\xrightarrow{n\to \infty} \frac{\pi^{2}}{6}- \frac{3}{2}.
\ee
\end{lemma}
\begin{proof}
We use Lemma \ref{lemSum1}

$$
\begin{array}{l}
\sum\limits_{j=1}^{n}\frac{1}{(j+1)^{2}(j+2)}
=
\sum\limits_{j=1}^{n}\frac{1}{j+1}\left(\frac{1}{j+1} -\frac{1}{j+2}\right)
=H_{n+1,2}-1-\frac{1}{2}+\frac{1}{n+2}
\\ =H_{n+1,2} - \frac{3n+4}{2(n+2)} \xrightarrow{n\to \infty} \frac{\pi^{2}}{6}- \frac{3}{2}
\approx 0.145.
\end{array}
$$
\end{proof}

\begin{lemma}\label{lemSumj1s1j3}

\be\label{eqSumj1s1j3}
\sum\limits_{j=1}^{n}\frac{1}{(j+1)^{2}(j+3)} =
\frac{1}{2}H_{n+1,2} -\frac{17}{24}
+\frac{1}{4(n+2)}+\frac{1}{4(n+3)} \xlongrightarrow{n \to \infty}
\frac{\pi^{2}}{12} -\frac{17}{24}
.
\ee
\end{lemma}

\begin{proof}
We directly apply Lemma \ref{lemSumj1sqjmapp} with $m=3$. 
\end{proof}

\begin{lemma}\label{lemSumi1sqi4}

\be\label{eqSumi1sqi4}
\sum\limits_{j=1}^{n}\frac{1}{(j+1)^{2}(j+4)} =
\frac{1}{3}H_{n+1,2} 
-\frac{49}{108}
+\frac{1}{9(n+2)}+\frac{1}{9(n+3)}+\frac{1}{9(n+4)} \xlongrightarrow{n \to \infty} \frac{\pi^{2}}{18}-\frac{49}{108}
.
\ee
\end{lemma}

\begin{proof}
We directly apply Lemma \ref{lemSumj1sqjmapp} with $m=3$. 
\end{proof}

\newpage 

\begin{lemma}\label{lemSumi1i2sq}

\be\label{eqSumi1i2sq}
\sum\limits_{j=1}^{n}\frac{1}{(j+1)(j+2)^{2}} = \frac{7}{4}-H_{n+2,2}-\frac{1}{n+2}
 \xrightarrow{n\to \infty} \frac{7}{4}-\frac{\pi^{2}}{6}.
\ee
\end{lemma}

\begin{proof}
We use Lemma \ref{lemSum1}

$$
\begin{array}{l}
\sum\limits_{j=1}^{n}\frac{1}{(j+1)(j+2)^{2}}
=
\sum\limits_{j=1}^{n}\frac{1}{j+2}\left(\frac{1}{j+1} -\frac{1}{j+2}\right)
=
\frac{1}{2}-\frac{1}{n+2}-H_{n+2,2}+1+\frac{1}{4}
\\ =\frac{7}{4}-H_{n+2,2}-\frac{1}{n+2}
\to \frac{7}{4}-\frac{\pi^{2}}{6}\approx 0.105.
\end{array}
$$
\end{proof}

\begin{lemma}\label{lemSumj2jm}

\be\label{eqSumj2jm}
\begin{array}{rcl}
\sum\limits_{j=1}^{n}\frac{1}{(j+2)^{2}(j+m)} & = &
\left\{
\begin{array}{cc}
\frac{7}{4}-H_{n+2,2}-\frac{1}{n+2}, & m=1, \\
H_{n+2,3}-\frac{9}{8}, & m=2, \\
\frac{H_{n+2,2}}{m-2}-\frac{5}{4(m-2)}
+\frac{3}{2(m-2)^{2}} -\frac{1}{(m-2)^{2}}H_{m,1}
+ \frac{1}{(m-2)^{2}}\sum\limits_{j=n+3}^{n+m}\frac{1}{j},
 & m\ge 3, 
\end{array}
\right.
\\
& \xlongrightarrow{n \to \infty} &
\left\{
\begin{array}{cc}
\frac{7}{4}-\frac{\pi^{2}}{6}, & m=1, \\
\zeta(3)-\frac{9}{8}, & m=2, \\
\frac{\pi^{2}}{6(m-2)}-\frac{5}{4(m-2)}
+\frac{3}{2(m-2)^{2}} -\frac{1}{(m-2)^{2}}H_{m,1},
 & m\ge 3.
\end{array}
\right.

\end{array}
\ee
\end{lemma}

\begin{proof}
For $m=1$ we have Lemma \ref{lemSumi1i2sq}, the result is obvious 
for $m=2$, 
and for $m\ge 3$  we have 

$$
\begin{array}{l}
\sum\limits_{j=1}^{n}\frac{1}{(j+2)^{2}(j+m)} 
=\sum\limits_{j=1}^{n}\frac{1}{(j+2)}\frac{1}{m-2}\left(\frac{1}{j+2}-\frac{1}{j+m}\right)
\\ =\frac{1}{m-2}\left(H_{n+2,2}-\frac{5}{4}\right)-\frac{1}{m-2}\sum\limits_{j=1}^{n}\frac{1}{(j+2)(j+m)}
=\frac{H_{n+2,2}}{m-2}-\frac{5}{4(m-2)}
-\frac{1}{(m-2)^{2}}\sum\limits_{j=1}^{n}\left(\frac{1}{j+2}-\frac{1}{j+m}\right)
\\ =\frac{H_{n+2,2}}{m-2}-\frac{5}{4(m-2)}
-\frac{1}{(m-2)^{2}}\sum\limits_{j=1}^{n}\frac{1}{j+2}
+ \frac{1}{(m-2)^{2}}\sum\limits_{j=1}^{n}\frac{1}{j+m}
\\ =
\frac{H_{n+2,2}}{m-2}-\frac{5}{4(m-2)}
-\frac{1}{(m-2)^{2}}H_{n+2,1}+\frac{1}{(m-2)^{2}}+\frac{1}{2(m-2)^{2}}
+ \frac{1}{(m-2)^{2}}H_{n+m,1}-\frac{1}{(m-2)^{2}}H_{m,1}
\\ =
\frac{H_{n+2,2}}{m-2}-\frac{5}{4(m-2)}
+\frac{3}{2(m-2)^{2}}
-\frac{1}{(m-2)^{2}}H_{m,1}
+ \frac{1}{(m-2)^{2}}\sum\limits_{j=n+3}^{n+m}\frac{1}{j}
\end{array}
$$
\end{proof}

\begin{lemma}\label{lemSumj2sqj4}

\be\label{eqSumj2sqj4}
\sum\limits_{j=1}^{n}\frac{1}{(j+2)^{2}(j+4)} =
\frac{1}{2}{H_{n+2,2}}
-\frac{37}{48}
+ \frac{1}{4(n+3)}+ \frac{1}{4(n+4)} \xlongrightarrow{n \to \infty} \frac{\pi^{2}}{12} -\frac{37}{48}
.
\ee
\end{lemma}

\begin{proof}
We directly apply Lemma \ref{lemSumj2jm} with $m=4$. 
\end{proof}

\begin{lemma}\label{lemSumi1i2cb}

\be\label{eqSumi1i2cb}
\sum\limits_{j=1}^{n}\frac{1}{(j+1)(j+2)^{3}} = \frac{23}{8}-H_{n+2,2}-H_{n+2,3}-\frac{1}{n+2}\xrightarrow{n\to \infty} 
\frac{23}{8}-\frac{\pi^{2}}{6}-\zeta(3).
\ee
\end{lemma}

\begin{proof}

$$
\begin{array}{l}
\sum\limits_{j=1}^{n}\frac{1}{(j+1)(j+2)^{3}}
=
\sum\limits_{j=1}^{n}\frac{1}{(j+2)^{2}}\left(\frac{1}{j+1} -\frac{1}{j+2}\right)
=\sum\limits_{j=1}^{n}\frac{1}{(j+1)(j+2)^{2}}
-\sum\limits_{j=1}^{n}\frac{1}{(j+2)^{3}}
\\ =
\sum\limits_{j=1}^{n}\frac{1}{(j+2)}\left(\frac{1}{j+1} -\frac{1}{j+2}\right)
-H_{n+2,3}+1+\frac{1}{8}
\\ =
\sum\limits_{j=1}^{n}\frac{1}{(j+1)(j+2)}
-\sum\limits_{j=1}^{n}\frac{1}{(j+2)^{2}}
-H_{n+2,3}+1+\frac{1}{8}
\\ \stackrel{Lemma~\ref{lemSum1}}{=}
\frac{1}{2}-\frac{1}{n+2}
-H_{n+2,2}+1+\frac{1}{4}-H_{n+2,3}+1+\frac{1}{8}
\\ =\frac{23}{8}-H_{n+2,2}-H_{n+2,3}-\frac{1}{n+2}
\xrightarrow{n\to \infty} \frac{23}{8}-\frac{\pi^{2}}{6}-\zeta(3)
\approx 0.028.
\end{array}
$$
\end{proof}

\newpage
\begin{lemma}\label{lemSumii1cb}

\be\label{eqSumii1cb}
\sum\limits_{j=1}^{n}\frac{1}{j(j+1)^{3}} = 
3-H_{n+1,2}-H_{n+1,3}-\frac{1}{n+1}\xrightarrow{n\to \infty} 
3-\frac{\pi^{2}}{6}-\zeta(3).
\ee
\end{lemma}

\begin{proof}

$$
\begin{array}{l}
\sum\limits_{j=1}^{n}\frac{1}{j(j+1)^{3}} 
= 
\frac{1}{8}+\sum\limits_{j=2}^{n}\frac{1}{j(j+1)^{3}} 
=
\frac{1}{8}+\sum\limits_{j=1}^{n-1}\frac{1}{(j+1)(j+2)^{3}} 
\\ \stackrel{Lemma~\ref{lemSumi1i2cb}}{=}
\frac{1}{8}+\frac{23}{8}-H_{n+1,2}-H_{n+1,3}-\frac{1}{n+1}
\\ =
3-H_{n+1,2}-H_{n+1,3}-\frac{1}{n+1}\xrightarrow{n\to \infty} 
3-\frac{\pi^{2}}{6}-\zeta(3) \approx 0.153.
\end{array}
$$
\end{proof}

\begin{lemma}\label{lemSumii2cb}

\be\label{eqSumii2cb}
\begin{array}{rcl}
\sum\limits_{j=1}^{n}\frac{1}{j(j+2)^{3}} & = &
\frac{1}{2}\left( 
\frac{17}{8} - \frac{1}{2}H_{n+2,2} -H_{n+2,3} - \frac{1}{4(n+1)} - \frac{1}{4(n+2)}  
\right)
\\ &\xrightarrow{n\to \infty}  &
\frac{17}{16} - \frac{\pi^{2}}{24} -\frac{1}{2}\zeta(3).
\end{array}
\ee
\end{lemma}

\begin{proof}

$$
\begin{array}{l}
\sum\limits_{j=1}^{n}\frac{1}{j(j+2)^{3}} 
= \frac{1}{2}\left( 
\sum\limits_{j=1}^{n}\frac{1}{j(j+2)^{2}} - \sum\limits_{j=1}^{n}\frac{1}{(j+2)^{3}} 
\right)
\\ \stackrel{Lemma~\ref{lemSumii2sq}}{=}
 \frac{1}{2}\left( 
1 - \frac{1}{2}H_{n+2,2} - \frac{1}{4(n+1)} - \frac{1}{4(n+2)}  
-H_{n+2,3} + 1 + \frac{1}{8}
\right)
\\ =
\frac{1}{2}\left( 
\frac{17}{8} - \frac{1}{2}H_{n+2,2} -H_{n+2,3} - \frac{1}{4(n+1)} - \frac{1}{4(n+2)}  
\right)
\\ \xrightarrow{n\to \infty} 
\frac{17}{16} - \frac{\pi^{2}}{24} -\frac{1}{2}\zeta(3) \approx 0.050.
\end{array}
$$
\end{proof}

\newpage
\begin{lemma}\label{lemSumj1sqj2sq}

\be\label{eqSumj1sqj2sq}
\sum\limits_{j=1}^{n}\frac{1}{(j+1)^{2}(j+2)^{2}} = 2H_{n+1,2}-\frac{13}{4}+\frac{2}{n+2}+\frac{1}{(n+2)^{2}} 
\xrightarrow{n\to \infty} \frac{\pi^{2}}{3}-\frac{13}{4}.
\ee
\end{lemma}

\begin{proof}
We first notice that $(j+1)^{-1}(j+2)^{-1}=(j+1)^{-1}-(j+2)^{-1}$, and then can write

$$
\begin{array}{l}
\sum\limits_{j=1}^{n}\left(\frac{1}{j+1}-\frac{1}{j+2}\right)^{2}
 =
\sum\limits_{j=1}^{n}\frac{1}{(j+1)^{2}}
+
\sum\limits_{j=1}^{n}\frac{1}{(j+2)^{2}}
-2\sum\limits_{j=1}^{n}\frac{1}{(j+1)(j+2)}
\\ \stackrel{Lemma~\ref{lemSum1}}{=}
H_{n+1,2}-1+H_{n+2,2}-1-\frac{1}{4}-2\left(\frac{1}{2}-\frac{1}{n+2}\right)
\\=
2H_{n+1,2}+\frac{2}{n+2}+\frac{1}{(n+2)^{2}}-\frac{13}{4}
\xrightarrow{n\to \infty} \frac{\pi^{2}}{3}-\frac{13}{4} \approx 0.04.
\end{array}
$$
\end{proof}

\begin{lemma}\label{lemSumj1quj2}

\be\label{eqSumj1quj2}
\begin{array}{rcl}
\sum\limits_{j=1}^{n}\frac{1}{(j+1)^{4}(j+2)}  & = &
H_{n+1,2}-H_{n+1,3}+H_{n+1,4}-\frac{3}{2}
+\frac{1}{n+2}
\\ &\xrightarrow{n\to \infty}  &
\frac{\pi^{4}}{90}-\zeta(3)+\frac{\pi^{2}}{6}-\frac{3}{2}.
\end{array}
\ee
\end{lemma}

\begin{proof}

$$
\begin{array}{l}
\sum\limits_{j=1}^{n}\frac{1}{(j+1)^{4}(j+2)} 
=\sum\limits_{j=1}^{n}\frac{1}{(j+1)^{3}}\left(\frac{1}{j+1}-\frac{1}{j+2} \right) 
=\sum\limits_{j=1}^{n}\frac{1}{(j+1)^{4}} -
\sum\limits_{j=1}^{n}\frac{1}{(j+1)^{2}}\left(\frac{1}{j+1}-\frac{1}{j+2} \right) 
\\ =H_{n+1,4}-1 -\sum\limits_{j=1}^{n}\frac{1}{(j+1)^{3}} +
\sum\limits_{j=1}^{n}\frac{1}{(j+1)}\left(\frac{1}{j+1}-\frac{1}{j+2} \right) 
\\ =H_{n+1,4}-1 -H_{n+1,3}+1
+\sum\limits_{j=1}^{n}\frac{1}{(j+1)^{2}} - 
\sum\limits_{j=1}^{n}\frac{1}{(j+1)(j+2)}
\\ =
H_{n+1,4}-H_{n+1,3} + H_{n+1,2}-1
-\sum\limits_{j=1}^{n}\frac{1}{j+1}
+\sum\limits_{j=1}^{n}\frac{1}{j+2}
\\=
H_{n+1,4}-H_{n+1,3} + H_{n+1,2}-1 -\frac{1}{2}+\frac{1}{n+2}
\\=
H_{n+1,4}-H_{n+1,3} + H_{n+1,2}-\frac{3}{2}+\frac{1}{n+2}
 \xrightarrow{n\to \infty} 
\frac{\pi^{4}}{90}-\zeta(3)+\frac{\pi^{2}}{6}-\frac{3}{2}
\approx 0.025.
\end{array}
$$
\end{proof}

\begin{lemma}\label{lemSumj1quj2sq}

\be\label{eqSumj1quj2sq}
\begin{array}{rcl}
\sum\limits_{j=1}^{n}\frac{1}{(j+1)^{4}(j+2)^{2}}  & = &
4H_{n+1,2}-2H_{n+1,3}+H_{n+1,4}-\frac{21}{4}
+\frac{4}{n+2}+\frac{1}{(n+2)^{2}} \\
& \xrightarrow{n\to \infty}  &
\frac{\pi^{4}}{90}-2\zeta(3)+\frac{2\pi^{2}}{3}-\frac{21}{4}.
\end{array}
\ee
\end{lemma}

\begin{proof}

$$
\begin{array}{l}
\sum\limits_{j=1}^{n}\frac{1}{(j+1)^{4}(j+2)^{2}}   =
\sum\limits_{j=1}^{n}\frac{1}{(j+1)^{2}}\left(\frac{1}{j+1}-\frac{1}{j+2}\right)^{2}   
\\ =
\sum\limits_{j=1}^{n}\frac{1}{(j+1)^{2}}\left(\frac{1}{(j+1)^{2}}-\frac{2}{(j+1)(j+2)}+\frac{1}{(j+2)^{2}}\right)   
\\ =
\sum\limits_{j=1}^{n}\frac{1}{(j+1)^{4}}
-2\sum\limits_{j=1}^{n}\frac{1}{(j+1)^{3}(j+2)}
\sum\limits_{j=1}^{n}\frac{1}{(j+1)^{2}(j+2)^{2}}
\\ \stackrel{Lemma~\ref{lemSumj1sqj2sq}}{=}
H_{n+1,4}-1
-2\sum\limits_{j=1}^{n}\frac{1}{(j+1)^{2}}\left(\frac{1}{j+1}-\frac{1}{j+2}\right)
+2H_{n+1,2}+\frac{2}{n+2}+\frac{1}{(n+2)^{2}}-\frac{13}{4} 
\\ =
H_{n+1,4}-1
-2\sum\limits_{j=1}^{n}\frac{1}{(j+1)^{3}}
+2\sum\limits_{j=1}^{n}\frac{1}{(j+1)^{2}(j+2)}
+2H_{n+1,2}+\frac{2}{n+2}+\frac{1}{(n+2)^{2}}-\frac{13}{4} 
\\ \stackrel{Lemma~\ref{lemSumi12i2}}{=}
H_{n+1,4}-1
-2H_{n+1,3}+2
+2H_{n+1,2} - \frac{3n+4}{n+2}
+2H_{n+1,2}+\frac{2}{n+2}+\frac{1}{(n+2)^{2}}-\frac{13}{4} 
\\ =
H_{n+1,4}-2H_{n+1,3}+4H_{n+1,2}-\frac{21}{4}
+\frac{4}{n+2}+\frac{1}{(n+2)^{2}}
\\ \xrightarrow{n\to \infty} \frac{\pi^{4}}{90}-2\zeta(3)+\frac{2\pi^{2}}{3}-\frac{21}{4}
\approx 0.008. 
\end{array}
$$
\end{proof}

\newpage
\begin{lemma}\label{lemSumj1sqj2qu}

\be\label{eqSumj1sqj2qu}
\begin{array}{rcl}
\sum\limits_{j=1}^{n}\frac{1}{(j+1)^{2}(j+2)^{4}}  & = &
4H_{n+1,2}+2H_{n+2,3}+H_{n+2,4}-\frac{161}{16}
+\frac{4}{n+2}
+\frac{3}{(n+2)^{2}}
\\ &\xrightarrow{n\to \infty}  &
\frac{\pi^{4}}{90}+2\zeta(3)+\frac{2\pi^{2}}{3}-\frac{161}{16}.
\end{array}
\ee
\end{lemma}

\begin{proof}

$$
\begin{array}{l}
\sum\limits_{j=1}^{n}\frac{1}{(j+1)^{2}(j+2)^{4}}   =
\sum\limits_{j=1}^{n}\frac{1}{(j+2)^{2}}\left(\frac{1}{j+1}-\frac{1}{j+2}\right)^{2}   
 =
\sum\limits_{j=1}^{n}\frac{1}{(j+2)^{2}}\left(\frac{1}{(j+1)^{2}}-\frac{2}{(j+1)(j+2)}+\frac{1}{(j+2)^{2}}\right)   
\\=
\sum\limits_{j=1}^{n}\frac{1}{(j+2)^{2}(j+1)^{2}}
-2\sum\limits_{j=1}^{n}\frac{1}{(j+1)(j+2)^{3}}
+\sum\limits_{j=1}^{n}\frac{1}{(j+2)^{4}}
\\ \stackrel{Lemma~\ref{lemSumj1sqj2sq}}{=}
2H_{n+1,2}+\frac{2}{n+2}+\frac{1}{(n+2)^{2}}-\frac{13}{4} 
-2\sum\limits_{j=1}^{n}\frac{1}{(j+2)^{2}}\left(\frac{1}{j+1}-\frac{1}{j+2}\right)
+H_{n+2,4}-1-\frac{1}{16}
\\ =
2H_{n+1,2}+\frac{2}{n+2}+\frac{1}{(n+2)^{2}}-\frac{13}{4} 
-2\sum\limits_{j=1}^{n}\frac{1}{(j+1)(j+2)^{2}}
+2\sum\limits_{j=1}^{n}\frac{1}{(j+2)^{3}}
+H_{n+2,4}-1-\frac{1}{16}
\end{array}
$$

$$
\begin{array}{l}
\stackrel{Lemma~\ref{lemSumi1i2sq}}{=}
2H_{n+1,2}+\frac{2}{n+2}+\frac{1}{(n+2)^{2}}-\frac{13}{4} 
-2\left(\frac{7}{4}-H_{n+2,2}-\frac{1}{n+2}\right)
+2H_{n+2,3}-2-\frac{2}{8}
\\ +H_{n+2,4}-1-\frac{1}{16}
 =
H_{n+2,4}
+2H_{n+2,3}
+4H_{n+1,2}
-\frac{161}{16}
+\frac{4}{n+2}
+\frac{3}{(n+2)^{2}}
\\ \xrightarrow{n\to \infty} 
\frac{\pi^{4}}{90}+2\zeta(3)+\frac{2\pi^{2}}{3}-\frac{161}{16} \approx 0.004. 
\end{array}
$$
\end{proof}

\begin{lemma}\label{lemSumii1sqi2}

\be\label{eqSumii1sqi2}
\sum\limits_{j=1}^{n}\frac{1}{j(j+1)^{2}(j+2)} = 
\frac{7}{4}-H_{n+1,2} - \frac{1}{2(n+1)} - \frac{1}{2(n+2)}
\xrightarrow{n\to \infty} \frac{7}{4}-\frac{\pi^{2}}{6}.
\ee
\end{lemma}

\begin{proof}

$$
\begin{array}{l}
\sum\limits_{j=1}^{n}\frac{1}{j(j+1)^{2}(j+2)} 
=
\sum\limits_{j=1}^{n}\frac{1}{j(j+1)^{2}} 
-\sum\limits_{j=1}^{n}\frac{1}{j(j+1)(j+2)} 
\\ \stackrel{Lemmata~\ref{lemSumii1i2},\ref{lemSumii12}}{=}
2-H_{n+1,2} - \frac{1}{n+1}
-\frac{1}{4} + \frac{1}{2(n+1)} - \frac{1}{2(n+2)}
\\ = \frac{7}{4}-H_{n+1,2} - \frac{1}{2(n+1)} - \frac{1}{2(n+2)}
\xrightarrow{n\to \infty} \frac{7}{4}-\frac{\pi^{2}}{6}
\approx 0.105.
\end{array}
$$
\end{proof}

\begin{lemma}\label{lemSumii1i2sq}

\be\label{eqSumii1i2sq}
\sum\limits_{j=1}^{n}\frac{1}{j(j+1)(j+2)^{2}} = \frac{1}{2}H_{n+2,2} - \frac{3}{4} - \frac{1}{4(n+1)} + \frac{3}{4(n+2)}
\xrightarrow{n\to \infty} \frac{\pi^{2}}{12}-\frac{3}{4}.
\ee
\end{lemma}

\begin{proof}

$$
\begin{array}{l}
\sum\limits_{j=1}^{n}\frac{1}{j(j+1)(j+2)^{2}} 
=\frac{1}{2}\left(
\sum\limits_{j=1}^{n}\frac{1}{j(j+1)(j+2)} 
- \sum\limits_{j=1}^{n}\frac{1}{(j+1)(j+2)^{2}} \right)
\\ \stackrel{Lemmata~\ref{lemSumii1i2},\ref{lemSumi1i2sq}}{=}
\frac{1}{2}\left(
\frac{1}{4} - \frac{1}{2(n+1)} + \frac{1}{2(n+2)}
- \frac{7}{4} + H_{n+2,2} + \frac{1}{n+2} \right)
\\ = \frac{1}{2}H_{n+2,2} - \frac{3}{4} - \frac{1}{4(n+1)} + \frac{3}{4(n+2)}
\xrightarrow{n\to \infty} \frac{\pi^{2}}{12}-\frac{3}{4} \approx 0.072.
\end{array}
$$
\end{proof}

\begin{lemma}\label{lemSumHj1}

\be\label{eqSumHj1}
G_{n,0,0}^{0,0,1}=\sum\limits_{j=1}^{n}H_{j,1} = (n+1)H_{n,1}-n \sim n\ln n.
\ee
\end{lemma}

\begin{proof}
We can write 

$$
\begin{array}{l}
\sum\limits_{j=1}^{n}H_{j,1} = 
\sum\limits_{j=1}^{n}\frac{n-j+1}{j}
=(n+1)\sum\limits_{j=1}^{n}\frac{1}{j}-\sum\limits_{j=1}^{n}1
=(n+1)H_{n,1}-n.
\end{array}
$$
\end{proof}

\begin{lemma}[\citepos{VAda1997} Lemma $1$ and \citepos{HAlzDKarHSri2006} Eq. $3.62$]\label{lemSumHii1}

\be\label{eqSumHii1}
G_{n,1,0}^{0,0,1}=\sum\limits_{j=1}^{n}\frac{H_{j,1}}{j} = \frac{1}{2}\left(H_{n,1}^{2} +H_{n,2}\right) \sim \frac{1}{2}\ln^{2} n.
\ee
\end{lemma}

\begin{proof}
We begin

$$
\begin{array}{l}
\sum\limits_{j=1}^{n}\frac{H_{j,1}}{j} 
= \sum\limits_{j=1}^{n}\frac{1}{j}\sum\limits_{k=1}^{j}\frac{1}{k}
= \sum\limits_{k=1}^{n}\frac{1}{k}\sum\limits_{j=k}^{n}\frac{1}{j}
= \sum\limits_{k=1}^{n}\frac{1}{k}\left(\sum\limits_{j=1}^{n}\frac{1}{j}-\sum\limits_{j=1}^{k-1}\frac{1}{j}\right)
\\= H_{n,1}^{2}  - \sum\limits_{k=1}^{n}\frac{H_{k,1}}{k} + H_{n,2}
\end{array}
$$
and carrying $\sum\limits_{k=1}^{n}\frac{H_{k,1}}{k}$ over to the LHS we obtain

$$
2\sum\limits_{j=1}^{n}\frac{H_{j,1}}{j} = H_{n,1}^{2}  + H_{n,2}
$$
which gives the desired result.
\end{proof}

\begin{lemma}\label{lemSumHj1jmapp}
For $m>1$ we have the following recursive representation 

\be\label{eqSumHj1jmapp}
G_{n,1,0}^{m,0,1}=\sum\limits_{j=1}^{n}\frac{H_{j,1}}{j+m} = 
G_{n,1,0}^{m-1,0,1}
+\frac{1}{m-1}\sum\limits_{j=n+2}^{n+m}\frac{1}{j}
-\frac{H_{m,1}}{m-1}
+\frac{1}{m(m-1)}
+\frac{H_{n+1,1}}{n+m},
\ee
with initial condition $m=1$ 

\be\label{eqSumHj1jqapp}
G_{n,1,0}^{1,0,1}=\sum\limits_{j=1}^{n}\frac{H_{j,1}}{j+1} 
= \frac{1}{2}\left(H_{n+1,1}^{2}-H_{n+1,2} \right)\sim \frac{1}{2}\ln^{2} n.
\ee
\end{lemma}

\begin{proof}
We first consider the special case $m=1$

$$
\begin{array}{l}
G_{n,1,0}^{1,0,1}=\sum\limits_{j=1}^{n}\frac{H_{j,1}}{j+1} 
= \sum\limits_{j=1}^{n}\frac{H_{j+1,1}}{j+1} - \sum\limits_{j=1}^{n}\frac{1}{(j+1)^{2}}
= \sum\limits_{j=1}^{n+1}\frac{H_{j,1}}{j}-1 - H_{n+1,2}+1
\\ \stackrel{Lemma~\ref{lemSumHii1}}{=}\frac{1}{2}\left(H_{n+1,1}^{2}+H_{n+1,2} \right)-H_{n+1,2}
=\frac{1}{2}\left(H_{n+1,1}^{2}-H_{n+1,2} \right).
\end{array}
$$
Now turning to the general case

$$
\begin{array}{l}
G_{n,1,0}^{m,0,1}=\sum\limits_{j=1}^{n}\frac{H_{j,1}}{j+m} = 
\sum\limits_{j=1}^{n}\frac{H_{j+1,1}}{j+m} -\sum\limits_{j=1}^{n}\frac{1}{(j+1)(j+m)}
\\ =\sum\limits_{j=1}^{n+1}\frac{H_{j,1}}{j+m-1}-\frac{1}{m} 
-\frac{1}{m-1}\left(
\sum\limits_{j=1}^{n}\frac{1}{j+1}- \sum\limits_{j=1}^{n}\frac{1}{j+m}
\right)
\\ =\sum\limits_{j=1}^{n+1}\frac{H_{j,1}}{j+m-1}-\frac{1}{m} 
-\frac{1}{m-1}\left(H_{n+1}-1-H_{n+m,1}+H_{m,1}
\right)
\\=\sum\limits_{j=1}^{n}\frac{H_{j,1}}{j+m-1}
+\frac{1}{m-1}\sum\limits_{j=n+2}^{n+m}\frac{1}{j}
-\frac{H_{m,1}}{m-1}
+\frac{1}{m(m-1)}
+\frac{H_{n+1,1}}{n+m}.
\end{array}
$$
\end{proof}

\begin{lemma}\label{lemSumHj1j2}

\be\label{eqSumHj1j2}
G_{n,1,0}^{2,0,1}=\sum\limits_{j=1}^{n}\frac{H_{j,1}}{j+2} = 
\frac{1}{2}\left(H_{n,1}^{2}-H_{n,2} \right)-1+\frac{H_{n,1}}{n+1}+\frac{H_{n,1}}{n+2}+\frac{1}{n+1}
\sim \frac{1}{2}\ln^{2}n.
\ee
\end{lemma}

\begin{proof}
We plug--in the initial condition from Lemma \ref{lemSumHj1jmapp}.
\end{proof}

\begin{lemma}\label{lemSumHj1j3}

\be\label{eqSumHj1j3}
\begin{array}{rcl}
G_{n,1,0}^{3,0,1}&=&\sum\limits_{j=1}^{n}\frac{H_{j,1}}{j+3} = 
\frac{1}{2}\left(H_{n,1}^{2}-H_{n,2} \right)-\frac{21}{12}+\frac{H_{n,1}}{n+1}+\frac{H_{n,1}}{n+2}+\frac{H_{n,1}}{n+3}\\ &&+\frac{3}{2(n+1)}+ \frac{1}{2(n+2)}
\sim \frac{1}{2}\ln^{2}n.
\end{array}
\ee
\end{lemma}

\begin{proof}
We plug--in the Lemma \ref{lemSumHj1j2} in the general formula from Lemma \ref{lemSumHj1jmapp}

$$
\begin{array}{l}
\sum\limits_{j=1}^{n}\frac{H_{j,1}}{j+3} 
=
\frac{1}{2}\left(H_{n+1,1}^{2}-H_{n+1,2} \right)-1+\frac{1}{n+2}+\frac{H_{n+1,1}}{n+2}
+\frac{1}{2(n+2)}+ \frac{1}{2(n+3)}
\\-\frac{H_{3,1}}{2}
+\frac{1}{6}
+\frac{H_{n+1,1}}{n+3}
=
\frac{1}{2}\left(H_{n+1,1}^{2}-H_{n+1,2} \right)-\frac{21}{12}
+\frac{H_{n+1,1}}{n+2}+\frac{H_{n+1,1}}{n+3}
\\+\frac{3}{2(n+2)}+ \frac{1}{2(n+3)}.
\end{array}
$$
\end{proof}

\begin{lemma}\label{lemSumHj1j4}

\be\label{eqSumHj1j4}
\begin{array}{rcl}
G_{n,1,0}^{4,0,1} &=&\sum\limits_{j=1}^{n}\frac{H_{j,1}}{j+4} = 
\frac{1}{2}\left(H_{n,1}^{2}-H_{n,2} \right) -\frac{85}{36}+\frac{H_{n,1}}{n+1}+\frac{H_{n,1}}{n+2}+\frac{H_{n,1}}{n+3}+\frac{H_{n,1}}{n+4}
\\&&+\frac{11}{6(n+1)}+ \frac{5}{6(n+2)}+\frac{1}{3(n+3)}
\sim \frac{1}{2}\ln^{2}n.
\end{array}
\ee
\end{lemma}

\begin{proof}
We plug--in the Lemma \ref{lemSumHj1j3} in the general formula from Lemma \ref{lemSumHj1jmapp}

$$
\begin{array}{l}
\sum\limits_{j=1}^{n}\frac{H_{j,1}}{j+4} 
=
\frac{1}{2}\left(H_{n+1,1}^{2}-H_{n+1,2} \right)-\frac{21}{12}
+\frac{H_{n+1,1}}{n+2}+\frac{H_{n+1,1}}{n+3}
+\frac{3}{2(n+2)}+ \frac{1}{2(n+3)}
\\ +\frac{1}{3(n+2)}+\frac{1}{3(n+3)}+\frac{1}{3(n+4)} 
-\frac{25}{36} 
+\frac{1}{12}
+\frac{H_{n+1,1}}{n+4}
\\=
\frac{1}{2}\left(H_{n+1,1}^{2}-H_{n+1,2} \right)
-\frac{85}{36}
+\frac{H_{n+1,1}}{n+2}+\frac{H_{n+1,1}}{n+3}+\frac{H_{n+1,1}}{n+4}
+\frac{11}{6(n+2)}+ \frac{5}{6(n+3)}+\frac{1}{3(n+4)}.
\end{array}
$$
\end{proof}

\begin{lemma}[\citepos{ASof2013} Eq. $1.2$]\label{lemSumHiii1}

\be\label{eqSumHiii1}
G_{n,1,1}^{0,1,1}=\sum\limits_{j=1}^{n}\frac{H_{j,1}}{j(j+1)} = H_{n,2} - \frac{H_{n,1}}{n+1} \xrightarrow{n\to \infty} \frac{\pi^{2}}{6}.
\ee
\end{lemma}

\begin{proof}

$$
\begin{array}{l}
\sum\limits_{j=1}^{n}\frac{H_{j,1}}{j(j+1)} 
= \sum\limits_{j=1}^{n}\frac{H_{j,1}}{j}-\sum\limits_{j=1}^{n}\frac{H_{j,1}}{j+1}
= \sum\limits_{j=1}^{n}\frac{H_{j,1}}{j}-\sum\limits_{j=1}^{n}\frac{H_{j+1,1}}{j+1}+\sum\limits_{j=1}^{n}\frac{1}{(j+1)^{2}}
\\ = \sum\limits_{j=1}^{n}\frac{H_{j,1}}{j}-\sum\limits_{j=1}^{n+1}\frac{H_{j,1}}{j}+1+\sum\limits_{j=1}^{n+1}\frac{1}{j^{2}}-1
= -\frac{H_{n+1,1}}{n+1}+H_{n+1,2} \xrightarrow{n\to \infty} \frac{\pi^{2}}{6} \approx 1.645.
\end{array}
$$
\end{proof}

\newpage
\begin{lemma}[\citepos{ASof2013} Lemma $1.1$]\label{lemSumHiii2}

\be\label{eqSumHiii2}
G_{n,1,1}^{0,2,1}=\sum\limits_{j=1}^{n}\frac{H_{j,1}}{j(j+2)} =
\frac{1}{2}\left(H_{n+1,2} + 1 -\frac{H_{n,1}}{n+1} -\frac{H_{n,1}}{n+2} -\frac{1}{n+1}  \right)
\xrightarrow{n\to \infty} \frac{1}{2}+\frac{\pi^{2}}{12}.
\ee
\end{lemma}

\begin{proof}

$$
\begin{array}{l}
\sum\limits_{j=1}^{n}\frac{H_{j,1}}{j(j+2)} 
= 
\frac{1}{2}\left(\sum\limits_{j=1}^{n}\frac{H_{j,1}}{j} - \sum\limits_{j=1}^{n}\frac{H_{j,1}}{j+2} \right)
\\ =
\frac{1}{2}\left(\sum\limits_{j=1}^{n}\frac{H_{j,1}}{j} - \sum\limits_{j=1}^{n}\frac{H_{j+2,1}}{j+2} 
+ \sum\limits_{j=1}^{n}\frac{1}{(j+1)(j+2)} + \sum\limits_{j=1}^{n}\frac{1}{(j+2)^{2}} \right)
\\ \stackrel{Lemma~\ref{lemSum1}}{=}
\frac{1}{2}\left(\sum\limits_{j=1}^{n}\frac{H_{j,1}}{j} - \sum\limits_{j=1}^{n+2}\frac{H_{j,1}}{j} +1 + \frac{3}{4}
+ \frac{1}{2}-\frac{1}{n+2}
+ H_{n+2,2}-1-\frac{1}{4}  \right)
\\ =
\frac{1}{2}\left(1 + H_{n+2,2} -\frac{H_{n+1,1}}{n+1} -\frac{H_{n+2,1}}{n+2} -\frac{1}{n+2}  \right)
\xrightarrow{n\to \infty} \frac{1}{2}+\frac{\pi^{2}}{12} 
\approx 1.322.
\end{array}
$$
\end{proof}

\begin{lemma}[\citepos{ASof2013} Eq. $1.4$]\label{lemSumHii1i2}

\be\label{eqSumHii1i2}
G_{n,1,1}^{1,2,1}=\sum\limits_{j=1}^{n}\frac{H_{j,1}}{(j+1)(j+2)} = 1  - \frac{H_{n,1}}{n+2}-\frac{1}{n+1} \xrightarrow{n\to \infty} 1
\ee
\end{lemma}

\begin{proof}

$$
\begin{array}{l}
\sum\limits_{j=1}^{n}\frac{H_{j,1}}{(j+1)(j+2)} =
\sum\limits_{j=1}^{n}\frac{H_{j,1}}{j+1} - \sum\limits_{j=1}^{n}\frac{H_{j,1}}{j+2} =
\sum\limits_{j=1}^{n}\frac{H_{j+1,1}}{j+1} -\sum\limits_{j=1}^{n}\frac{1}{(j+1)^{2}} 
- \sum\limits_{j=1}^{n}\frac{H_{j+2,1}}{j+2} 
\\ + \sum\limits_{j=1}^{n}\frac{1}{(j+1)(j+2)} 
+\sum\limits_{j=1}^{n}\frac{1}{(j+2)^{2}} 
 \\ 
 \stackrel{Lemma~\ref{lemSum1}}{=}
\sum\limits_{j=1}^{n+1}\frac{H_{j,1}}{j} -1
-\sum\limits_{j=1}^{n+1}\frac{1}{j^{2}} +1
- \sum\limits_{j=1}^{n+2}\frac{H_{j,1}}{j} 
 +1 +\frac{3}{4}
+ \frac{1}{2}-\frac{1}{n+2}
\\+\sum\limits_{j=1}^{n+2}\frac{1}{j^{2}} -1 - \frac{1}{4}
 =
1
-\frac{H_{n+2,1}}{n+2}
-\frac{1}{n+2}+\frac{1}{(n+2)^{2}}
 =
1 -\frac{H_{n+1,1}}{n+2} -\frac{1}{n+2}.
\end{array}
$$
\end{proof}

\begin{lemma}\label{lemSumHj1j1j3}

\be\label{eqSumHj1j1j3}
\begin{array}{l}
G_{n,1,1}^{1,3,1}=\sum\limits_{j=1}^{n}\frac{H_{j,1}}{(j+1)(j+3)} = 
\frac{7}{8}
-\frac{H_{n,1}}{2(n+2)}-\frac{H_{n,1}}{2(n+3)}
-\frac{3}{4(n+1)}-\frac{1}{4(n+2)}
\xlongrightarrow{n\to \infty}\frac{7}{8}.
\end{array}
\ee
\end{lemma}

\begin{proof}

$$
\begin{array}{l}
\sum\limits_{j=1}^{n}\frac{H_{j,1}}{(j+1)(j+3)} =
\frac{1}{2}\left(
\sum\limits_{j=1}^{n}\frac{H_{j,1}}{j+1}  - \sum\limits_{j=1}^{n}\frac{H_{j,1}}{j+3} 
\right)
=
\frac{1}{2}\left(
\sum\limits_{j=1}^{n}\frac{H_{j+1,1}}{j+1}-\sum\limits_{j=1}^{n}\frac{1}{(j+1)^{2}}  
-\sum\limits_{j=1}^{n}\frac{H_{j+3,1}}{j+3} 
\right. \\ \left. 
+\sum\limits_{j=1}^{n}\frac{1}{(j+1)(j+3)} 
+\sum\limits_{j=1}^{n}\frac{1}{(j+2)(j+3)} 
+\sum\limits_{j=1}^{n}\frac{1}{(j+3)^{2}} 
\right)
=
\frac{1}{2}\left(
\sum\limits_{j=1}^{n+1}\frac{H_{j,1}}{j}-1-H_{n+1,2}+1
- \sum\limits_{j=1}^{n+3}\frac{H_{j,1}}{j} 
\right. \\ \left. 
+1+\frac{3}{4}+\frac{11}{18}
+\frac{1}{2}\left(H_{n+1,1}-1-H_{n+3}+1+\frac{1}{2}+\frac{1}{3}\right)
+H_{n+2,1}-1-\frac{1}{2}-H_{n+3,1}
\right. \\ \left. 
+1+\frac{1}{2}+\frac{1}{3}
+H_{n+3,2}-1-\frac{1}{4}-\frac{1}{9}
\right)
\\=
\frac{7}{8}
-\frac{H_{n+2,1}}{2(n+2)}-\frac{H_{n+3,1}}{2(n+3)}
-\frac{1}{4(n+2)}-\frac{3}{4(n+3)}
+\frac{1}{2(n+2)^{2}}+\frac{1}{2(n+3)^{2}}
.
\end{array}
$$
\end{proof}

\begin{lemma}\label{lemSumHj1j1j4}

\be\label{eqSumHj1j1j4}
\begin{array}{rcl}
G_{n,1,1}^{1,4,1}&=&\sum\limits_{j=1}^{n}\frac{H_{j,1}}{(j+1)(j+4)}  = 
\frac{85}{108}
-\frac{H_{n,1}}{3(n+2)}-\frac{H_{n,1}}{3(n+3)}-\frac{H_{n,1}}{3(n+4)}
-\frac{11}{18(n+1)}-\frac{5}{18(n+2)}-\frac{1}{9(n+3)}
\\& \xlongrightarrow{n\to \infty} &\frac{85}{108}.
\end{array}
\ee
\end{lemma}

\begin{proof}

$$
\begin{array}{l}
\sum\limits_{j=1}^{n}\frac{H_{j,1}}{(j+1)(j+4)} =
\frac{1}{3}\left(
\sum\limits_{j=1}^{n}\frac{H_{j,1}}{j+1}
-
\sum\limits_{j=1}^{n}\frac{H_{j,1}}{j+4} 
\right)
\\ =
\frac{1}{3}\left(
\sum\limits_{j=1}^{n}\frac{H_{j+1,1}}{j+1}-\sum\limits_{j=1}^{n}\frac{1}{(j+1)^{2}}  
-\sum\limits_{j=1}^{n}\frac{H_{j+4,1}}{j+4} 
%\right. \\ \left. 
+\sum\limits_{j=1}^{n}\frac{1}{(j+1)(j+4)} 
+\sum\limits_{j=1}^{n}\frac{1}{(j+2)(j+4)} 
\right. \\ \left. 
+\sum\limits_{j=1}^{n}\frac{1}{(j+3)(j+4)} 
+\sum\limits_{j=1}^{n}\frac{1}{(j+4)^{2}} 
\right)
=
\frac{1}{3}\left(
-1-H_{n+1,2}+1
+1+\frac{3}{4}+\frac{11}{18}+\frac{25}{48}
\right. \\ \left. 
-\frac{H_{n+2,1}}{n+2}-\frac{H_{n+3,1}}{n+3}-\frac{H_{n+4,1}}{n+4}
%\right. \\ \left. 
+\frac{1}{3}\left(
H_{n+1,1}-1-H_{n+4,1}-1-\frac{1}{2}-\frac{1}{3}-\frac{1}{4}
\right)
\right. \\ \left. 
+\frac{1}{2}\left(
H_{n+2,1}-1-\frac{1}{2}-H_{n+4,1}-1-\frac{1}{2}-\frac{1}{3}-\frac{1}{4}
\right)
\right. \\ \left. 
+H_{n+3,1}-1-\frac{1}{2}-\frac{1}{3}-H_{n+4,1}+1+\frac{1}{2}+\frac{1}{3}+\frac{1}{4}
+H_{n+4,2}-1-\frac{1}{4}-\frac{1}{9}-\frac{1}{16}
\right)
\end{array}
$$

$$
\begin{array}{l}
%\\
 =
\frac{1}{3}\left(
+\frac{11}{18}+\frac{25}{48}
-\frac{H_{n+2,1}}{n+2}-\frac{H_{n+3,1}}{n+3}-\frac{H_{n+4,1}}{n+4}
-\frac{1}{3(n+2)}-\frac{1}{3(n+3)}-\frac{1}{3(n+3)}
\right. \\ \left. 
+\frac{1}{6}+\frac{1}{9}+\frac{1}{12}
-\frac{1}{2(n+3)}-\frac{1}{2(n+4)}
+\frac{1}{6}+\frac{1}{8}
-\frac{1}{n+4}
+\frac{1}{2}+\frac{1}{4}
-\frac{1}{9}-\frac{1}{16}
\right. \\ \left. 
+\frac{1}{(n+2)^{2}}+\frac{1}{(n+3)^{2}}+\frac{1}{(n+4)^{2}}
\right)
 =
\frac{1}{3}\left(
+\frac{85}{36}
-\frac{H_{n+2,1}}{n+2}-\frac{H_{n+3,1}}{n+3}-\frac{H_{n+4,1}}{n+4}
\right. \\ \left. 
-\frac{1}{3(n+2)}-\frac{5}{6(n+3)}-\frac{11}{6(n+4)}
+\frac{1}{(n+2)^{2}}+\frac{1}{(n+3)^{2}}+\frac{1}{(n+4)^{2}}
\right)
.
\end{array}
$$
\end{proof}

\begin{lemma}\label{lemSumHj1j2j4}

\be\label{eqSumHj1j2j4}
\begin{array}{rcl}
G_{n,1,1}^{2,4,1}&=&\sum\limits_{j=1}^{n}\frac{H_{j,1}}{(j+2)(j+4)}  = 
\frac{49}{72}
-\frac{H_{n,1}}{2(n+3)}-\frac{H_{n,1}}{2(n+4)}
-\frac{5}{12(n+1)}-\frac{5}{12(n+2)}-\frac{1}{6(n+3)}
\\& 
\xlongrightarrow{n \to \infty}& \frac{49}{72}.
\end{array}
\ee
\end{lemma}
\begin{proof}

$$
\begin{array}{l}
\sum\limits_{j=1}^{n}\frac{H_{j,1}}{(j+2)(j+4)} 
=
\frac{1}{2}\left(
\sum\limits_{j=1}^{n}\frac{H_{j,1}}{j+2} 
-
\sum\limits_{j=1}^{n}\frac{H_{j,1}}{j+4}
\right)
\\ =
\frac{1}{2}\left(
\sum\limits_{j=1}^{n}\frac{H_{j+2,1}}{j+2} 
-\sum\limits_{j=1}^{n}\frac{1}{(j+1)(j+2)} 
-\sum\limits_{j=1}^{n}\frac{1}{(j+2)^{2}} 
-
\sum\limits_{j=1}^{n}\frac{H_{j+4,1}}{j+4}
+\sum\limits_{j=1}^{n}\frac{1}{(j+1)(j+4)}
\right. \\ \left.
+\sum\limits_{j=1}^{n}\frac{1}{(j+2)(j+4)}
+\sum\limits_{j=1}^{n}\frac{1}{(j+3)(j+4)}
+\sum\limits_{j=1}^{n}\frac{1}{(j+4)^{2}}
\right)
\end{array}
$$

$$
\begin{array}{l}
\stackrel{Lemma~ \ref{lemSum1}}{=}
\frac{1}{2}\left(
-\frac{1}{2}+\frac{1}{n+2}
-H_{n+2,2}
%\right. \\ \left.
+1+\frac{1}{4}
+\frac{11}{18}+\frac{25}{48}
-\frac{H_{n+3}}{n+3}-\frac{H_{n+4}}{n+4}
\right. \\ \left.
+\frac{1}{3}\left(H_{n+1,1}-1-H_{n+4,1}+1+\frac{1}{2}+\frac{1}{3}+\frac{1}{4}\right)
\right. \\ \left.
+\frac{1}{2}\left(H_{n+2,1}-1-\frac{1}{2}-H_{n+4,1}+1+\frac{1}{2}+\frac{1}{3}+\frac{1}{4}\right)
\right. \\ \left.
+H_{n+3,1}-1-\frac{1}{2}-\frac{1}{3}-H_{n+4,1}+1+\frac{1}{2}+\frac{1}{3}+\frac{1}{4}
+H_{n+4,2}-1-\frac{1}{4}-\frac{1}{9}-\frac{1}{16}
\right)
\end{array}
$$

$$
\begin{array}{l}
%\\
 =
\frac{1}{2}\left(
\frac{1}{n+2}
-\frac{1}{4}+\frac{11}{18}+\frac{25}{48}
-\frac{H_{n+3}}{n+3}-\frac{H_{n+4}}{n+4}
+\frac{1}{3}\left(\frac{1}{2}+\frac{1}{3}+\frac{1}{4}-\frac{1}{n+2}-\frac{1}{n+3}-\frac{1}{n+4}\right)
\right. \\ \left.
+\frac{1}{2}\left(\frac{1}{3}+\frac{1}{4}-\frac{1}{n+3}-\frac{1}{n+4}\right)
-\frac{1}{n+4}
-\frac{1}{9}-\frac{1}{16}
+\frac{1}{(n+3)^{2}}+\frac{1}{(n+4)^{2}}
\right)
\\=
\frac{1}{2}\left(
\frac{1}{n+2}
+\frac{1}{4}+\frac{25}{48}
-\frac{H_{n+3}}{n+3}-\frac{H_{n+4}}{n+4}
+\frac{1}{6}+\frac{1}{9}+\frac{1}{12}-\frac{1}{3(n+2)}-\frac{1}{3(n+3)}-\frac{1}{3(n+4)}
\right. \\ \left.
+\frac{1}{6}+\frac{1}{8}
-\frac{1}{2(n+3)}-\frac{1}{2(n+4)}
-\frac{1}{n+4}
-\frac{1}{16}
+\frac{1}{(n+3)^{2}}+\frac{1}{(n+4)^{2}}
\right)
\\=
\frac{1}{2}\left(
\frac{49}{36}
-\frac{H_{n+3}}{n+3}-\frac{H_{n+4}}{n+4}
+\frac{2}{3(n+2)}-\frac{5}{6(n+3)}-\frac{11}{6(n+4)}
+\frac{1}{(n+3)^{2}}+\frac{1}{(n+4)^{2}}
\right)
.
\end{array}
$$
\end{proof}

\begin{lemma}[See also \citet{ASof2011}, p. $18$, for the limit]\label{lemSumHiip1sq}

\be\label{eqSumHiip1sq}
G_{n,2,0}^{1,0,1}=\sum\limits_{j=1}^{n}\frac{H_{j,1}}{(j+1)^{2}} = \sum\limits_{j=1}^{n+1}\frac{H_{j,1}}{j^{2}}-H_{n+1,3} \to \zeta(3).
\ee
\end{lemma}

\begin{proof}

$$
\begin{array}{l}
\sum\limits_{j=1}^{n}\frac{H_{j,1}}{(j+1)^{2}} 
= \sum\limits_{j=1}^{n}\frac{H_{j+1,1}}{(j+1)^{2}} -\sum\limits_{j=1}^{n}\frac{1}{(j+1)^{3}}
=\sum\limits_{j=1}^{n+1}\frac{H_{j,1}}{j^{2}}-1 -\sum\limits_{j=1}^{n+1}\frac{1}{j^{3}}+1
\\ =\sum\limits_{j=1}^{n+1}\frac{H_{j,1}}{j^{2}}-H_{n+1,3} \to \zeta(3) \approx 1.202.
\end{array}
$$
\end{proof}

\begin{lemma}\label{lemSumHiip2sq}

\be\label{eqSumHiip2sq}
\begin{array}{rcl}
G_{n,2,0}^{2,0,1}&=&\sum\limits_{j=1}^{n}\frac{H_{j,1}}{(j+2)^{2}} = 
\sum\limits_{j=1}^{n+2}\frac{H_{j,1}}{j^{2}}
+H_{n+2,2}- H_{n+2,3}-2
+\frac{1}{n+2} 
\\ &\xrightarrow{n\to \infty} &
\zeta(3)+\frac{\pi^{2}}{6}-2.
\end{array}
\ee
\end{lemma}

\begin{proof}

$$
\begin{array}{l}
\sum\limits_{j=1}^{n}\frac{H_{j,1}}{(j+2)^{2}} = 
\sum\limits_{j=1}^{n}\frac{1}{(j+2)^{2}}\left(
H_{j,1}+\frac{1}{j+1}+\frac{1}{j+2}-\left(\frac{1}{j+1}+\frac{1}{j+2} \right)
\right)
\\
=
\sum\limits_{j=1}^{n}\frac{H_{j+2,1}}{(j+2)^{2}}
- \sum\limits_{j=1}^{n}\frac{1}{(j+1)(j+2)^{2}}
- \sum\limits_{j=1}^{n}\frac{1}{(j+2)^{3}}
\\
=
\sum\limits_{j=3}^{n+2}\frac{H_{j,1}}{j^{2}}
- \sum\limits_{j=1}^{n}\frac{1}{(j+1)(j+2)^{2}}
- \sum\limits_{j=3}^{n+2}\frac{1}{j^{3}}
\\ =
\sum\limits_{j=1}^{n+2}\frac{H_{j,1}}{j^{2}}-1-\frac{3}{8}
- \sum\limits_{j=1}^{n}\frac{1}{(j+1)(j+2)^{2}}
- \sum\limits_{j=1}^{n+2}\frac{1}{j^{3}} + 1 +\frac{1}{8}
\\ \stackrel{Lemma~\ref{lemSumi1i2sq}}{=}
\sum\limits_{j=1}^{n+2}\frac{H_{j,1}}{j^{2}}-\frac{2}{8}
-\frac{7}{4}+H_{n+2,2}+\frac{1}{n+2}
- H_{n+2,3}
\\ =
\sum\limits_{j=1}^{n+2}\frac{H_{j,1}}{j^{2}}
+H_{n+2,2}- H_{n+2,3}-2
+\frac{1}{n+2}
\\ \xrightarrow{n\to \infty}
2\zeta(3)+\frac{\pi^{2}}{6}-\zeta(3)-2
=
\zeta(3)+\frac{\pi^{2}}{6}-2 \approx 0.847.
\end{array}
$$
\end{proof}

\begin{lemma}\label{lemSumHii1i2sq}

\be\label{eqSumHii1i2sq}
\begin{array}{rcl}
G_{n,1,2}^{1,2,1}&=&\sum\limits_{j=1}^{n}\frac{H_{j,1}}{(j+1)(j+2)^{2}} 
 = 
3-H_{n+2,2}+H_{n+2,3}
-\sum\limits_{j=1}^{n+2}\frac{H_{j,1}}{j^{2}}
\\ &&
- \frac{H_{n+1,1}}{n+2}
-\frac{2}{n+2}
%\\ & 
\xrightarrow{n\to \infty}  3-\frac{\pi^{2}}{6}-\zeta(3).
\end{array}
\ee
\end{lemma}

\begin{proof}

$$
\begin{array}{l}
\sum\limits_{j=1}^{n}\frac{H_{j,1}}{(j+1)(j+2)^{2}} =
\sum\limits_{j=1}^{n}\frac{H_{j,1}}{(j+2)}\left(\frac{1}{j+1}-\frac{1}{j+2}\right) =
\sum\limits_{j=1}^{n}\frac{H_{j,1}}{(j+1)(j+2)}
-\sum\limits_{j=1}^{n}\frac{H_{j,1}}{(j+2)^{2}}
\\ \stackrel{Lemma~\ref{lemSumHii1i2}}{=}
\frac{n+1}{n+2} - \frac{H_{n+1,1}}{n+2}
-\sum\limits_{j=1}^{n}\frac{H_{j+2,1}}{(j+2)^{2}}
+ \sum\limits_{j=1}^{n}\frac{1}{(j+1)(j+2)^{2}}
+ \sum\limits_{j=1}^{n}\frac{1}{(j+2)^{3}}
\\ \stackrel{Lemma~\ref{lemSumi1i2sq}}{=}
1-\frac{1}{n+2} - \frac{H_{n+1,1}}{n+2}
-\sum\limits_{j=1}^{n+2}\frac{H_{j,1}}{j^{2}}+1+\frac{3}{8}
+\frac{7}{4}-H_{n+2,2}-\frac{1}{n+2}
\\ +H_{n+2,3}-1-\frac{1}{8}
\\ =
3-H_{n+2,2}+H_{n+2,3}
-\sum\limits_{j=1}^{n+2}\frac{H_{j,1}}{j^{2}}
- \frac{H_{n+1,1}}{n+2}
-\frac{2}{n+2}
\\ \xrightarrow{n\to \infty}
3-\frac{\pi^{2}}{6}+\zeta(3)-2\zeta(3)
 = 3-\frac{\pi^{2}}{6}-\zeta(3) \approx 0.153.
\end{array}
$$
\end{proof}

\begin{lemma}\label{lemSumHiii1sq}

\be\label{eqSumHiii1sq}
\begin{array}{rcl}
G_{n,1,2}^{0,1,1} & = & \sum\limits_{j=1}^{n}\frac{H_{j,1}}{j(j+1)^{2}} 
=
H_{n,2}+H_{n+1,3} -\sum\limits_{j=1}^{n+1}\frac{H_{j,1}}{j^{2}} - \frac{H_{n,1}}{n+1}
\\ & \xrightarrow{n\to \infty} &
\frac{\pi^{2}}{6}-\zeta(3)
.
\end{array}
\ee
\end{lemma}

\begin{proof}

$$
\begin{array}{l}
\sum\limits_{j=1}^{n}\frac{H_{j,1}}{j(j+1)^{2}} =
\frac{1}{4}+
\sum\limits_{j=2}^{n}\frac{H_{j,1}}{j(j+1)^{2}} =
\frac{1}{4}+
\sum\limits_{j=1}^{n-1}\frac{H_{j+1,1}}{(j+1)(j+2)^{2}} 
=
\frac{1}{4}+
\sum\limits_{j=1}^{n-1}\frac{1}{(j+1)^{2}(j+2)^{2}} 
\\ +\sum\limits_{j=1}^{n-1}\frac{H_{j+1,1}}{(j+1)(j+2)^{2}}
\\ 
\stackrel{Lemmata~\ref{lemSumj1sqj2sq},\ref{lemSumHii1i2sq}}{=}
\frac{1}{4}+
2H_{n,2}+\frac{2}{n+1}+\frac{1}{(n+1)^{2}}-\frac{13}{4}
+3-H_{n+1,2}+H_{n+1,3}
\\ -\sum\limits_{j=1}^{n+1}\frac{H_{j,1}}{j^{2}}
- \frac{H_{n,1}}{n+1}
-\frac{2}{n+1}
 =
H_{n,2}
+H_{n+1,3}
-\sum\limits_{j=1}^{n+1}\frac{H_{j,1}}{j^{2}}
- \frac{H_{n,1}}{n+1}
\\ \xrightarrow{n\to \infty}
\frac{\pi^{2}}{6}+\zeta(3)-2\zeta(3)
=\frac{\pi^{2}}{6}-\zeta(3)
\approx 0.443.
\end{array}
$$
\end{proof}

\begin{lemma}\label{lemSumHiii2sq}

\be\label{eqSumHiii2sq}
\begin{array}{rcl}
G_{n,1,2}^{0,2,1}
& = &
\sum\limits_{j=1}^{n}\frac{H_{j,1}}{j(j+2)^{2}} 
 =
\frac{5}{4} + \frac{H_{n+2,3}}{2} - \frac{H_{n+2,2}}{4}  - \frac{1}{2}\sum\limits_{j=1}^{n+2}\frac{H_{j,1}}{j^{2}}
\\ && -\frac{H_{n+1,1}}{4(n+1)} -\frac{H_{n+2,1}}{4(n+2)} -\frac{3}{4(n+2)}  
%\\ & 
\xrightarrow{n\to \infty} 
\frac{5}{4} - \frac{\zeta(3)}{2} - \frac{\pi^{2}}{24}.
\end{array}
\ee
\end{lemma}

\begin{proof}

$$
\begin{array}{l}
\sum\limits_{j=1}^{n}\frac{H_{j,1}}{j(j+2)^{2}} =
\frac{1}{2}\left(
\sum\limits_{j=1}^{n}\frac{H_{j,1}}{j(j+2)} - \sum\limits_{j=1}^{n}\frac{H_{j,1}}{(j+2)^{2}} \right)
\\ \stackrel{Lemmata~\ref{lemSumHiii2},\ref{lemSumHiip2sq}}{=}
\frac{1}{2}\left(
\frac{1}{2} + \frac{H_{n+2,2}}{2} -\frac{H_{n+1,1}}{2(n+1)} -\frac{H_{n+2,1}}{2(n+2)} -\frac{1}{2(n+2)}  
-
\sum\limits_{j=1}^{n+2}\frac{H_{j,1}}{j^{2}} - H_{n+2,2}
\right. \\ \left. 
+ H_{n+2,3}+2 -\frac{1}{n+2} 
\right)
\\ =
\frac{1}{2}\left(
\frac{5}{2} + H_{n+2,3} - \frac{H_{n+2,2}}{2}  - \sum\limits_{j=1}^{n+2}\frac{H_{j,1}}{j^{2}}
-\frac{H_{n+1,1}}{2(n+1)} -\frac{H_{n+2,1}}{2(n+2)} -\frac{3}{2(n+2)}  
\right)
\\ \xrightarrow{n\to \infty} 
\frac{5}{4} - \frac{\zeta(3)}{2} - \frac{\pi^{2}}{24}
\approx 0.238.
\end{array}
$$
\end{proof}

\begin{lemma}\label{lemSumHii1sqi2sq}

\be\label{eqSumHii1sqi2sq}
\begin{array}{rcl}
G_{n,2,2}^{1,2,1}
& = &
\sum\limits_{j=1}^{n}\frac{H_{j,1}}{(j+1)^{2}(j+2)^{2}} 
=
2\sum\limits_{j=1}^{n}\frac{H_{j,1}}{j^{2}}
+H_{n,2}-2H_{n,3}  -4  
\\ && + \frac{2H_{n,1}}{n+2} + \frac{3}{n+1} + \frac{2H_{n,1}}{(n+1)^{2}}+ \frac{H_{n,1}}{(n+2)^{2}}+\frac{1}{(n+1)^{2}}
\\ & \xrightarrow{n\to \infty} & 2\zeta(3)+\frac{\pi^{2}}{6}-4.
\end{array}
\ee
\end{lemma}

\begin{proof}

$$
\begin{array}{l}
\sum\limits_{j=1}^{n}\frac{H_{j,1}}{(j+1)^{2}(j+2)^{2}}
=
\sum\limits_{j=1}^{n}H_{j,1}\left(\frac{1}{j+1}-\frac{1}{j+2}\right)^{2}
=
\sum\limits_{j=1}^{n}\frac{H_{j,1}}{(j+1)^{2}}
-
2\sum\limits_{j=1}^{n}\frac{H_{j,1}}{(j+1)(j+2)}
+
\sum\limits_{j=1}^{n}\frac{H_{j,1}}{(j+2)^{2}}
\\ =
\sum\limits_{j=1}^{n}\frac{H_{j+1,1}}{(j+1)^{2}}-\sum\limits_{j=1}^{n}\frac{1}{(j+1)^{3}}
-
2\sum\limits_{j=1}^{n}\frac{H_{j,1}}{(j+1)(j+2)}
+
\sum\limits_{j=1}^{n}\frac{H_{j+2,1}}{(j+2)^{2}}-\sum\limits_{j=1}^{n}\frac{1}{(j+1)(j+2)^{2}}
-
\sum\limits_{j=1}^{n}\frac{1}{(j+2)^{3}}
\\ =
\sum\limits_{j=1}^{n+1}\frac{H_{j,1}}{j^{2}}-1-\sum\limits_{j=1}^{n+1}\frac{1}{j^{3}}+1
-
2\sum\limits_{j=1}^{n}\frac{H_{j,1}}{(j+1)(j+2)}
+
\sum\limits_{j=1}^{n+2}\frac{H_{j,1}}{j^{2}}-1-\frac{3}{8}-
\sum\limits_{j=1}^{n}\frac{1}{(j+1)(j+2)^{2}}
\\ -
\sum\limits_{j=1}^{n+2}\frac{1}{j^{3}}+1+\frac{1}{8}
\\ =
2\sum\limits_{j=1}^{n+1}\frac{H_{j,1}}{j^{2}}
-2H_{n+1,3}
-
2\sum\limits_{j=1}^{n}\frac{H_{j,1}}{(j+1)(j+2)}
+\frac{H_{n+2,1}}{(n+2)^{2}}
-
\sum\limits_{j=1}^{n}\frac{1}{(j+1)(j+2)^{2}}
-\frac{1}{(n+2)^{3}}-\frac{2}{8}
\\ \stackrel{Lemmata~\ref{lemSumi1i2sq},\ref{lemSumHii1i2}}{=}
2\sum\limits_{j=1}^{n+1}\frac{H_{j,1}}{j^{2}}
-2H_{n+1,3}
-2\frac{n+1}{n+2} + \frac{2H_{n+1,1}}{n+2}
-\frac{7}{4}+H_{n+2,2}
+\frac{H_{n+2,1}}{(n+2)^{2}}
\\ +\frac{1}{n+2}
-\frac{1}{(n+2)^{3}}-\frac{1}{4}
\\ \to 
4\zeta(3)-2\zeta(3)-2-\frac{7}{4}+\frac{\pi^{2}}{6}-\frac{1}{4}
=
2\zeta(3)+\frac{\pi^{2}}{6}-\frac{16}{4}
=
2\zeta(3)+\frac{\pi^{2}}{6}-4 \approx 0.049.
\end{array}
$$
\end{proof}

\begin{lemma}\label{lemSumHiii1i2}

\be\label{eqSumHiii1i2}
\sum\limits_{j=1}^{n}\frac{H_{j,1}}{j(j+1)(j+2)} = 
\frac{1}{2}\left( H_{n,2}-1-\frac{H_{n,1}}{n+1}+\frac{H_{n,1}}{n+2}+\frac{1}{n+1} \right) 
\xrightarrow{n\to \infty} \frac{\pi^{2}}{12} - \frac{1}{2}.
\ee
\end{lemma}

\begin{proof}

$$
\begin{array}{l}
\sum\limits_{j=1}^{n}\frac{H_{j,1}}{j(j+1)(j+2)} 
= \sum\limits_{j=1}^{n}\left(\frac{H_{j,1}}{j(j+1)}-\frac{H_{j,1}}{j(j+2)}\right)
= \sum\limits_{j=1}^{n}\frac{H_{j,1}}{j} 
- \sum\limits_{j=1}^{n}\frac{H_{j,1}}{j+1}
- \frac{1}{2}\sum\limits_{j=1}^{n}\frac{H_{j,1}}{j}
+ \frac{1}{2}\sum\limits_{j=1}^{n}\frac{H_{j,1}}{j+2}
\\ =
\frac{1}{2}\sum\limits_{j=1}^{n}\frac{H_{j,1}}{j} 
- \sum\limits_{j=1}^{n}\frac{H_{j+1,1}}{j+1} + \sum\limits_{j=1}^{n}\frac{1}{(j+1)^{2}}
+ \frac{1}{2}\sum\limits_{j=1}^{n}\frac{H_{j+2,1}}{j+2} 
-  \frac{1}{2}\sum\limits_{j=1}^{n}\frac{1}{(j+1)(j+2)}
-  \frac{1}{2}\sum\limits_{j=1}^{n}\frac{1}{(j+2)^{2}}
\\ \stackrel{Lemma~\ref{lemSum1}}{=}
\frac{1}{2}\sum\limits_{j=1}^{n}\frac{H_{j,1}}{j} 
- \sum\limits_{j=1}^{n+1}\frac{H_{j,1}}{j}+1 + H_{n+1,2}-1
+ \frac{1}{2}\sum\limits_{j=1}^{n+2}\frac{H_{j,1}}{j}  - \frac{1}{2} - \frac{3}{8}
-  \frac{1}{2}\left(\frac{1}{2} - \frac{1}{n+2}\right)
\\ -  \frac{1}{2}H_{n+2,2} + \frac{1}{2} + \frac{1}{8}
\\=
\sum\limits_{j=1}^{n}\frac{H_{j,1}}{j} 
- \sum\limits_{j=1}^{n+1}\frac{H_{j,1}}{j} + \frac{1}{2}H_{n+1,2}-\frac{1}{2}\frac{1}{(n+2)^{2}}
+ \frac{1}{2}\frac{H_{n+1,1}}{n+1}+ \frac{1}{2}\frac{H_{n+2,1}}{n+2}
-\frac{1}{2}  + \frac{1}{2}\frac{1}{n+2}
\\ =
-\frac{H_{n+1,1}}{n+1} 
+ \frac{1}{2}H_{n+1,2}-\frac{1}{2}\frac{1}{(n+2)^{2}}
+ \frac{1}{2}\frac{H_{n+1,1}}{n+1}+ \frac{1}{2}\frac{H_{n+2,1}}{n+2}
-\frac{1}{2}
 + \frac{1}{2}\frac{1}{n+2}
\\=
\frac{1}{2}\left(
H_{n+1,2}-1+\frac{1}{n+2}
- \frac{H_{n+1,1}}{n+1}+ \frac{H_{n+1,1}}{n+2}
\right)
=
\frac{1}{2}\left(
H_{n+1,2}-1+\frac{1}{n+2}
- \frac{H_{n+1,1}}{(n+1)(n+2)}
\right)
\\ \xrightarrow{n\to \infty} \frac{\pi^{2}}{12} - \frac{1}{2} \approx 0.322.
\end{array}
$$
\end{proof}

\begin{lemma}\label{lemSumHj1j1j2j3}

\be \label{eqSumHj1j1j2j3}
\begin{array}{rcl}
\sum\limits_{j=1}^{n}\frac{H_{j,1}}{(j+1)(j+2)(j+3)} 
& = & \frac{1}{8}-\frac{H_{n,1}}{2(n+2)}+\frac{H_{n,1}}{2(n+3)}-\frac{1}{4(n+1)}+\frac{1}{4(n+2)}
\\ & \xlongrightarrow{n\to\infty} & \frac{1}{8}.
\end{array}
\ee
\end{lemma}

\begin{proof}

$$
\begin{array}{l}
\sum\limits_{j=1}^{n}\frac{H_{j,1}}{(j+1)(j+2)(j+3)} 
=\sum\limits_{j=1}^{n}\frac{H_{j,1}}{j+3}\left(\frac{1}{j+1}-\frac{1}{j+2}\right)
\\=\sum\limits_{j=1}^{n}H_{j,1}\left(\frac{1}{2}\left(\frac{1}{j+1}-\frac{1}{j+3}\right)
-\frac{1}{j+2}+\frac{1}{j+3}\right)
=\frac{1}{2}\sum\limits_{j=1}^{n}\frac{H_{j,1}}{j+1}-\sum\limits_{j=1}^{n}\frac{H_{j,1}}{j+2}
+\frac{1}{2}\sum\limits_{j=1}^{n}\frac{H_{j,1}}{j+3}
\end{array}
$$

$$
\begin{array}{l}
\stackrel{Lemmata~\ref{lemSumHj1jmapp}, \ref{lemSumHj1j2},\ref{lemSumHj1j3}}{=}
\frac{1}{2}\frac{1}{2}\left(H_{n+1,1}^{2}-H_{n+1,2}\right)
-\frac{1}{2}\left(H_{n+1,1}^{2}-H_{n+1,2}\right)+1-\frac{1}{n+2}
\\
-\frac{H_{n,1}}{n+2} +\frac{1}{2}\left(\frac{1}{2}\left(H_{n+1,1}^{2}-H_{n+1,2}\right)
-\frac{21}{12}+\frac{H_{n+1,1}}{n+2}+\frac{H_{n+1,1}}{n+3}+\frac{3}{2(n+2)}+\frac{1}{2(n+3)}
\right)
\\ =
\frac{1}{8}-\frac{H_{n+1,1}}{2(n+2)}+\frac{H_{n+1,1}}{2(n+3)}-\frac{1}{4(n+2)}+\frac{1}{4(n+3)}.
\end{array}
$$
\end{proof}

\begin{lemma}\label{lemSumHj1j1j2j4}

\be\label{eqSumHj1j1j2j4}
\begin{array}{rcl}
\sum\limits_{j=1}^{n}\frac{H_{j,1}}{(j+1)(j+2)(j+4)} & = &
\frac{23}{216}-\frac{H_{n,1}}{3(n+2)}+\frac{H_{n,1}}{6(n+3)}+\frac{H_{n,1}}{6(n+4)} -\frac{7}{36(n+1)}+\frac{5}{36(n+2)}
\\&& +\frac{1}{18(n+3)} \xlongrightarrow{n\to \infty} \frac{23}{216}.
\end{array}
\ee
\end{lemma}

\begin{proof}

$$
\begin{array}{l}
\sum\limits_{j=1}^{n}\frac{H_{j,1}}{(j+1)(j+2)(j+4)} =
\frac{1}{2}\left(\sum\limits_{j=1}^{n}\frac{H_{j,1}}{(j+1)(j+2)} -
\sum\limits_{j=1}^{n}\frac{H_{j,1}}{(j+1)(j+4)}\right) 
\\\stackrel{Lemmata \ref{lemSumHii1i2},\ref{lemSumHj1j1j4}}{=}
 \frac{1}{2}\left(1-\frac{H_{n+1,1}}{n+2}-\frac{1}{n+2}
-\frac{85}{108}
+\frac{H_{n+2,1}}{3(n+2)}+\frac{H_{n+3,1}}{3(n+3)}+\frac{H_{n+4,1}}{3(n+4)}
\right. \\ \left.
+\frac{1}{9(n+2)}+\frac{5}{18(n+3)}+\frac{11}{18(n+4)}
-\frac{1}{3(n+2)^{2}}-\frac{1}{3(n+3)^{2}}-\frac{1}{3(n+4)^{2}}
\right)
\\ = 
\frac{1}{2}\left(
\frac{23}{108}
-\frac{2H_{n+2,1}}{3(n+2)}+\frac{H_{n+3,1}}{3(n+3)}+\frac{H_{n+4,1}}{3(n+4)}
-\frac{8}{9(n+2)}+\frac{5}{18(n+3)}+\frac{11}{18(n+4)}
\right. \\ \left.
+\frac{2}{3(n+2)^{2}}-\frac{1}{3(n+3)^{2}}-\frac{1}{3(n+4)^{2}}
\right)
\\ =
\frac{1}{2}\left(
\frac{23}{108}
-\frac{2H_{n+1,1}}{3(n+2)}+\frac{H_{n+1,1}}{3(n+3)}+\frac{H_{n+1,1}}{3(n+4)}
-\frac{7}{18(n+2)}+\frac{5}{18(n+3)}+\frac{1}{9(n+4)}
\right).
\end{array}
$$
\end{proof}

\
\begin{lemma}\label{lemSumHj1j1sqj2j3}

\be\label{eqSumHj1j1sqj2j3}
\begin{array}{rcl}
\sum\limits_{j=1}^{n}\frac{H_{j,1}}{(j+1)^{2}(j+2)(j+3)} & = &
\frac{1}{2}\sum\limits_{j=1}^{n}\frac{H_{j,1}}{j^{2}}-\frac{1}{2}H_{n,3}-\frac{9}{16} +\frac{3H_{n,1}}{4(n+2)}-\frac{H_{n,1}}{4(n+3)}
\\ && +\frac{5}{8(n+1)}-\frac{1}{8(n+2)} + \frac{H_{n,1}}{2(n+1)^{2}}
\xlongrightarrow{n \to \infty}\frac{1}{2}\zeta(3)-\frac{9}{16}
.
\end{array}
\ee
\end{lemma}

\begin{proof}

$$
\begin{array}{l}
\sum\limits_{j=1}^{n}\frac{H_{j,1}}{(j+1)^{2}(j+2)(j+3)} =
\sum\limits_{j=1}^{n}\frac{H_{j,1}}{(j+1)^{2}(j+3)}
-\sum\limits_{j=1}^{n}\frac{H_{j,1}}{(j+1)(j+2)(j+3)}
=
\frac{1}{2}\left(
\sum\limits_{j=1}^{n}\frac{H_{j,1}}{(j+1)^{2}}
-
\sum\limits_{j=1}^{n}\frac{H_{j,1}}{(j+1)(j+3)}
\right)
\\ -\sum\limits_{j=1}^{n}\frac{H_{j,1}}{(j+1)(j+2)(j+3)}
\end{array}
$$

$$
\begin{array}{l}
\stackrel{Lemmata \ref{lemSumHiip1sq},\ref{lemSumHj1j1j3},\ref{lemSumHj1j1j2j3}}{=}
\frac{1}{2}\left(
\sum\limits_{j=1}^{n+1}\frac{H_{j,1}}{j^{2}}-H_{n+1,3}
-\frac{7}{8}
+\frac{H_{n+2,1}}{2(n+2)}+\frac{H_{n+3,1}}{2(n+3)}
\right. \\ \left. 
+\frac{1}{4(n+2)}+\frac{3}{4(n+3)}
-\frac{1}{2(n+2)^{2}}-\frac{1}{2(n+3)^{2}}
\right)
-\frac{1}{8}
+\frac{H_{n+2,1}}{2(n+2)}-\frac{H_{n+3,1}}{2(n+3)}
+\frac{3}{4(n+2)}-\frac{3}{4(n+3)}
\\-\frac{1}{2(n+2)^{2}}+\frac{1}{2(n+3)^{2}}
 =
\frac{1}{2}\sum\limits_{j=1}^{n+1}\frac{H_{j,1}}{j^{2}}-\frac{1}{2}H_{n+1,3}
-\frac{7}{16}
+\frac{H_{n+2,1}}{4(n+2)}+\frac{H_{n+3,1}}{4(n+3)}
+\frac{1}{8(n+2)}+\frac{3}{8(n+3)}
\\-\frac{1}{4(n+2)^{2}}-\frac{1}{4(n+3)^{2}}
-\frac{1}{8}
+\frac{H_{n+2,1}}{2(n+2)}-\frac{H_{n+3,1}}{2(n+3)}
+\frac{3}{4(n+2)}-\frac{3}{4(n+3)}
-\frac{1}{2(n+2)^{2}}+\frac{1}{2(n+3)^{2}}
\\=
\frac{1}{2}\sum\limits_{j=1}^{n+1}\frac{H_{j,1}}{j^{2}}-\frac{1}{2}H_{n+1,3}-\frac{9}{16}
+\frac{3H_{n+2,1}}{4(n+2)}-\frac{H_{n+3,1}}{4(n+3)}
+\frac{7}{8(n+2)}-\frac{3}{8(n+3)}
-\frac{3}{4(n+2)^{2}}+\frac{1}{4(n+3)^{2}}
\\=
\frac{1}{2}\sum\limits_{j=1}^{n+1}\frac{H_{j,1}}{j^{2}}-\frac{1}{2}H_{n+1,3}-\frac{9}{16}
+\frac{3H_{n+1,1}}{4(n+2)}-\frac{H_{n+1,1}}{4(n+3)}
+\frac{5}{8(n+2)}-\frac{1}{8(n+3)}
\xlongrightarrow{n \to \infty}\frac{1}{2}\zeta(3)-\frac{9}{16}\approx 0.039
.
\end{array}
$$
\end{proof}

\begin{lemma}\label{lemSumHj1j1sqj2j4}

\be\label{eqSumHj1j1sqj2j4}
\begin{array}{rcl}
\sum\limits_{j=1}^{n}\frac{H_{j,1}}{(j+1)^{2}(j+2)(j+4)} & = &
\frac{1}{3}\sum\limits_{j=1}^{n}\frac{H_{j,1}}{j^{2}}-\frac{1}{3}H_{n,3}-\frac{239}{648}
+\frac{4H_{n,1}}{9(n+2)}-\frac{H_{n,1}}{18(n+3)}-\frac{H_{n,1}}{18(n+4)}
\\ && 
+\frac{43}{108(n+1)}-\frac{5}{108(n+2)}-\frac{1}{54(n+3)} + \frac{H_{n,1}}{3(n+1)^{2}}
\xlongrightarrow{n\to\infty} \frac{1}{3}\zeta(3)-\frac{239}{648}
.
\end{array}
\ee
\end{lemma}

\begin{proof}

$$
\begin{array}{l}
\sum\limits_{j=1}^{n}\frac{H_{j,1}}{(j+1)^{2}(j+2)(j+4)} =
\sum\limits_{j=1}^{n}\frac{H_{j,1}}{(j+1)^{2}(j+4)}
-\sum\limits_{j=1}^{n}\frac{H_{j,1}}{(j+1)(j+2)(j+4)}
=
\frac{1}{3}\left(
\sum\limits_{j=1}^{n}\frac{H_{j,1}}{(j+1)^{2}}
-
\sum\limits_{j=1}^{n}\frac{H_{j,1}}{(j+1)(j+4)}
\right)
\\ -\sum\limits_{j=1}^{n}\frac{H_{j,1}}{(j+1)(j+2)(j+4)}
\\ \stackrel{Lemmata~ \ref{lemSumHj1j1j4},\ref{lemSumHiip1sq},\ref{lemSumHj1j1j2j4}}{=}
\frac{1}{3}\left(
\sum\limits_{j=1}^{n+1}\frac{H_{j,1}}{j^{2}}-H_{n+1,3}
-\frac{85}{108}
+\frac{H_{n+2,1}}{3(n+2)}+\frac{H_{n+3,1}}{3(n+3)}+\frac{H_{n+4,1}}{3(n+4)}
\right. \\ \left. 
+\frac{1}{9(n+2)}+\frac{5}{18(n+3)}+\frac{11}{18(n+4)}
-\frac{1}{3(n+2)^{2}}-\frac{1}{3(n+3)^{2}}-\frac{1}{3(n+4)^{2}}
\right)
-\frac{23}{216}
\\
+\frac{H_{n+2,1}}{3(n+2)}-\frac{H_{n+3,1}}{6(n+3)}-\frac{H_{n+4,1}}{6(n+4)}
+\frac{4}{9(n+2)}-\frac{5}{36(n+3)}-\frac{11}{36(n+4)}
-\frac{1}{3(n+2)^{2}}+\frac{1}{6(n+3)^{2}}+\frac{1}{6(n+4)^{2}}
\end{array}
$$

$$
\begin{array}{l}
=
\frac{1}{3}\sum\limits_{j=1}^{n+1}\frac{H_{j,1}}{j^{2}}-\frac{1}{3}H_{n+1,3}
-\frac{85}{324}
+\frac{H_{n+2,1}}{9(n+2)}+\frac{H_{n+3,1}}{9(n+3)}+\frac{H_{n+4,1}}{9(n+4)}
+\frac{1}{27(n+2)}+\frac{5}{54(n+3)}
\\
+\frac{11}{54(n+4)}
-\frac{1}{9(n+2)^{2}}-\frac{1}{9(n+3)^{2}}-\frac{1}{9(n+4)^{2}}
-\frac{23}{216}
+\frac{H_{n+2,1}}{3(n+2)}-\frac{H_{n+3,1}}{6(n+3)}-\frac{H_{n+4,1}}{6(n+4)}
\\
+\frac{4}{9(n+2)}-\frac{5}{36(n+3)}-\frac{11}{36(n+4)}
-\frac{1}{3(n+2)^{2}}+\frac{1}{6(n+3)^{2}}+\frac{1}{6(n+4)^{2}}
\\=
\frac{1}{3}\sum\limits_{j=1}^{n+1}\frac{H_{j,1}}{j^{2}}-\frac{1}{3}H_{n+1,3}-\frac{239}{648}
+\frac{4H_{n+2,1}}{9(n+2)}-\frac{H_{n+3,1}}{18(n+3)}-\frac{H_{n+4,1}}{18(n+4)}
+\frac{13}{27(n+2)}
\\
-\frac{5}{108(n+3)}-\frac{11}{108(n+4)}
-\frac{4}{9(n+2)^{2}}+\frac{1}{18(n+3)^{2}}+\frac{1}{18(n+4)^{2}}
\\=
\frac{1}{3}\sum\limits_{j=1}^{n+1}\frac{H_{j,1}}{j^{2}}-\frac{1}{3}H_{n+1,3}-\frac{239}{648}
+\frac{4H_{n+1,1}}{9(n+2)}-\frac{H_{n+1,1}}{18(n+3)}-\frac{H_{n+1,1}}{18(n+4)}
\\+\frac{43}{108(n+2)}-\frac{5}{108(n+3)}-\frac{1}{54(n+4)}
\xlongrightarrow{n\to\infty} \frac{1}{3}\zeta(3)-\frac{239}{648} \approx 0.032
.
\end{array}
$$
\end{proof}

\begin{lemma}\label{lemSumHj1j1j2j3j4}

\be\label{eqSumHj1j1j2j3j4}
\begin{array}{rcl}
\sum\limits_{j=1}^{n}\frac{H_{j,1}}{(j+1)(j+2)(j+3)(j+4)}& = &
\frac{1}{54} -\frac{H_{n,1}}{6(n+2)}+\frac{H_{n,1}}{3(n+3)}-\frac{H_{n,1}}{6(n+4)}
\\&& -\frac{1}{18(n+1)} +\frac{1}{9(n+2)} -\frac{1}{18(n+3)} 
\xlongrightarrow{n\to  \infty}\frac{1}{54}
.
\end{array}
\ee
\end{lemma}

\begin{proof}

$$
\begin{array}{l}
\sum\limits_{j=1}^{n}\frac{H_{j,1}}{(j+1)(j+2)(j+3)(j+4)} 
=\frac{1}{6}\sum\limits_{j=1}^{n}\frac{H_{j,1}}{j+1}-\frac{1}{2}\sum\limits_{j=1}^{n}\frac{H_{j,1}}{j+2}+\frac{1}{2}\sum\limits_{j=1}^{n}\frac{H_{j,1}}{j+3}-\frac{1}{6}\sum\limits_{j=1}^{n}\frac{H_{j,1}}{j+4}
\\ \stackrel{Lemmata~ \ref{lemSumHj1jmapp},\ref{lemSumHj1j2},\ref{lemSumHj1j3},\ref{lemSumHj1j4}}{=}
\frac{1}{6}\cdot\frac{1}{2}\left(H_{n+1,1}^{2}-H_{n+1,2}\right)
-\frac{1}{2}\left(\frac{1}{2}\left(H_{n+1,1}^{2}-H_{n+1,2}\right)
\right. \\ \left.
-1+\frac{1}{n+2}+\frac{H_{n,1}}{n+2}\right)
 +\frac{1}{2}\left(\frac{1}{2}\left(H_{n+1,1}^{2}-H_{n+1,2}\right)
-\frac{21}{12}+\frac{H_{n+1,1}}{n+2}+\frac{H_{n+1,1}}{n+3}+\frac{3}{2(n+2)}+\frac{1}{2(n+3)}
\right)
\\ -\frac{1}{6}\left(
\frac{1}{2}\left(H_{n+1,1}^{2}-H_{n+1,2} \right)
-\frac{85}{36}
+\frac{H_{n+1,1}}{n+2}+\frac{H_{n+1,1}}{n+3}+\frac{H_{n+1,1}}{n+4}
+\frac{11}{6(n+2)}+ \frac{5}{6(n+3)}+\frac{1}{3(n+4)}
\right)
\\ =
\frac{1}{54}
-\frac{1}{18(n+2)}
+\frac{1}{9(n+3)}
-\frac{1}{18(n+4)}
-\frac{H_{n+1,1}}{6(n+2)}+\frac{H_{n+1,1}}{3(n+3)}-\frac{H_{n+1,1}}{6(n+4)}
.
\end{array}
$$
\end{proof}

\begin{lemma}\label{lemSumHi2}

\be\label{eqSumHi2}
\sum\limits_{j=1}^{n}H_{j,2} = (n+1)H_{n,2} - H_{n,1} \sim \frac{\pi^{2}}{6}n.
\ee
\end{lemma}

\begin{proof}

$$
\begin{array}{l}
\sum\limits_{j=1}^{n}H_{j,2} 
= \sum\limits_{j=1}^{n}\sum\limits_{k=1}^{j}\frac{1}{k^{2}}  
= \sum\limits_{k=1}^{n}\frac{1}{k^{2}}\sum\limits_{j=k}^{n}1 
= \sum\limits_{k=1}^{n}\frac{1}{k^{2}}\left(n-(k-1)\right)
\\= (n+1)\sum\limits_{k=1}^{n}\frac{1}{k^{2}} - \sum\limits_{k=1}^{n}\frac{1}{k}
= (n+1)H_{n,2}-H_{n,1}.
\end{array}
$$
\end{proof}

\begin{lemma}\label{lemSumHi2i}

\be\label{eqSumHi2i}
\sum\limits_{j=1}^{n}\frac{H_{j,2}}{j} = 
H_{n,1}H_{n,2} + H_{n,3}-  \sum\limits_{j=1}^{n}\frac{H_{j,1}}{j^{2}} \sim \frac{\pi^{2}}{6} \ln n .
\ee
\end{lemma}

\begin{proof}

$$
\begin{array}{l}
\sum\limits_{j=1}^{n}\frac{H_{j,2}}{j} 
= \sum\limits_{j=1}^{n}\frac{1}{j}\sum\limits_{k=1}^{j}\frac{1}{k^{2}} 
= \sum\limits_{k=1}^{n}\frac{1}{k^{2}} \sum\limits_{j=k}^{n}\frac{1}{j}
= H_{n,1}H_{n,2} - \sum\limits_{k=1}^{n}\frac{H_{k-1,1}}{k^{2}} 
\\ = H_{n,1}H_{n,2} + H_{n,3}-  \sum\limits_{k=1}^{n}\frac{H_{k,1}}{k^{2}}. 
\end{array}
$$
\end{proof}

\begin{lemma}\label{lemSumHi2i1}

\be\label{eqSumHi2i1}
\sum\limits_{j=1}^{n}\frac{H_{j,2}}{j+1} = 
H_{n,1}H_{n,2} - \sum\limits_{j=1}^{n}\frac{H_{j,1}}{j^{2}}+\frac{H_{n,2}}{n+1} \sim \frac{\pi^{2}}{6} \ln n .
\ee
\end{lemma}

\begin{proof}

$$
\begin{array}{l} \sum\limits_{j=1}^{n}\frac{H_{j,2}}{j+1} 
=\sum\limits_{j=1}^{n}\frac{H_{j+1,2}}{j+1} - H_{n+1,3}+1
=\sum\limits_{j=1}^{n+1}\frac{H_{j,2}}{j}-1 - H_{n+1,3}+1
\\ \stackrel{Lemma~\ref{lemSumHi2i}}{=} H_{n+1,1}H_{n+1,2} + H_{n+1,3}-  \sum\limits_{j=1}^{n+1}\frac{H_{j,1}}{j^{2}} -1 - H_{n+1,3}+1
\\ = H_{n+1,1}H_{n+1,2} -  \sum\limits_{j=1}^{n+1}\frac{H_{j,1}}{j^{2}}.
\end{array}
$$
\end{proof}

\begin{lemma}\label{lemSumHj2jmapp}
For $m>1$ we have the following recursive representation 

\be\label{eqSumHj2jmapp}
\begin{array}{rcl}
G_{n,1,0}^{m,0,2}=\sum\limits_{j=1}^{n}\frac{H_{j,2}}{j+m}& = &
G_{n,1,0}^{m-1,0,2}
-\frac{1}{m(m-1)^{2}}-\frac{H_{n+1,2}}{m-1} + \frac{H_{m,1}}{(m-1)^{2}}
+\frac{H_{n+1,2}}{n+m} 
\\&&-\frac{1}{(m-1)^{2}}\sum\limits_{j=n+2}^{n+m}\frac{1}{j}, 
\end{array}
\ee
with initial condition $m=1$ 

\be\label{eqSumHj2j1app}
G_{n,1,0}^{1,0,2}=\sum\limits_{j=1}^{n}\frac{H_{j,2}}{j+1} 
= H_{n+1,1}H_{n+1,2}-\sum\limits_{j=1}^{n+1}\frac{H_{j,1}}{j^{2}}.
\ee

\end{lemma}
\begin{proof}

The special case $m=1$ is calculated in Lemma \ref{lemSumHi2i1}, we just notice that it also can be written
as  

$$G_{n,1,0}^{1,0,2}=H_{n,1}H_{n,2}-G_{n,2,0}^{0,0,1}-\frac{H_{n,2}}{n+1}.$$
In the case $m>1$ we have.

$$
\begin{array}{l}
\sum\limits_{j=1}^{n}\frac{H_{j,2}}{j+m} 
= \sum\limits_{j=1}^{n}\frac{H_{j+1,2}}{j+m} -\sum\limits_{j=1}^{n}\frac{1}{(j+1)^{2}(j+m)}
= \sum\limits_{j=2}^{n+1}\frac{H_{j,2}}{j+m-1} -\sum\limits_{j=1}^{n}\frac{1}{(j+1)^{2}(j+m)}
\\ \stackrel{Lemma~ \ref{lemSumj1sqjmapp}}{=}
\sum\limits_{j=1}^{n}\frac{H_{j,2}}{j+m-1}-\frac{1}{m}+\frac{H_{n+1,2}}{n+m}
-\left(
\frac{H_{n+1,2}}{m-1} - \frac{1}{m-1}
- \frac{H_{m,1}}{(m-1)^{2}}+ \frac{1}{(m-1)^{2}}+\frac{1}{(m-1)^{2}}\sum\limits_{j=n+2}^{n+m}\frac{1}{j}
\right)
\\
= G_{n,1,0}^{m-1,0,2}
-\frac{1}{m}+ \frac{1}{m-1} - \frac{1}{(m-1)^{2}}
+\frac{H_{n+1,2}}{n+m}-\frac{H_{n+1,2}}{m-1} + \frac{H_{m,1}}{(m-1)^{2}}
-\frac{1}{(m-1)^{2}}\sum\limits_{j=n+2}^{n+m}\frac{1}{j}
\\ =
G_{n,1,0}^{m-1,0,2}
-\frac{1}{m(m-1)^{2}}-\frac{H_{n+1,2}}{m-1} + \frac{H_{m,1}}{(m-1)^{2}}
+\frac{H_{n+1,2}}{n+m} -\frac{1}{(m-1)^{2}}\sum\limits_{j=n+2}^{n+m}\frac{1}{j}.
\end{array}
$$
\end{proof}

\newpage
\begin{lemma}\label{lemSumHi2i2}

\be\label{eqSumHi2i2}
G_{n,1,0}^{2,0,2}=\sum\limits_{j=1}^{n}\frac{H_{j,2}}{j+2} = 
H_{n,1}H_{n,2}-H_{n,2} -  \sum\limits_{j=1}^{n}\frac{H_{j,1}}{j^{2}}+ 1   + \frac{H_{n,2}}{n+1}+\frac{H_{n,2}}{n+2}- \frac{1}{n+1}
\sim \frac{\pi^{2}}{6} \ln n .
\ee
\end{lemma}

\begin{proof}

$$
\begin{array}{l}
\sum\limits_{j=1}^{n}\frac{H_{j,2}}{j+2} 
= \sum\limits_{j=1}^{n}\frac{H_{j+2,2}}{j+2}  - \sum\limits_{j=1}^{n}\frac{1}{(j+2)(j+1)^{2}} - H_{n+2,3}+1+\frac{1}{8}
\\ \stackrel{Lemma~\ref{lemSumi12i2}}{=}
\sum\limits_{j=1}^{n+2}\frac{H_{j,2}}{j}  
-1-\frac{5}{8} 
- H_{n+1,2} + \frac{3}{2} - \frac{1}{n+2}
- H_{n+2,3}+1+\frac{1}{8} 
\\ \stackrel{Lemma~\ref{lemSumHi2i}}{=}
H_{n+2,1}H_{n+2,2} + H_{n+2,3}-  \sum\limits_{j=1}^{n+2}\frac{H_{j,1}}{j^{2}}
+ 1 - H_{n+1,2}  - \frac{1}{n+2} - H_{n+2,3}
\\ =
(H_{n+2,1}-1)H_{n+2,2} -  \sum\limits_{j=1}^{n+2}\frac{H_{j,1}}{j^{2}} + 1   - \frac{1}{n+2}+ \frac{1}{(n+2)^{2}}.
\end{array}
$$
\end{proof}

\begin{lemma}\label{lemSumHj2j3}

\be\label{eqSumHj2j3}
\begin{array}{rcl}
G_{n,1,0}^{3,0,2}=\sum\limits_{j=1}^{n}\frac{H_{j,2}}{j+3} & = &
H_{n,1}H_{n,2}-\frac{3}{2}H_{n,2} -\sum\limits_{j=1}^{n}\frac{H_{j,1}}{j^{2}} +\frac{11}{8} +\frac{H_{n,2}}{n+2} 
\\&&+\frac{H_{n,2}}{n+3} +\frac{H_{n,2}}{n+3} -\frac{5}{4(n+1)}-\frac{1}{4(n+2)} 
\sim \frac{\pi^{2}}{6}\ln n .
\end{array}
\ee
\end{lemma}

\begin{proof}
By Lemma \ref{lemSumHi2i2} we have

$$
G_{n,1,0}^{2,0,2}=\left(H_{n+2,1}-1 \right)H_{n+2,2}-\sum\limits_{j=1}^{n+2}\frac{H_{j,1}}{j^{2}}+1-\frac{1}{n+2}+\frac{1}{(n+2)^{2}}
$$
and then by Lemma \ref{lemSumHj2jmapp}

$$
\begin{array}{l}
\sum\limits_{j=1}^{n}\frac{H_{j,2}}{j+3} =
G_{n,1,0}^{2,0,2}
-\frac{1}{12}-\frac{1}{2}H_{n+1,2} + \frac{11}{44}
+\frac{H_{n+1,2}}{n+3} -\frac{1}{4}\sum\limits_{j=n+2}^{n+3}\frac{1}{j}
\\ =
\left(H_{n+2,1}-1 \right)H_{n+2,2}-\frac{1}{2}H_{n+2,2}
-\sum\limits_{j=1}^{n+2}\frac{H_{j,1}}{j^{2}}
+1+ \frac{11}{24}-\frac{1}{12}
+\frac{H_{n+1,2}}{n+3}-\frac{1}{n+2}
\\ +\frac{1}{(n+2)^{2}} + \frac{1}{2(n+2)^{2}}
-\frac{1}{4(n+2)}-\frac{1}{4(n+3)}
\\ =
\left(H_{n+2,1}-\frac{3}{2} \right)H_{n+2,2}
-\sum\limits_{j=1}^{n+2}\frac{H_{j,1}}{j^{2}}
+\frac{11}{8}
+\frac{H_{n+1,2}}{n+3}
-\frac{5}{4(n+2)}-\frac{1}{4(n+3)}
+\frac{3}{2(n+2)^{2}} 
.
\end{array}
$$
\end{proof}

\begin{lemma}\label{lemSumHj2j4}

\be\label{eqSumHj2j4}
\begin{array}{rcl}
G_{n,1,0}^{4,0,2} &=&\sum\limits_{j=1}^{n}\frac{H_{j,2}}{j+4} =
H_{n,1}H_{n,2}-\frac{11}{6}H_{n,2}-\sum\limits_{j=1}^{n}\frac{H_{j,1}}{j^{2}} +\frac{341}{216} \\&&+\frac{H_{n,2}}{n+1} +\frac{H_{n,2}}{n+2}+\frac{H_{n,2}}{n+3}+\frac{H_{n,2}}{n+4} -\frac{49}{36(n+1)}-\frac{13}{36(n+2)}-\frac{1}{9(n+3)}
\\&\sim&\frac{\pi^{2}}{6}\ln n .
\end{array}
\ee
\end{lemma}

\begin{proof}
Lemma \ref{lemSumHj2jmapp} gives us 

$$
\begin{array}{l}
G_{n,1,0}^{4,0,2} = 
G_{n,1,0}^{3,0,2}-\frac{1}{36}-\frac{1}{3}H_{n+1,2} + \frac{25}{108}
+\frac{H_{n+1,2}}{n+4} -\frac{1}{9(n+2)}-\frac{1}{9(n+3)}-\frac{1}{9(n+4)}
\end{array}
$$
and then by Lemma \ref{lemSumHj2j3} we obtain

$$
\begin{array}{l}
G_{n,1,0}^{4,0,2} = 
\left(H_{n+2,1}-\frac{11}{6} \right)H_{n+2,2}
-\sum\limits_{j=1}^{n+2}\frac{H_{j,1}}{j^{2}}
+\frac{341}{216}
+\frac{H_{n+1,2}}{n+3}+\frac{H_{n+1,2}}{n+4} 
\\ -\frac{49}{36(n+2)}-\frac{13}{36(n+3)}-\frac{1}{9(n+4)}
+\frac{11}{6(n+2)^{2}}.
\end{array}
$$
\end{proof}

\newpage
\begin{lemma}[\citepos{ASof2013} Eq. $1.2$]\label{lemSumHi2ii1}
\be\label{eqSumHi2ii1}
G_{n,1,1}^{0,1,2}=\sum\limits_{j=1}^{n}\frac{H_{j,2}}{j(j+1)} = H_{n,3} - \frac{H_{n,2}}{n+1}
\xrightarrow{n\to \infty} \zeta(3)
\ee
\end{lemma}

\begin{proof}

$$
\begin{array}{l}
\sum\limits_{j=1}^{n}\frac{H_{j,2}}{j(j+1)} 
=  \sum\limits_{j=1}^{n}\frac{H_{j,2}}{j} -  \sum\limits_{j=1}^{n}\frac{H_{j,2}}{j+1}
\\ \stackrel{Lemmata~\ref{lemSumHi2i},\ref{lemSumHi2i1}}{=}
H_{n,1}H_{n,2} + H_{n,3}-  \sum\limits_{j=1}^{n}\frac{H_{j,1}}{j^{2}}
-H_{n+1,1}H_{n+1,2} + \sum\limits_{j=1}^{n+1}\frac{H_{j,1}}{j^{2}}
\\ =
H_{n,1}H_{n,2} -\left(H_{n,1}+\frac{1}{n+1}\right)\left(H_{n,2}+\frac{1}{(n+1)^{2}}\right) 
+ H_{n,3} + \frac{H_{n+1,1}}{(n+1)^{2}}
\\ =
H_{n,1}H_{n,2} - H_{n,1}H_{n,2}
-\frac{H_{n,1}}{(n+1)^{2}} -\frac{H_{n,2}}{n+1}
-\frac{1}{(n+1)^{3}}
+ H_{n,3} + \frac{H_{n+1,1}}{(n+1)^{2}}
\\ =
-\frac{H_{n,2}}{n+1}
-\frac{1}{(n+1)^{3}}
+ H_{n,3} 
+\frac{1}{(n+1)^{3}}
=
H_{n,3} -\frac{H_{n,2}}{n+1} \xrightarrow{n\to \infty} \zeta(3) \approx 1.202.
\end{array}
$$
\end{proof}

\begin{lemma}[\citepos{ASof2013} Lemma $1.1$]\label{lemSumHi2ii2}

\be\label{eqSumHi2ii2}
\begin{array}{rcl}
G_{n,1,1}^{0,2,2}& = &\sum\limits_{j=1}^{n}\frac{H_{j,2}}{j(j+2)}  =
\frac{1}{2}\left(H_{n,2} + H_{n,3}    - 1  -\frac{H_{n,2}}{n+1} -\frac{H_{n,2}}{n+2} + \frac{1}{n+1}  \right)
\\ &\xrightarrow{n\to \infty} & \frac{\zeta(3)}{2}+\frac{\pi^{2}}{12}-\frac{1}{2}.
\end{array}
\ee
\end{lemma}

\begin{proof}

$$
\begin{array}{l}
\sum\limits_{j=1}^{n}\frac{H_{j,2}}{j(j+2)}
=  \frac{1}{2}\left( \sum\limits_{j=1}^{n}\frac{H_{j,2}}{j}-\sum\limits_{j=1}^{n}\frac{H_{j,2}}{j+2} \right)
\\ \stackrel{Lemmata~\ref{lemSumHi2i},\ref{lemSumHi2i2}}{=}
\frac{1}{2}\left( 
H_{n,1}H_{n,2} + H_{n,3}-  \sum\limits_{j=1}^{n}\frac{H_{j,1}}{j^{2}}
-(H_{n+2,1}-1)H_{n+2,2} 
\right. \\ \left.
+  \sum\limits_{j=1}^{n+2}\frac{H_{j,1}}{j^{2}}
- 1   + \frac{1}{n+2}- \frac{1}{(n+2)^{2}}
\right)
\\ =
\frac{1}{2}\left( 
H_{n,1}H_{n,2} + H_{n,3}
-H_{n+2,1}H_{n+2,2} + H_{n+2,2} 
+  \frac{H_{n+1,1}}{(n+1)^{2}} +  \frac{H_{n+2,1}}{(n+2)^{2}}
\right. \\ \left.
- 1   + \frac{1}{n+2}- \frac{1}{(n+2)^{2}}
\right)
\end{array}
$$

$$
\begin{array}{l}
=
\frac{1}{2}\left( 
H_{n,1}H_{n,2} + H_{n,3}
-\left(H_{n,1}+\frac{1}{n+1}+\frac{1}{n+2}\right)\left(H_{n,2}+\frac{1}{(n+1)^{2}}+\frac{1}{(n+2)^{2}}\right) 
\right. \\ \left.
+ H_{n+1,2} 

+  \frac{H_{n+1,1}}{(n+1)^{2}} +  \frac{H_{n+2,1}}{(n+2)^{2}}
- 1   + \frac{1}{n+2}
\right)
\\ =
\frac{1}{2}\left( 
H_{n,1}H_{n,2} + H_{n,3}
-H_{n,1}H_{n,2}
-\frac{H_{n,1}}{(n+1)^{2}}-\frac{H_{n,1}}{(n+2)^{2}}
-\frac{H_{n,2}}{n+1}-\frac{H_{n,2}}{n+2}
\right. \\ \left.
-\frac{1}{(n+1)^{3}}
-\frac{1}{(n+1)(n+2)^{2}}
-\frac{1}{(n+2)(n+1)^{2}}
-\frac{1}{(n+2)^{3}}
+ H_{n+1,2} 
 +  \frac{H_{n+1,1}}{(n+1)^{2}} +  \frac{H_{n+2,1}}{(n+2)^{2}}
\right. \\ \left.
- 1   + \frac{1}{n+2}
\right)
\\ =
\frac{1}{2}\left( 
H_{n,3}
-\frac{H_{n+2,2}}{n+1}
-\frac{H_{n+1,2}}{n+2}
-\frac{1}{(n+2)^{3}}
+ H_{n+1,2}  
+  \frac{1}{(n+1)^{3}} 
+  \frac{1}{(n+1)(n+2)^{2}}
\right. \\ \left.
+  \frac{1}{(n+2)^{3}}
- 1   + \frac{1}{n+2}
\right)
\\ =
\frac{1}{2}\left( 
H_{n+1,3} + H_{n+1,2}  - 1 
-\frac{H_{n+1,2}}{n+1}
-\frac{H_{n+1,2}}{n+2}
+ \frac{1}{n+2} 
\right)
\\ \xrightarrow{n\to \infty} 
\frac{\zeta(3)}{2}+\frac{\pi^{2}}{12}-\frac{1}{2} \approx 0.923.
\end{array}
$$
\end{proof}

\begin{lemma}[\citepos{ASof2013} Eq. $1.4$]\label{lemSumHi2i1i2}

\be\label{eqSumHi2i1i2}
G_{n,1,1}^{1,2,2}=\sum\limits_{j=1}^{n}\frac{H_{j,2}}{(j+1)(j+2)} 
=H_{n,2}-1-\frac{H_{n,2}}{n+2}+\frac{1}{n+1}
\xrightarrow{n\to \infty} \frac{\pi^{2}}{6}-1
\ee
\end{lemma}

\begin{proof}

$$
\begin{array}{l}
\sum\limits_{j=1}^{n}\frac{H_{j,2}}{(j+1)(j+2)} 
= \sum\limits_{j=1}^{n}\frac{H_{j,2}}{j+1} - \sum\limits_{j=1}^{n}\frac{H_{j,2}}{j+2} 
\\ \stackrel{Lemmata~\ref{lemSumHi2i1},\ref{lemSumHi2i2}}{=}
H_{n+1,1}H_{n+1,2} - \sum\limits_{j=1}^{n+1}\frac{H_{j,1}}{j^{2}}
-
(H_{n+2,1}-1)H_{n+2,2}
\\ +  \sum\limits_{j=1}^{n+2}\frac{H_{j,1}}{j^{2}}
- 1   + \frac{1}{n+2}- \frac{1}{(n+2)^{2}}
\\ =
H_{n+1,1}H_{n+1,2} -\left(H_{n+1,1}+\frac{1}{n+2}\right)\left(H_{n+1,2}+\frac{1}{(n+2)^{2}}\right)+H_{n+2,2} 
\\ +  \frac{H_{n+2,1}}{n+2^{2}}
-  \frac{n+1}{n+2} - \frac{1}{(n+2)^{2}}
\\=
H_{n+1,1}H_{n+1,2} 
-H_{n+1,1}H_{n+1,2}
-\frac{H_{n+1,1}}{(n+2)^{2}}
-\frac{H_{n+1,2}}{n+2}
-\frac{1}{(n+2)^{3}}
+H_{n+1,2}\\  +  \frac{H_{n+2,1}}{n+2^{2}} -  \frac{n+1}{n+2} 
\\=
-\frac{H_{n+1,2}}{n+2}
+H_{n+1,2}    -  \frac{n+1}{n+2} 
=\frac{n+1}{n+2}\left(H_{n+1,2}-1\right)
\xrightarrow{n\to \infty} \frac{\pi^{2}}{6}-1 \approx 0.645.
\end{array}
$$
\end{proof}

\begin{lemma}\label{lemSumHj2j1j3}

\be\label{eqSumHj2j1j3}
\begin{array}{rcl}
G_{n,1,1}^{1,3,2} & = & \sum\limits_{j=1}^{n}\frac{H_{j,2}}{(j+1)(j+3)} = 
\frac{3}{4}H_{n,2}-\frac{11}{16}-\frac{H_{n,2}}{2(n+2)}-\frac{H_{n,2}}{2(n+3)} \\&&+\frac{5}{8(n+1)}+\frac{1}{8(n+2)}
\xlongrightarrow{n \to \infty}\frac{\pi^{2}}{8}-\frac{11}{16}.
\end{array}
\ee
\end{lemma}
\begin{proof}

$$
\begin{array}{l}
\sum\limits_{j=1}^{n}\frac{H_{j,2}}{(j+1)(j+3)} =
\frac{1}{2}\left(
\sum\limits_{j=1}^{n}\frac{H_{j,2}}{j+1} 
-
\sum\limits_{j=1}^{n}\frac{H_{j,2}}{j+3} 
\right)
\stackrel{Lemmata \ref{lemSumHi2i1},\ref{lemSumHj2j3}}{=}
\\ \frac{1}{2}\left(
H_{n+1,1}H_{n+1,2}-\sum\limits_{j=1}^{n+1}\frac{H_{j,1}}{j^{2}}
-\left(H_{n+2,1}-\frac{3}{2} \right)H_{n+2,2}
+\sum\limits_{j=1}^{n+2}\frac{H_{j,1}}{j^{2}}
-\frac{11}{8}
\right. \\ \left.
-\frac{H_{n+1,2}}{n+3}
+\frac{5}{4(n+2)}+\frac{1}{4(n+3)}
-\frac{3}{2(n+2)^{2}}
\right)
=
\frac{1}{2}\left(
H_{n+1,1}H_{n+1,2}
-H_{n+1,1}H_{n+1,2}
\right. \\ \left.
-\frac{H_{n+2,1}}{(n+2)^{2}}
-\frac{H_{n+1,2}}{n+2}
+\frac{3}{2}H_{n+2,2}
+\frac{H_{n+2,1}}{(n+2)^{2}}
-\frac{11}{8}
%\right. \\ \left.
-\frac{H_{n+1,2}}{n+3}
+\frac{5}{4(n+2)}+\frac{1}{4(n+3)}
-\frac{3}{2(n+2)^{2}}
\right)
\\ = 
\frac{1}{2}\left(
\frac{3}{2}H_{n+2,2}
-\frac{11}{8}
-\frac{H_{n+1,2}}{n+2}
-\frac{H_{n+1,2}}{n+3}
+\frac{5}{4(n+2)}+\frac{1}{4(n+3)}
-\frac{3}{2(n+2)^{2}}
\right)
\\ \xlongrightarrow{n \to \infty}\frac{\pi^{2}}{8}-\frac{11}{16}
\approx 0.546
.
\end{array}
$$
\end{proof}

\begin{lemma}\label{lemSumHj2j1j4}

\be\label{eqSumHj2j1j4}
\begin{array}{rcl}
G_{n,1,1}^{1,4,2} &=&\sum\limits_{j=1}^{n}\frac{H_{j,2}}{(j+1)(j+4)}  =  
\frac{11}{18}H_{n,2}-\frac{341}{648}-\frac{H_{n,2}}{3(n+2)}-\frac{H_{n,2}}{3(n+3)}-\frac{H_{n,2}}{3(n+4)} \\ &&+\frac{49}{108(n+1)}+\frac{13}{108(n+2)}+\frac{1}{27(n+3)}
 \xlongrightarrow{n \to \infty}\frac{11\pi^{2}}{108}-\frac{341}{648}
.
\end{array}
\ee
\end{lemma}

\begin{proof}

$$
\begin{array}{l}
\sum\limits_{j=1}^{n}\frac{H_{j,2}}{(j+1)(j+4)} =
\frac{1}{3}\left(
\sum\limits_{j=1}^{n}\frac{H_{j,2}}{j+1} 
-
\sum\limits_{j=1}^{n}\frac{H_{j,2}}{j+4} 
\right)
\\ \stackrel{Lemmata~ \ref{lemSumHi2i1},\ref{lemSumHj2j4}}{=}
\\ \frac{1}{3}\left(
H_{n+1,1}H_{n+1,2}-\sum\limits_{j=1}^{n+1}\frac{H_{j,1}}{j^{2}}
-
\left(H_{n+2,1}-\frac{11}{6} \right)H_{n+2,2}
+\sum\limits_{j=1}^{n+2}\frac{H_{j,1}}{j^{2}}
-\frac{341}{216}
\right. \\ \left.
-\frac{H_{n+1,2}}{n+3}-\frac{H_{n+1,2}}{n+4} 
+\frac{49}{36(n+2)}+\frac{13}{36(n+3)}+\frac{1}{9(n+4)}
-\frac{11}{6(n+2)^{2}}
\right)
\\=
\frac{1}{3}\left(
H_{n+1,1}H_{n+1,2}
-H_{n+2,1}H_{n+2,2}+\frac{11}{6}H_{n+2,2}
+\frac{H_{n+2,1}}{(n+2)^{2}}
-\frac{341}{216}
-\frac{H_{n+1,2}}{n+3}-\frac{H_{n+1,2}}{n+4} 
\right. \\ \left.
+\frac{49}{36(n+2)}+\frac{13}{36(n+3)}+\frac{1}{9(n+4)}
-\frac{11}{6(n+2)^{2}}
\right)
\\=
\frac{1}{3}\left(
H_{n+1,1}H_{n+1,2}
-H_{n+1,1}H_{n+1,2}
-\frac{H_{n+1,2}}{n+2}
-\frac{H_{n+2,1}}{(n+2)^{2}}
+\frac{11}{6}H_{n+2,2}
-\frac{341}{216}
\right. \\ \left.
-\frac{H_{n+1,2}}{n+3}-\frac{H_{n+1,2}}{n+4} 
+\frac{H_{n+2,1}}{(n+2)^{2}}
+\frac{49}{36(n+2)}+\frac{13}{36(n+3)}+\frac{1}{9(n+4)}
-\frac{11}{6(n+2)^{2}}
\right)
\\=
\frac{1}{3}\left(
\frac{11}{6}H_{n+2,2}
-\frac{341}{216}
-\frac{H_{n+1,2}}{n+2}
-\frac{H_{n+1,2}}{n+3}-\frac{H_{n+1,2}}{n+4} 
+\frac{49}{36(n+2)}+\frac{13}{36(n+3)}+\frac{1}{9(n+4)}
-\frac{11}{6(n+2)^{2}}
\right)
\\
 \xlongrightarrow{n \to \infty}\frac{11\pi^{2}}{108}-\frac{341}{648} \approx 0.479
.
\end{array}
$$
\end{proof}

\begin{lemma}\label{lemSumHj2j2j4}

\be\label{eqSumHj2j2j4}
\begin{array}{rcl}
G_{n,1,1}^{2,4,2} &=&\sum\limits_{j=1}^{n}\frac{H_{j,2}}{(j+2)(j+4)} = 
\frac{5}{12}H_{n,2} -\frac{125}{432}-\frac{H_{n,2}}{2(n+3)}-\frac{H_{n,2}}{2(n+4)} 
\\&&+\frac{13}{72(n+1)}+\frac{13}{72(n+2)}+\frac{1}{18(n+3)} 
 \xlongrightarrow{n \to \infty}\frac{5\pi^{2}}{72}-\frac{125}{432}
.
\end{array}
\ee
\end{lemma}

\begin{proof}

$$
\begin{array}{l}
\sum\limits_{j=1}^{n}\frac{H_{j,2}}{(j+2)(j+4)} 
=\frac{1}{2}\left(
\sum\limits_{j=1}^{n}\frac{H_{j,2}}{j+2} 
-
\sum\limits_{j=1}^{n}\frac{H_{j,2}}{j+4} 
\right)
\\ \stackrel{Lemmata \ref{lemSumHi2i2},\ref{lemSumHj2j4}}{=}
\frac{1}{2}\left(
\left(H_{n+2,1}-1\right)H_{n+2,2}-\sum\limits_{j=1}^{n+2}\frac{H_{j,1}}{j^{2}}+1-\frac{1}{n+2}+\frac{1}{(n+2)^{2}}
\right. \\ \left.
-\left(H_{n+2,1}-\frac{11}{6} \right)H_{n+2,2}
+\sum\limits_{j=1}^{n+2}\frac{H_{j,1}}{j^{2}}
-\frac{341}{216}
-\frac{H_{n+1,2}}{n+3}-\frac{H_{n+1,2}}{n+4} 
+\frac{49}{36(n+2)}
\right. \\ \left.
+\frac{13}{36(n+3)}+\frac{1}{9(n+4)}
-\frac{11}{6(n+2)^{2}}
\right)
\end{array}
$$

$$
\begin{array}{l}
=
\frac{1}{2}\left(
H_{n+2,1}H_{n+2,2}+\frac{5}{6}H_{n+2,2} 
-\frac{125}{216}
\right. \\ \left.
-H_{n+2,1}H_{n+2,2}
-\frac{H_{n+1,2}}{n+3}-\frac{H_{n+1,2}}{n+4} 
%\right. \\ \left.
+\frac{13}{36(n+2)}+\frac{13}{36(n+3)}+\frac{1}{9(n+4)}
-\frac{5}{6(n+2)^{2}}
\right)
\\=
\frac{1}{2}\left(
\frac{5}{6}H_{n+2,2} 
-\frac{125}{216}
-\frac{H_{n+1,2}}{n+3}-\frac{H_{n+1,2}}{n+4} 
+\frac{13}{36(n+2)}+\frac{13}{36(n+3)}+\frac{1}{9(n+4)}
-\frac{5}{6(n+2)^{2}}
\right)
\\ \xlongrightarrow{n \to \infty}\frac{5\pi^{2}}{72}-\frac{125}{432}
 \approx 0.398 .
\end{array}
$$
\end{proof}

\begin{lemma}\label{lemSumHj2j3j4}

\be\label{eqSumHj2j3j4}
\begin{array}{rcl}
G_{n,1,1}^{3,4,2}&=&\sum\limits_{j=1}^{n}\frac{H_{j,2}}{(j+3)(j+4)} =
\frac{1}{3}H_{n,2}-\frac{11}{54}-\frac{H_{n,2}}{n+4} \\&& +\frac{1}{9(n+1)}+\frac{1}{9(n+2)}+\frac{1}{9(n+3)}
\xlongrightarrow{n \to \infty}\frac{\pi^{2}}{18}-\frac{11}{54}.
\end{array}
\ee
\end{lemma}

\begin{proof}

$$
\begin{array}{l}
\sum\limits_{j=1}^{n}\frac{H_{j,2}}{(j+3)(j+4)} =
\sum\limits_{j=1}^{n}\frac{H_{j,2}}{j+3} - \sum\limits_{j=1}^{n}\frac{H_{j,2}}{j+4} 
\stackrel{Lemmata~\ref{lemSumHj2j3},\ref{lemSumHj2j4}}{=}
\left(H_{n+2,1}-\frac{3}{2} \right)H_{n+2,2}
-\sum\limits_{j=1}^{n+2}\frac{H_{j,1}}{j^{2}}
\\ +\frac{11}{8}
+\frac{H_{n+1,2}}{n+3}
-\frac{5}{4(n+2)}-\frac{1}{4(n+3)}
+\frac{3}{2(n+2)^{2}}
 -\left(\left(H_{n+2,1}-\frac{11}{6} \right)H_{n+2,2}
-\sum\limits_{j=1}^{n+2}\frac{H_{j,1}}{j^{2}}
+\frac{341}{216}
\right. \\ \left.
+\frac{H_{n+1,2}}{n+3}+\frac{H_{n+1,2}}{n+4} 
 -\frac{49}{36(n+2)}-\frac{13}{36(n+3)}-\frac{1}{9(n+4)}
+\frac{11}{6(n+2)^{2}}\right)
\\=
\left(H_{n+2,1}-\frac{3}{2}-H_{n+2,1}+\frac{11}{6} \right)H_{n+2,2}
-\sum\limits_{j=1}^{n+2}\frac{H_{j,1}}{j^{2}}+\sum\limits_{j=1}^{n+2}\frac{H_{j,1}}{j^{2}}
+\frac{11}{8}-\frac{341}{216}
 +\frac{H_{n+1,2}}{n+3}
 \\-\frac{H_{n+1,2}}{n+3}-\frac{H_{n+1,2}}{n+4} 
-\frac{5}{4(n+2)}
+\frac{49}{36(n+2)}-\frac{1}{4(n+3)}
+\frac{13}{36(n+3)}+\frac{1}{9(n+4)}
+\frac{3}{2(n+2)^{2}}-\frac{11}{6(n+2)^{2}}
\\=
\frac{1}{3}H_{n+2,2}-\frac{11}{54}
-\frac{H_{n+1,2}}{n+4} 
+\frac{1}{9(n+2)}+\frac{1}{9(n+3)}+\frac{1}{9(n+4)}
-\frac{1}{3(n+2)^{2}}
 \xlongrightarrow{n \to \infty}\frac{\pi^{2}}{18}-\frac{11}{54}\approx0.345.
\end{array}
$$
\end{proof}

\newpage
\begin{lemma}[See also \citepos{ASof2013} Eq. $1.2$]\label{lemSumHi2p1ii1}

\be\label{eqSumHi2p1ii1}
\sum\limits_{j=1}^{n}\frac{H_{j+1,2}}{j(j+1)}  =  3-H_{n+1,2} - \frac{H_{n+1,2}}{n+1} - \frac{1}{n+1}
\xrightarrow{n\to \infty}  3-\frac{\pi^{2}}{6}
\ee
\end{lemma}

\begin{proof}

$$
\begin{array}{l}
\sum\limits_{j=1}^{n}\frac{H_{j+1,2}}{j(j+1)} 
= \sum\limits_{j=1}^{n}\frac{H_{j+1,2}}{j} - \sum\limits_{j=1}^{n}\frac{H_{j+1,2}}{j+1} 
 = \sum\limits_{j=1}^{n}\frac{H_{j,2}}{j} +  \sum\limits_{j=1}^{n}\frac{1}{j(j+1)^{2}}
- \sum\limits_{j=1}^{n+1}\frac{H_{j,2}}{j} + 1
\\ \stackrel{Lemma~\ref{lemSumii12}}{=}
2-H_{n+1,2} - \frac{1}{n+1}
- \frac{H_{n+1,2}}{n+1} + 1
=3-H_{n+1,2} - \frac{H_{n+1,2}}{n+1} - \frac{1}{n+1}
\\ \xrightarrow{n\to \infty} 3-\frac{\pi^{2}}{6}  \approx 1.355
\end{array}
$$
\end{proof}

\begin{lemma}[See also the Remark following \citepos{ASof2013} Lemma $1.1$]\label{lemSumHi2p1ii2}

\be\label{eqSumHi2p1ii2}
\begin{array}{rcl}
\sum\limits_{j=1}^{n}\frac{H_{j+1,2}}{j(j+2)} 
&=& 
\frac{1}{2}\left(\frac{5}{2}- H_{n,2} +H_{n,3} - \frac{H_{n,2}}{n+1} - \frac{H_{n,2}}{n+2} -\frac{1}{n+2}-\frac{1}{(n+1)^{2}} \right)
\\ &\xrightarrow{n\to \infty}& \frac{5}{4}+\frac{\zeta(3)}{2}-\frac{\pi^{2}}{12}.
\end{array}
\ee
\end{lemma}

\begin{proof}

$$
\begin{array}{l}
\sum\limits_{j=1}^{n}\frac{H_{j+1,2}}{j(j+2)}
= \frac{1}{2}\left(\sum\limits_{j=1}^{n}\frac{H_{j+1,2}}{j} - \sum\limits_{j=1}^{n}\frac{H_{j+1,2}}{j+2}\right)
\\ = \frac{1}{2}\left(\sum\limits_{j=1}^{n}\frac{H_{j,2}}{j} +  \sum\limits_{j=1}^{n}\frac{1}{j(j+1)^{2}}
- \sum\limits_{j=1}^{n}\frac{H_{j+2,2}}{j+2} + \sum\limits_{j=1}^{n}\frac{1}{(j+2)^{3}} \right)
\\ \stackrel{Lemma~\ref{lemSumii12}}{=}
\frac{1}{2}\left(\sum\limits_{j=1}^{n}\frac{H_{j,2}}{j} +  2-H_{n+1,2}-\frac{1}{n+1}
- \sum\limits_{j=1}^{n+2}\frac{H_{j,2}}{j}+1+\frac{5}{8} + H_{n+2,3}-1-\frac{1}{8}\right)
\\ =
\frac{1}{2}\left(2-H_{n+1,2}-\frac{1}{n+1}
- \frac{H_{n+1,2}}{n+1} - \frac{H_{n+2,2}}{n+2}
+\frac{1}{2} + H_{n+2,3}\right)
\\ =
\frac{1}{2}\left(\frac{5}{2}+ H_{n+1,3} -H_{n+1,2} - \frac{H_{n+1,2}}{n+1} - \frac{H_{n+1,2}}{n+2} -\frac{1}{n+1} \right)
\\ \xrightarrow{n\to \infty} \frac{5}{4}+\frac{\zeta(3)}{2}-\frac{\pi^{2}}{12} \approx 1.029.
\end{array}%
$$
\end{proof}

\newpage
\begin{lemma}[See also the Remark following \citepos{ASof2013} Lemma $1.1$]\label{lemSumHi2p2ii2}

\be\label{eqSumHi2p2ii2}
%
$$
and taking $\sum\limits_{k=1}^{n}\frac{H_{k,2}}{k^{2}}$ over to the LHS we obtain

$$
2\sum\limits_{k=1}^{n}\frac{H_{k,2}}{k^{2}} = H_{n,2}^{2} + H_{n,4}, 
$$
and finally

$$
\sum\limits_{k=1}^{n}\frac{H_{k,2}}{k^{2}} = \frac{1}{2}\left( H_{n,2}^{2} + H_{n,4} \right) 
 \xrightarrow{n\to \infty}  \frac{7\pi^{4}}{360} \approx 1.894 .
$$
\end{proof}

\begin{remark}\label{remlemSumHi2isqRef}
A generalized version of Lemma \ref{lemSumHi2isq} is presented as an ``almost obvious'' result by
\citet{CXuMZhaWZhu2016} in their Lemma $2.7$, i.e., for $m_{1}, m_{2}, n \in \mathbb{N}$ it holds

$$
\sum\limits_{j=1}^{n}\frac{H_{j,m_{1}}}{j^{m_{2}}} +
\sum\limits_{j=1}^{n}\frac{H_{j,m_{2}}}{j^{m_{1}}} = H_{n,m_{1}}H_{n,m_{2}} + H_{n,m_{1}+m_{2}}.
$$
Earlier, \citet{PFlaBSal1997}, on p. $23$, provided the limit value,
and \cite{ASof2011} attributed this limit to Euler. 
\end{remark}

\newpage
\begin{lemma}\label{lemSumHi2i1sq}

\be\label{eqSumHi2i1sq}
G_{n,2,0}^{1,0,2}=\sum\limits_{j=1}^{n}\frac{H_{j,2}}{(j+1)^{2}} = 
\frac{1}{2}\left(H_{n+1,2}^{2}-H_{n+1,4}\right)
 \xrightarrow{n\to \infty} \frac{\pi^{4}}{120}.
\ee
\end{lemma}

\begin{proof}

$$

$$
\end{proof}

\begin{lemma}\label{lemSumHj1sq}

\be\label{eqSumHj1sq}
V_{n,0,0}^{0,0,1}=\sum\limits_{j=1}^{n}H_{j,1}^{2} = (n+1)H_{n,1}^{2}-(2n+1)H_{n,1}+2n\sim n\ln^{2}n.
\ee
\end{lemma}

\begin{proof}
We can write 

$$
%
$$
\end{proof}

\newpage
\begin{lemma}\label{lemSumHj1sqj}

\be\label{eqSumHj1sqj}
V_{n,1,0}^{0,0,1}=\sum\limits_{j=1}^{n}\frac{H_{j,1}^{2}}{j} = 
\frac{1}{3}H_{n,1}^{3}+\sum\limits_{i=1}^{n}\frac{H_{j,1}}{j^{2}}-\frac{1}{3}H_{n,3} \sim \frac{1}{3}\ln^{3}n .
\ee
\end{lemma}

\begin{proof}
We can write 

$$
%
$$
Taking over to the other side we obtain

$$
%
$$

\end{proof}

\newpage
\begin{lemma}\label{lemSumHj1sqjmapp}
For $m>1$ we have the following recursive representation 

\be\label{eqSumHj1sqjmapp}
%
\ee
with initial condition $m=1$ 

\be\label{lemSumHj1sqj1app}
%
\ee
\end{lemma}

\begin{proof}
We first consider the initial condition $m=1$ 

$$
%
$$
Then in the case $m>1$ we have

$$
%
\ee
\end{lemma}

\begin{proof}
Using the initial conditions in Lemmata \ref{lemSumHj1jmapp} and \ref{lemSumHj1sqjmapp}

$$
%
\ee
\end{lemma}

\begin{proof}
By Lemma \ref{lemSumHj1sqjmapp}, with $m=3$, we may write

$$
%
\ee
\end{lemma}
\begin{proof}
By Lemma \ref{lemSumHj1sqjmapp}, with $m=4$, we may write 

$$
%
$$
Taking over to the LHS we obtain that

$$
%
$$
We will use Lemma \ref{lemSum1} to simplify

$$
%
$$
We input now the formul\ae\  obtained in Lemmata 
\ref{lemSumHii1}, \ref{lemSumHiii1},  and \ref{lemSumHi2i1i2}
to obtain

$$
%
$$
\end{proof}

\begin{remark}
\citet{ASofMHas2012} considered (their Corollary $2$) a similar infinite sum

$$
\sum\limits_{j=1}^{\infty}\frac{H_{j,1}^{2}}{j(j+1)} = 3\zeta(3),
$$
\citet{CXuMZhaWZhu2016}, in their Eq. $(2.39)$, considered 
$\sum_{j=1}^{\infty}H_{j,1}^{2}/(j(j+k))$, 
\citet{KCheYChe2020} in their Thm. $3.1$ consider the limit for a more general series,
and
\citet{XWanWChu2018}, in their Corollary $3$, present the 
limit of Eq. \eqref{eqSumHisqi1i2app}. 
The closed form of the sum is presented, without proof, on \citepos{KBarSSag2015bart} p. $74$.
\end{remark}

\begin{lemma}\label{lemSumHj1sqj1j3}

\be\label{eqSumHj1sqj1j3}

\ee
\end{lemma}

\begin{proof}
We can write 

$$
%
$$
\end{proof}

\newpage
\begin{lemma}[See also \citepos{DBorJBor1995art} Eq. $(2)$ for the limit]\label{lemSumHisqip1sqapp}

\be\label{eqSumHisqip1sqapp}
%
\ee
\end{lemma}

\begin{proof}
We can write

$$
%
$$
We will not be able to obtain a closed form formula for this sum, but are able to obtain its limiting constant.
Considering each element in turn we have,

\begin{enumerate}
\item 

$$
%
$$
as  

\item $\sum\limits_{j=1}^{n}\frac{1}{(j+1)^{2}(j+2)^{2}}=2H_{n+1,2}+\frac{2}{n+2}+\frac{1}{(n+2)^{2}}-\frac{13}{4} \xrightarrow{n\to \infty} \frac{\pi^{2}}{3}-\frac{13}{4} $ (Lemma \ref{lemSumj1sqj2sq});

\item $\sum\limits_{j=1}^{n}\frac{1}{(j+1)(j+2)^{3}}=\frac{23}{8}-H_{n+2,2}-H_{n+2,3}-\frac{1}{n+2}\xrightarrow{n\to \infty} \frac{23}{8}-\frac{\pi^{2}}{6}-\zeta(3)$ (Lemma \ref{lemSumi1i2cb});
\end{enumerate}
Putting this together we obtain

$$
%
$$
\end{proof}

\begin{lemma}\label{lemSumHj1sqjmsqapp}
For $m>1$ we have the following recursive representation 

\be\label{eqSumHj1sqjmsqapp}
%
\ee
with initial condition, $m=1$, given by Lemma \ref{lemSumHisqip1sqapp}.
\end{lemma}

\begin{proof}
The initial condition has been calculated directly in Lemma \ref{lemSumHisqip1sqapp}
and we also have done the special case $m=2$ directly in Lemma \ref{lemSumHisqip2sq}.
Then, in the case $m>1$, we have 

$$
%
$$
\end{proof}

\begin{lemma}\label{lemSumHi2sqimapp}
For $m>1$ we have the following recursive representation 

\be\label{eqSumHi2sqimapp}
%
\ee
with initial condition, $m=1$, given by Lemma \ref{lemSumHi2sqi1app}.
\end{lemma}

\begin{proof}
The initial condition has been calculated directly in Lemma \ref{lemSumHi2sqi1app}
and we also have done the special case $m=2$ directly in Lemma \ref{lemSumHi2sqi2}.
Then, in the case $m>1$, we have 

$$
%
$$
\end{proof}

\begin{lemma}[see also \citepos{KCheYChe2020} Thm. $3.1$ for the limit]\label{lemSumHi2sqi1i2app}

\be\label{eqSumHi2sqi1i2app}
%
\ee
\end{lemma}

\begin{proof}
We will use Abel's summation by parts, \citep[see, e.g,][]{WChu2007} for two sequences $A_{j}$ 
and $B_{j}$

\be\label{eqAbelSumPartsapp}
\sum\limits_{j=1}^{\infty}B_{j}\bigtriangledown A_{j+1} = B_{n}A_{n+1} -B_{1}A_{1}+\sum\limits_{j=2}^{n}A_{j}\bigtriangledown B_{j},
\ee
where $\bigtriangledown \tau_{j}:=\tau_{j}-\tau_{j-1}$. 
We take $A_{j}:=H_{j,2}$ and
$B_{j}:=H_{j,2}^{2}$. Then, $\bigtriangledown A_{j} = 1/j^{2}$ and 

\be\label{eqtridownBjapp}
%
\ee
With the above we apply Eq. \eqref{eqAbelSumPartsapp}

$$
%
$$
Taking $2\sum\limits_{j=1}^{n}\frac{H^{2}_{j,2}}{j^{2}}$ to the other side we obtain

$$
%
$$
Taking the limit $n\to \infty$ in the above we obtain

$$
%
\ee
\end{lemma}

\begin{proof}
Notice, that we have already considered this sum in the proof of Lemma \ref{lemSumHi2sqisqapp}.
Here, we present its derivation for completeness.

$$
%
$$
\end{proof}

\newpage
\begin{lemma}\label{lemSumHi2sqipmsq}
For $m>1$ we have the following recursive representation 

\be\label{eqSumHi2sqipmsq}
%
\ee
with initial condition, $m=1$, given by Lemma \ref{lemSumHi2sqip1sq}.
\end{lemma}
\begin{proof}
The initial condition has been calculated directly in Lemma \ref{lemSumHi2sqip1sq}
and we also have done the special case $m=2$ directly in Lemma \ref{lemSumHi2sqip2sq}.
Then, in the case $m>1$, we have 

$$
%
$$
Taking over $\sum\limits_{k=1}^{n}\frac{H_{k,1}H_{k,2}}{k}$ to the LHS we obtain

$$
%
$$

\end{proof}

\section*{Check with \proglang{Mathematica}}

Many of the presented closed form sums 
can be obtained directly using 
\proglang{Mathematica 12.0.0} 
(Kernel for Linux x86 (64-bit) running on a openSUSE Leap 15.6 box, \cite{Mathematica}).
Some of our formul\ae\ are presented in terms of harmonic sums, only whose
limits are known (and presented at the start of our work).
However, it is interesting that we managed to provide closed form formul\ae ,
in terms of harmonic numbers, 
for a number of sums,
Lemmata  
\ref{lemSumHi2i1sqi2app}--\ref{lemSumHj1sq}, 
\ref{lemSumHisqi1i2app}--\ref{lemSumHj1sqj3j4} 
\ref{lemSumHj1sqj1j2j3}, 
\ref{lemSumHj1sqj1j2j4},
and \ref{lemSumHi2sqi1i2app}, 
which we could not get \proglang{Mathematica 12.0.0} to handle 
and some we could not locate in the literature. 
Furthermore, we were not able to make \proglang{Mathematica 12.0.0} return
the limit in Lemmata  
\ref{lemSumHi2i1sqi2app}, 
\ref{lemSumHi2sqi1i2app}--\ref{lemSumHi2sqip2sq}, 
and \ref{lemSumHi2sqip1sqip2sq}. 
The script with our \proglang{Mathematica} code is available at
\url{https://github.com/krzbar/HarmonicSums}.